\documentclass{book}
\usepackage{amsthm}
\newtheorem{theorem}{Theorem}[section]

\newtheorem{proposition}{Proposition}[section]
\newtheorem{lemma}{Lemma}[section]
\theoremstyle{definition}
\newtheorem{definition}{Definition}[section]

\newcommand{\preface}{{\noindent \bfseries \Large Preface}\vskip2cm}
\newcommand{\abstract}[1]{{\noindent \bfseries Abstract.} #1}

\usepackage{amsmath, amssymb, amscd}

\usepackage{makeidx}         
\usepackage{graphicx}        
\usepackage{multicol}        
\usepackage[bottom]{footmisc}

\newcommand{\Hc}{{\mathcal H}}

\newcommand{\NN}{{\mathbb N}}
\newcommand{\RR}{{\mathbb R}}

\usepackage{bm}
\DeclareMathAlphabet{\mathbfsf}{\encodingdefault}{\sfdefault}{bx}{n}

\DeclareBoldMathCommand\Fb{F}
\DeclareBoldMathCommand\Mb{M}
\DeclareBoldMathCommand\Nb{N}
\DeclareBoldMathCommand\Pb{P}
\DeclareBoldMathCommand\Ob{O}
\DeclareBoldMathCommand\rb{R}
\DeclareBoldMathCommand\ab{a}
\DeclareBoldMathCommand\bb{b}
\DeclareBoldMathCommand\cb{c}
\DeclareBoldMathCommand\eb{e}
\DeclareBoldMathCommand\ib{i}
\DeclareBoldMathCommand\jb{j}
\DeclareBoldMathCommand\kb{k}
\DeclareBoldMathCommand\pb{p}
\DeclareBoldMathCommand\rb{r}
\DeclareBoldMathCommand\ub{u}
\DeclareBoldMathCommand\vb{v}
\DeclareBoldMathCommand\xb{x}

\makeatletter
\newcommand*\rel@kern[1]{\kern#1\dimexpr\macc@kerna}
\newcommand*\widebar[1]{%
  \begingroup
  \def\mathaccent##1##2{%
    \rel@kern{0.8}%
    \overline{\rel@kern{-0.8}\macc@nucleus\rel@kern{0.2}}%
    \rel@kern{-0.2}%
  }%
  \macc@depth\@ne
  \let\math@bgroup\@empty \let\math@egroup\macc@set@skewchar
  \mathsurround\z@ \frozen@everymath{\mathgroup\macc@group\relax}%
  \macc@set@skewchar\relax
  \let\mathaccentV\macc@nested@a
  \macc@nested@a\relax111{#1}%
  \endgroup
}
\makeatother

\def\WF{{\mathrm{WF}}}

\def\cA{{\ca A}}

\def\cD{{\ca D}}
\def\cE{{\ca E}}
\def\cF{{\ca F}}
\newcommand\cG{{\ca G}}
\def\cH{{\ca H}}
\def\cI{{\ca I}}

\def\cL{{\ca L}}

\def\cO{{\ca O}}
\def\cP{{\ca P}}

\def\cV{{\ca V}}
\def\cW{{\ca W}}

\def\bC{{\mathbb C}}           

\def\bI{{\mathbb I}}

\def\bN{{\mathbb N}}
\def\bM{{\mathbb M}}
\def\NN{{\mathbb N}}
\def\bR{{\mathbb R}}

\def\bS{{\mathbb S}}

\def\RR{{\mathbb R}}

\def\gA{{\mathfrak A}}       

\def\gC{{\mathfrak C}}

\def\gG{{\mathfrak G}}

\def\beq{\begin{eqnarray}}
\def\eeq{\end{eqnarray}}
\def\pa{\partial}

\newcommand{\ca}[1]{{\cal #1}}         

\newcommand{\supp}{\mathrm{supp}\;}
\newcommand{\id}{\mathrm{id}}
\newcommand{\obj}{\mathrm{obj}}
\newcommand{\Tr}{\mathrm{Tr}}
\newcommand{\Sol}{\mathrm{Sol}}
\newcommand{\Solsc}{\mathrm{Sol}_\mathrm{sc}}
\newcommand{\kerc}{\text{Ker}_0}

\newcommand{\wick}[1]{:\!{#1}\!:}

\newcommand{\Gg}{{\mathfrak G}}

\newcommand{\Se}{\Gamma}
\newcommand{\Sec}{\Gamma_0}
\newcommand{\Sesc}{\Gamma_\mathrm{sc}}
\newcommand{\Setc}{\Gamma_\mathrm{tc}}
\newcommand{\Seinf}{\Gamma_\infty}

\newcommand{\ip}[1]{\left\langle{#1}\right\rangle}
\newcommand{\ipp}[1]{\left\langle\left\langle{#1}\right\rangle\right\rangle}

\newcommand{\vol}{d_g x\;}

\newcommand\cGsc{\cG_\mathrm{sc}}
\newcommand\cGscc{\cG_{\mathrm{sc},0}}

\newcommand{\Freg}{\cF_\mathrm{reg}}
\newcommand{\Floc}{\cF_\mathrm{loc}}
\newcommand{\Fmuc}{\cF_{\mu\mathrm{c}}}
\newcommand{\Fmucsol}{\cF_{\mu\mathrm{c},\mathrm{sol}}}
\newcommand{\Fregsol}{\cF_{\mathrm{reg},\mathrm{sol}}}

\newcommand{\sym}{{\mathrm{sym}}}
\newcommand{\mat}{{\mathrm{mat}}}
\newcommand{\rad}{{\mathrm{rad}}}

\newcommand{\EL}{{\mathrm{EL}}}

\newcommand{\wSol}{{\widebar{\mathrm{Sol}}}}
\newcommand{\wSolsc}{{\widebar{\mathrm{Sol}}_\mathrm{sc}}}

\newcommand\wcG{{\widebar{\cG}}}
\newcommand\wcGsc{{\widebar{\cG}_\mathrm{sc}}}
\newcommand\wcGscc{{\widebar{\cG}_{\mathrm{sc},0}}}

\makeindex             

\begin{document}

\author{Thomas--Paul Hack}
\title{Cosmological Applications of Algebraic Quantum
Field Theory in Curved Spacetimes}
\maketitle

\frontmatter

%
%

\preface

This monograph provides a largely self--contained and broadly accessible exposition of two cosmological applications of algebraic quantum field theory (QFT) in curved spacetime: a fundamental analysis of the cosmological evolution according to the Standard Model of Cosmology and a fundamental study of the perturbations in Inflation. The two central sections of the book dealing with these applications are preceded by sections containing a pedagogical introduction to the subject as well as introductory material on the construction of linear QFTs on general curved spacetimes with and without gauge symmetry in the algebraic approach, physically meaningful quantum states on general curved spacetimes, and the backreaction of quantum fields in curved spacetimes via the semiclassical Einstein equation. The target reader should have a basic understanding of General Relativity and QFT on Minkowski spacetime, but does not need to have a background in QFT on curved spacetimes or the algebraic approach to QFT. In particular, I took a great deal of care to provide a thorough motivation for all concepts of algebraic QFT touched upon in this monograph, as they partly may seem rather abstract at first glance. Thus, it is my hope that this work can help non--experts to make `first contact' with the algebraic approach to QFT.

I would like to thank my colleagues and friends Claudio Dappiaggi, Klaus Fredenhagen, Valter Moretti,  Nicola Pinamonti	and Alexander Schenkel, among others, for their past and ongoing support and the fruitful collaborations on some of the topics covered in this monograph. Special thanks are due to Jan M\"oller for the persistent encouragement to apply algebraic quantum field theory to cosmology. I would also like to thank Aldo Rampioni and Kirsten Theunissen at Springer for their patient collaboration on the realisation of this monograph.

The research reported on in this monograph has been supported by the Hamburg research cluster LEXI `Connecting Particles with the Cosmos' as well as a research fellowship of the Deutsche Forschungsgemeinschaft (DFG).

\vspace{\baselineskip}
\begin{flushright}\noindent
June 2016,\hfill {\it Thomas Hack}\\
\end{flushright}

\tableofcontents

\mainmatter

\chapter{Introduction}
\label{chap:intro}

\abstract{In this chapter we give a pedagogical introduction to algebraic quantum field theory and explain the concepts and the relevance of this framework in quantum field theory on curved spacetimes. This introduction should serve as a guide for the next chapter, where many concepts of and constructions in algebraic quantum field theory on curved spacetimes are reviewed in detail. Afterwards, we give a non--technical overview of the cosmological applications discussed in the final chapter of this monograph.}

\section{A Pedagogical Introduction to Algebraic Quantum Field Theory on Curved Spacetimes}
\label{sec:pedagogical_aqft}

{\em Algebraic Quantum Field Theory} (AQFT) \cite{chap1_Haag:1992hx} is a framework which focusses on the local and algebraic properties of QFT and thus aims for understanding structural properties of relativistic quantum field theory from first principles and in a model--independent fashion. In standard textbook treatments of QFT in Minkowski spacetime, the formalism of QFT is developed by constructing operators and deriving relations based on the vacuum state and the associated Hilbert space. However, this approach is not directly generalisable to curved spacetimes as we shall explain now. 

To this avail, we consider a quantized Hermitean scalar field $\phi(x)$ and assume that it can be decomposed in two different ways as
\beq\label{eq_fockdecomp}
\phi(x)=\sum\limits_{i} A_i(x) a_i +  \overline{A_i}(x)a^\dagger_i=\sum\limits_{i} B_i(x) b_i +  \overline{B_i}(x)b^\dagger_i
\eeq
where $a^\dagger_i, a_i$ and $b^\dagger_i, b_i$ are two sets of creation and annihilation operators with corresponding modes $A_i(x)$, $B_i(x)$ and vacua $|\Omega_a\rangle$, $|\Omega_b\rangle$
$$a_i|\Omega_a\rangle = 0\,,\qquad b_i|\Omega_b\rangle = 0\,.$$
The two sets of creating and annihilation operators are related by a Bogoliubov transformation
$$
b_i = \alpha_i a_i + \beta_i a^\dagger_i\,. 
$$
Mathematically the two possible decompositions of $\phi(x)$ seem to be equivalent, whereas physically we would ask the question: Which of the two decompositions of $\phi(x)$ is better, or, alternatively, is there a preferred way to decompose $\phi(x)$? In Minkowski spacetime, we have Poincar\'e symmetry at our disposal, in particular Minkowski spacetime is time--translation invariant and we can construct a Hamilton operator $H$ and obtain a related notion of `energy'. If $H|\Omega_a\rangle=0$, but $H|\Omega_b\rangle\neq 0$, we would call $|\Omega_a\rangle$ the {\em ground state} or {\em vacuum state} and choose to work with the decomposition of $\phi(x)$ in terms of $a^\dagger_i, a_i$. In other words, we would consider -- i.e. {\em represent} -- $\phi(x)$ as an operator in the Fock space of the vacuum state. In curved spacetimes, these ideas fail in general because generic curved spacetimes are not time--translation invariant; prominent examples of such backgrounds are cosmological spacetimes with a metric line element
\beq\label{eq_comsosimple}
ds^2 =-dt^2 + a(t)^2 d\vec{x}^2\,,
\eeq
where the function $a(t)$ is non--constant in time. In the absence of time--translation invariance, no meaningful notion of a Hamilton operator or `energy' exists, and thus we have no means to select or define a vacuum state. Even under these circumstances, one might still think that various possible decompositions of the form \eqref{eq_fockdecomp} are in a sense equivalent, so that it does not really matter which of these one chooses to work with. However, in general one is facing the additional problem that
$$N_{ba}\doteq \sum\limits_i\langle \Omega_a| b_i^\dagger b_i|\Omega_a\rangle = \infty\,,$$
i.e. the `$a$--vacuum' contains infinitely many `$b$--particles', which in mathematical terms implies that $|\Omega_a\rangle$ and $|\Omega_b\rangle$ can not lie in the same Hilbert space. An example of this situation can be constructed by considering a cosmological spacetime of the form \eqref{eq_comsosimple} with $a(t)=\tanh(c t)$ where $c$ is a constant with dimension of inverse time. The asymptotic regions of this spacetime for $t\to\pm\infty$ are time--translation invariant and one can define corresponding asymptotic vacua $|\Omega_\pm\rangle$. One may then compute that e.g. the `$+$'--vacuum contains infinitely many `$-$'--particles which can be physically interpreted by saying that the expansion of space encoded in the functional form of $a(t)$ creates infinitely many particles. This occurs because the quantum field $\phi(x)$ has infinitely many degrees of freedom, which are all excited by the expansion. Consequently, the second equality sign in \eqref{eq_fockdecomp} is in general purely heuristic because the two decompositions of $\phi(x)$ listed there are in general not related in a physically and mathematically meaningful way.

From the above discussion we can infer that in the context of curved spacetimes the very notion of `particle' is strictly speaking meaningless, because it relies on a preferred choice of vacuum state, which in general does not exist. Consequently, one should not regard particles as a fundamental concept in quantum field theory, but at most as a derived concept which is meaningful in an approximative sense if the time scales on which the spacetime background is changing -- or more general, the spacetime curvature scales -- are large compared to the scales relevant for the physical processes we would like to discuss. On more conceptual grounds, we see that a Hilbert space can not be a primary object in the general construction of QFT models. 

Consequently, in algebraic quantum field theory one aims for constructing a QFT in a purely algebraic fashion without recourse to a Hilbert space. In fact, one first constructs an algebra which contains all observables of the theory and encodes their algebraic relations such as commutation relations and equations of motion. Hilbert spaces then appear in a second step as the spaces on which the algebra of observables can be represented. Different physical situations, e.g. states with a different temperature on Minkowski spacetime, are mathematically modelled by in general inequivalent representations. Thus, the algebraic approach to QFT has the advantage that it enables us to discuss the physical properties of a physical system and the physical properties of a state of this system separately. Moreover, it allows us to treat {\em all} states of a physical system described by a QFT at once and in a coherent fashion.

To understand the essential concepts of algebraic quantum field theory in curved spacetimes, it is advisable to investigate the easiest field model, the free Hermitean scalar field. We thus briefly sketch how a Klein--Gordon field $\phi$ propagating on a curved spacetime is treated in the algebraic framework. All concepts and notions we touch upon in the following will be explained in detail in the next chapter of this monograph. 

The first fundamental algebraic property of $\phi$ is the Klein--Gordon equation
$$P\phi\doteq(-\Box+m^2+\xi R)\phi=0\,,$$
where $\Box=\nabla_\mu \nabla^\mu$ is the d'Alembert operator, $m$ is the mass, $R$ is the Ricci curvature scalar and $\xi$ quantifies a non--minimal coupling of $\phi$ to the scalar curvature. $P$ is a {\em (normally) hyperbolic differential operator}\index{normally hyperbolic differential operator}, which means that classical solutions of the Klein--Gordon equation {\em exist} if we prescribe initial data, i.e. the value of $\phi(x)$ and its time--derivative $\partial_t \phi(x)$ at a fixed time $t$. Moreover, the solutions for given initial data are {\em unique} and the value of a solution $\phi$ at a point $x$ depends only on the initial data in the past (or future) lightcone of $x$. Consequently, physical systems described by hyperbolic equations propagate in a {\em causal} and {\em predictive} fashion. In fact, all these statements are not correct on general spacetimes, but they hold on {\em globally hyperbolic spacetimes}\index{globally hyperbolic spacetime}. These are spacetimes which are of the form $(M,g)$, where the spacetime manifold $M$ may be decomposed as $\bR\times\Sigma$, with $\bR$ corresponding to `time' and the manifold $\Sigma$ corresponding to `space', and where the causal (i.e. lightcone) structure specified by the metric $g$ is such that the worldline of any physical observer, i.e. any inextendible timelike curve, hits an `equal--time surface' $\{t_0\}\times \Sigma$ exactly once. Minkowski spacetime and e.g. cosmological spacetimes are globally hyperbolic.

Being a hyperbolic operator, the Klein--Gordon operator on a globally hyperbolic spacetime has unique advanced and retarded Green's functions $E_R(x,y)$, $E_A(x,y)=E_R(y,x)$ which satisfy $P_x E_R(x,y) = P_x E_A(x,y)=\delta(x,y)$ and are such that $E_R(x,y)$ ($E_A(x,y)$) vanishes if $x$ is not in the forward (backward) lightcone of $y$. Given these Green's functions, we may construct the antisymmetric {\em causal propagator}\index{causal propagator} $E(x,y)\doteq E_R(x,y)-E_A(x,y)$ which is sometimes also called {\em commutator function}, {\em spectral function} or {\em Pauli--Jordan distribution}. The second fundamental algebraic property of the quantum field $\phi$ are the {\em canonical commutation relations}\index{canonical commutation relations} (CCR)
$$[\phi(x),\phi(y)]=iE(x,y)\,.$$
Clearly, $E(x,y)$ vanishes if $x$ and $y$ are not causally related and thus the CCR encode the physical requirement that causally unrelated observables commute. Given coordinates $(t,\vec{x})\in M=\bR\times\Sigma$ one can show that $\partial_{t_1} E(t_1,\vec{x}_1,t_2,\vec{x}_2)|_{t_1=t_2}=\delta(\vec{x}_1,\vec{x}_2)$ and $E(t_1,\vec{x}_1,t_2,\vec{x}_2)|_{t_1=t_2}=0$. Consequently, the above covariant CCR are equivalent to equal--time CCR
$$[\partial_t\phi(t,\vec{x}_1),\phi(t,\vec{x}_2)]=i\delta(\vec{x}_1,\vec{x}_2)\,,\qquad [\phi(t,\vec{x}_1),\phi(t,\vec{x}_2)]=0\,.$$
$E(x,y)$ is a singular object -- a distribution -- and diverges for $x$ and $y$ which are lightlike related. Consequently, the quantum field $\phi(x)$ is a singular object as well, which is rooted in the fact that it encodes infinitely many degrees of freedom. For this reason one often considers in the algebraic approach to QFT `smeared fields' 
$$
\phi(f)\doteq \langle f,\phi\rangle\doteq \int\limits_M dx\sqrt{|\det g|}\;\phi(x) f(x)\,,
$$
where $f$, called a `test function', is infinitely often differentiable and has compact support in spacetime. Physically, $f$ has to be interpreted as a `weighting function' such that $\phi(f)$ models a `weighted measurement' of the observable $\phi(x)$. The compact and thus bounded support of $f$ in spacetime reflects the physically realistic situation that detectors have a finite spatial size and measurements are performed in a finite time interval.

The basic algebra of observables $\cA(M)$ of the Hermitean scalar field $\phi$ is constructed by considering sums of products of smeared fields $\phi(f)$ where $f$ ranges over all possible test functions. The Klein--Gordon equation is encoded by identifying $\phi(f)$ with $0$ if $f$ is of the form $P h$ with a test function $h$ and the `smeared CCR' read $[\phi(f_1),\phi(f_2)]=iE(f_1,f_2)$, where $E(f_1,f_2)$ is the causal propagator integrated with the test functions $f_1$ and $f_2$. To have a notion of `taking the adjoint', one introduces a $*$-operation specified by e.g. $(\phi(f)\phi(g))^*=\phi(g)^*\phi(f)^*=\phi(\overline{g})\phi(\overline{f})$. A state $\omega$ on $\cA(M)$ is a linear functional $\omega:\cA(M)\to \bC$, which is positive and normalised, namely, $\omega(A^*A)\ge 0$ for all $A\in\cA(M)$, and $\omega(1)=1$. In this context, $\omega(A)$ for $A\in\cA(M)$ has the physical interpretation of being the expectation value of $A$. Given an algebraic state $\omega$, one obtains a canonical representation $\pi_\omega$ of $\cA(M)$ on a Hilbert space $\cH_\omega$ with vacuum $|\Omega_\omega\rangle$ via the so--called {\em GNS construction}\index{GNS construction}, such that e.g. $\omega(\phi(f))=\langle\Omega_\omega|\pi_\omega(\phi(f))|\Omega_\omega\rangle$. In particular, the algebraic positivity condition $\omega(A^*A)\ge 0$ ensures that the Hilbert space vector $\pi_\omega(A)|\Omega_\omega\rangle$ has positive norm for all $A$ and the normalisation condition $\omega(1)=1$ ensures that $|\Omega_\omega\rangle$ has unit norm.  Conversely, given a Hilbert space $\cH$ with the Klein--Gordon field realised as an operator (valued distribution) on $\cH$, the algebra constituted by these operators together with a normalised Hilbert space state are naturally a field algebra and a state in the abstract sense. Finally, an algebraic state $\omega$ on $\cA(M)$ is uniquely determined, once all its {\em $n$--point correlation functions}\index{correlation functions} $\omega_n(x_1,\cdots, x_n)=\omega(\phi(x_1)\cdots\phi(x_n))$ are known. 

The algebra $\cA(M)$ may be interpreted as the canonical quantization of a space of classical observables with a {\em Poisson bracket} defined by the causal propagator $E$. It is convenient to discuss physical properties such as e.g. gauge--invariance on the level of this classical space before passing to the quantized algebra $\cA(M)$, because -- in the simple case of free field theories -- many such properties of the classical theory automatically carry over to the quantum theory.

The algebra $\cA(M)$ contains only products of quantum fields at {\em different} points, in particular it does not contain quantized versions of the products of scalar fields at the same point such as $\phi(x)^2$. The naive definition of $\phi(x)^2$ as $\lim_{x\to y} \phi(x)\phi(y)$ leads to divergences because the two--point correlation function $\omega_2(x,y)=\omega(\phi(x)\phi(y))$ of any quantum state $\omega$ is singular for $x=y$; this essentially follows from the singularities of the causal propagator $E(x,y)$ and the fact that for every state $\omega$, $\omega_2(x,y)-\omega_2(y,x)=iE(x,y)$ must hold on account of the canonical commutation relations. The singularity of $\omega_2(x,y)$ for $x=y$ is nothing but the `tadpole singularity' well--known from perturbative QFT in Minkowski spacetime. A way to cure this singularity in Minkowski spacetime is to decompose the quantum field $\phi(x)$ represented on the Hilbert space corresponding to a state $\omega$ into creation and annihilation operators and to define a {\em normal ordered} $\wick{\phi(x)^2}_\omega$ by `normal ordering' the creation and annihilation operators in the naive square $\phi(x)^2$. One the algebraic level, this is equivalent to defining
\beq\label{eq_wickintro}
\wick{\phi(x)^2}_\omega \;\doteq\; \lim_{x\to y}\left(\phi(x)\phi(y)-\omega_2(x,y)\right).
\eeq
If $\omega$ is the vacuum state on Minkowski spacetime, then one finds that products of $\wick{\phi(x)^2}_\omega$ at different points are well--defined and may be computed by the {\em Wick theorem} which implies e.g.
$$
\omega\left(\wick{\phi(x)^2}_\omega\;\wick{\phi(y)^2}_\omega\right) = 2\omega_2(x,y)^2\,.
$$
Consequently, normal--ordered quantities form an algebra themselves. This observation relies heavily on the UV--regularity properties of the Minkowski vacuum state, i.e. loosely speaking products of $\wick{\phi(x)^2}_\omega$ at different points are well--defined because the two--point correlation function of the Minkowski vacuum is `singular but not too singular'. Mathematically, one has that the expression $\omega_2(x,y)^2$ is a well--defined distribution such that integrating it with any pair of test functions $f_1$, $f_2$ gives a finite result. If we consider for example the massless case, then the two--point correlation function of the Minkowski vacuum has the form 
$$\omega_2(x,y)=\frac{1}{4\pi^2}\frac{1}{(x-y)^2}\,.$$
In order to mimic the procedure of normal ordering in Minkowski spacetime also in curved spacetimes, we need to consider the proper generalisation of the UV--regularity properties of the Minkowski vacuum. In turns out that {\em Hadamard states}\index{Hadamard state} are precisely the class of states in QFT on curved spacetimes which have the property that products of their correlation functions such as $\omega_2(x,y)^2$ are well--defined. Hadamard states are characterised by having two--point functions of the form
$$\omega_2(x,y)=\frac{1}{8\pi^2}\left(\frac{u(x,y)}{\sigma(x,y)}+v\log(\sigma(x,y))+w(x,y)\right)=H(x,y)+\frac{w(x,y)}{8\pi^2}\,,$$
where $\sigma(x,y)$ is one half the squared geodesic distance, and $u$, $v$, $w$ are infinitely often differentiable (smooth) functions. In this functional form, the UV divergences of $\omega_2(x,y)$ are clearly visible, and they are completely contained in $H(x,y)$. Moreover, we see that the massless Minkowski vacuum is in fact a Hadamard state with $u=1$ and $v=w=0$. It turns out that the {\em Hadamard coefficients}\index{Hadamard coefficients} $u$, $v$, and, hence, $H$ are completely specified by the parameters in the Klein--Gordon operator $P$ and the local curvature in the neighbourhood of the points $x$ and $y$. Hence, the singular part $H$ is completely state--independent, and the two--point functions of two Hadamard states differ only in the regular part $w$.

Following the above discussion, given any Hadamard state $\omega$, we can define meaningful normal ordered field expressions such as $\wick{\phi(x)^2}_\omega$ by \eqref{eq_wickintro}. However, the paradigm in algebraic QFT is to define observables in a state--independent way and $\wick{\phi(x)^2}_\omega$ clearly fails to satisfy this property. Even worse, the regular part $w$ in the correlation function $\omega_2(x,y)$ of any Hadamard state is a highly non--local object because $\omega_2(x,y)$ satisfies the Klein--Gordon equation in both arguments and thus is e.g. sensitive to the functional form of the metric $g$ in the full past lightcone of $x$ and $y$. Consequently $\wick{\phi(x)^2}_\omega$ is a non--local observable, which is conceptually unsatisfactory. In order to cure this problem while still maintaining the property to have well--defined products at different points, we can define 
$$
\wick{\phi(x)^2}_H \;\doteq\; \lim_{x\to y}\left(\phi(x)\phi(y)-H(x,y)\right),
$$
i.e. we subtract only the state--independent singular part and obtain a truly local observable which is independent of the geometry of spacetime far away from $x$. However, this definition of local normal ordering is not the only possibility. Imposing further algebraic properties such as canonical commutation relations with the linear field $\phi(x)$ and particular scaling and regularity properties with respect to the metric and the parameters in the Klein--Gordon equation, one finds that any expression which satisfies these conditions as well as locality must be of the form
$$
\wick{\phi(x)^2} \;=\;\wick{\phi(x)^2}_H + (\alpha R + \beta m^2)\bI\,,
$$
where $\alpha$ and $\beta$ are arbitrary dimensionless real constants. Consequently, we find that normal ordering -- or equivalently the renormalisation of tadpoles -- is inherently ambiguous on curved spacetimes. 

In this section we have only sketched many concepts of algebraic QFT on curved spacetimes and have explained them mostly at the basis of examples. Even though we will provide more details in the next chapter, we would like to stress already at this point that the algebraic construction of perturbative interacting models on curved spacetimes with and without local gauge symmetries is by now well--understood in conceptual terms \cite{chap1_Brunetti:1995rf, chap1_Brunetti:1999jn, chap1_HW01, chap1_HW02, chap1_Hollands:2002ux, chap1_HW04, chap1_HollandsRuan, chap1_Hollands:2007zg, chap1_Chilian:2008ye, chap1_Brunetti:2009qc, chap1_Fredenhagen:2011mq}.

\section{Outline of the Cosmological Applications}
\label{sec:cosmo_outline}

We briefly outline the two cosmological applications of algebraic quantum field theory discussed in detail in the main and final chapter of this monograph.

\subsection{The Cosmological Expansion in QFT on Curved Spacetimes}
\label{sec:cosmo_outline_lcdm}

According to the Standard Model of Cosmology\index{LCDM@$\Lambda$CDM--model} -- the $\Lambda$CDM--model -- our universe contains matter, radiation, and Dark Energy, whose combined energy density determines the expansion of the universe. In the $\Lambda$CDM--model, these three kinds of matter--energy are modelled macroscopically as a perfect fluid and are thus completely determined by an energy density $\rho$ and a pressure $p$, with different equations of state $p = p(\varrho) = w \varrho$, $w=0, \frac13, -1$ for matter, radiation and Dark Energy (assuming that the latter is just due to a cosmological constant) respectively. 

However, at least the contributions to the macroscopic matter and radiation energy densities which are in principle well--understood originate microscopically from particle physics. Hence, it should be possible to derive these contributions from first principles within QFT on curved spacetimes. However, in the standard literature usually a mixed classical/quantum analysis is performed on the basis of effective Boltzmann equations in which the collision terms are computed within QFT on flat spacetime whereas the expansion/curvature of spacetime is taken into account by means of redshift/dilution--terms, see e.g. \cite{chap1_Kolb:1990vq}. After a sufficient amount of cosmological expansion, i.e. in the late universe, the collision terms become negligible and the energy densities of matter and radiation just redshift as dictated by their equation of state.

As a first application of AQFT on curved spacetimes to cosmology, we aim to improve on this situation and to demonstrate that it is indeed possible to derive the form of the energy density in the $\Lambda$CDM--model microscopically within quantum field theory on curved spacetime: we model matter and radiation by quantum fields propagating on a cosmological spacetime and we show that there exist states for these quantum fields in which the energy density has the form assumed in the $\Lambda$CDM--model up to small corrections. Indeed, we find that these small corrections are a possible explanation for the phenomenon of Dark Radiation, which shows that a fundamental analysis of the $\Lambda$CDM--model is not only interesting from the conceptual point of view but also from the phenomenological one.

Due to the complexity of the problem and for the sake of clarity we shall make a few simplifying assumptions. On the one hand, we shall model both matter and radiation by scalar and neutral quantum fields for the ease of presentation, but all concepts and principal constructions we shall use have been developed for fields of higher spin and non--trivial charge as well and we shall mention the relevant literature whenever appropriate. Thus, a treatment taking into account these more realistic fields is straightforward. On the other hand, we shall consider only non--interacting quantum fields and thus the effects of the field interactions which presumably played an important role in the early universe will only appear indirectly as characteristics of the states of the free quantum fields in our description. Notwithstanding, all concepts necessary to extend our treatment to interacting fields are have already been developed as pointed out at the end of the previous section. Finally, in this work we are only interested in modelling the history of the universe from the time of Big Bang Nucleosynthesis (BBN) until today. This restriction also justifies our approximation of considering non--interacting quantum fields, as one usually assumes that field interactions can be neglected on cosmological scales after electron--positron annihilation, which happened roughly at the same time as BBN \cite{chap1_Kolb:1990vq}. 

The quantum states which we will find to microscopically model the macroscopic energy densities in the $\Lambda$CDM--model are generalised thermal excitations of so--called states of low energy\index{states of low energy}, which are homogeneous and isotropic Hadamard states that minimise the energy density integrated against a weighting function $f$ \cite{chap1_Olbermann:2007gn}. Whereas in e.g. Minkowski spacetime the vacuum is the only state of low energy in this sense, in general cosmological spacetimes states of low energy depend on the sampling function $f$ and are thus non--unique as expected. Notwithstanding, we shall compute that for sufficiently large width of the sampling function $f$, the energy density in states of low energy on cosmological spacetimes of $\Lambda$CDM--type is negligible in comparison to the macroscopic energy density in the $\Lambda$CDM--model. This generalises the results found in \cite{chap1_Degner} for the special case of de Sitter spacetime. Consequently, states of low energy with sufficiently large width of their characteristic sampling function {\em all} deserve to be considered as `generalised vacuum states', i.e. as a good approximation to the concept of `vacuum' in cosmological spacetimes. The generalised thermal excitations of the states of low energy we shall consider are phenomenologically well--motivated. For the case of the massless, conformally coupled scalar field modelling radiation, they are just conformal transformations of thermal states in Minkowski spacetime (conformal KMS states) which phenomenologically reflect the observation that the Cosmic Microwave Background (CMB)\index{Cosmic Microwave Background} radiation is thermal. On the other hand, Dark Matter\index{Dark Matter}, which constitutes the major part of the cosmological matter density in the $\Lambda$CDM--model, is in many models considered to be of thermal origin and the quantum states we shall consider for the massive conformally coupled scalar field are thought to model this thermal origin in simple terms.

The analysis in Section \ref{sec:lcdmqft} is essentially identical to \cite{chap1_Hack:2013aya} and constitutes the first full publication of the results reported there.

\subsection{A Birds--Eye View of Perturbations in Inflation}
\label{sec:cosmo_outline_inflation}

The inflationary paradigm\index{Inflation} is by now an important cornerstone of modern cosmology. In the simplest models of Inflation, one assumes that a classical real Klein--Gordon field $\phi$ with a suitable potential $V(\phi)$, coupled to spacetime metric via the Einstein equations, drives a phase of exponential expansion in the early universe. After this phase, the universe respectively its matter--energy content is thought to be almost completely homogenised, whereby the quantized perturbations of the scalar field and the metric are believed to constitute the seeds for the small--scale inhomogeneities in the universe that we observe today.

Mathematically, this idea is usually implemented by considering the coupled Einstein--Klein--Gordon system on a Friedmann--Lema\^itre--Robertson--Walker (FLRW) spacetime\index{FLRW (Friedmann-Lema\^itre-Robertson-Walker) spacetimes}. Given a suitable potential $V(\phi)$, this coupled system will have solutions which display the wanted exponential behaviour. In order to analyse the perturbations in Inflation, the Einstein--Klein--Gordon system is linearised and the resulting linear field theory is quantized on the background solution in the framework of quantum field theory in curved spacetimes. The theory of perturbations in Inflation thus constitutes one of the major applications of this framework.

However, a direct quantization of the linearised Einstein--Klein--Gordon system is potentially obstructed by the fact that this system has gauge symmetries. Thus the usual approach to the quantization of perturbations in Inflation, see e.g. the reviews \cite{chap1_Ellis,chap1_Mukhanov:2005sc,chap1_Straumann:2005mz} and the recent work \cite{chap1_Eltzner:2013soa}, consists of first splitting the degrees of freedom of the perturbed metric into components which transform as scalars, vectors and tensors under the isometry group of the FLRW background, the Euclidean group. Subsequently, gauge--invariant linear combinations of these components as well as the perturbed scalar field are identified, which are then quantized in the standard manner. Thereby it turns out that the tensor components of the perturbed metric are manifestly gauge--invariant, whereas the vector components are essentially pure gauge and thus unphysical. The scalar perturbations instead are usually quantized in terms of the gauge--invariant Mukhanov--Sasaki variable\index{Mukhanov--Sasaki variable}, which is essentially a conformally coupled Klein--Gordon field with a time-dependent mass. In the recent work \cite{chap1_Eltzner:2013soa}, this choice of dynamical variable has been shown to be uniquely fixed by certain natural requirements. The relation between the quantized perturbations of the Einstein--Klein--Gordon system and the small--scale inhomogeneities in the present universe is usually established by relating the power-spectrum of the latter to the power spectrum of the former in several non--trivial steps, cf. e.g. \cite{chap1_Ellis,chap1_Mukhanov:2005sc,chap1_Straumann:2005mz}. An approach which differs in the way this relation is made, and is closer to the spirit of stochastic gravity, may be found in the recent work \cite{chap1_Pinamonti:2013zba}.

The conceptual drawback of the standard approach to quantizing perturbations in Inflation is that this approach makes heavy use of the isometry group and the related preferred coordinate system of FLRW spacetimes and is thus inherently non--covariant. In that sense, it is a bottom--up approach, which is of course well--motivated by the fact that it allows one to make explicit computations. Notwithstanding, it seems advisable to check whether the same results can be obtained in a rather top--down approach, as this would provide a firm conceptual underpinning of the standard approach. Motivated by this, the quantum theory of the linearised Einstein--Klein--Gordon system on arbitrary on--shell backgrounds, and with arbitrary potential $V(\phi)$ and non--minimal coupling to the scalar curvature $\xi$, has been developed in \cite{chap1_ThomasInflation}. In order to deal with the gauge symmetries of this system, \cite{chap1_ThomasInflation} follows ideas of \cite{chap1_DimockVector}, which deals with the gauge--invariant quantization of the vector potential on curved spacetimes. This approach was later used in \cite{chap1_Fewster:2012bj} for quantizing linearised pure gravity on cosmological vacuum spacetimes and generalised in \cite{chap1_Hack:2012dm} in order to encompass arbitrary (Bosonic and Fermionic) linear gauge theories on curved spacetimes. In contrast to the BRST/BV approach to quantum gauge theories, see e.g. \cite{chap1_Hollands:2007zg,chap1_Fredenhagen:2011mq}, and \cite{chap1_Brunetti:2013maa} for an application to perturbative pure quantum gravity on curved spacetimes, the formalism used in \cite{chap1_ThomasInflation} works without the introduction of auxiliary fields, at the expense of being applicable only to linear field theories. We shall review this general formalism to quantize linear gauge theories in Section \ref{sec:linear_gauge} and the quantization of the linearised Einstein--Klein--Gordon system on arbitrary on--shell backgrounds in the first part of Section \ref{sec:inflation}.

In the second part of that section, we consider the special case of FLRW backgrounds and review the results of \cite{chap1_ThomasInflation} on comparing the quantum theory obtained from the general quantization of the linearised Einstein--Klein--Gordon system on on--shell backgrounds with the standard approach to the quantization of perturbations in Inflation. Thereby it turns out that the set of quantum observables in the standard approach, which is spanned by local observables of scalar and tensor type, is contained in the set of observables obtained in the general construction, but strictly smaller. However, one further finds that this discrepancy seems to be alleviated if one restricts to configurations of the linearised Einstein--Klein--Gordon system which vanish at spatial infinity, which apparently is a general assumption in the standard approach, see e.g. \cite{chap1_Makino:1991sg}, because these configurations are considered to be `small' and thus truly perturbative. Namely, it is argued in \cite{chap1_ThomasInflation} that local observables of scalar and tensor type are sufficient for measuring this subset of configurations.

\chapter{Algebraic Quantum Field Theory on Curved Spacetimes}
\label{chap_aqft}

\abstract{In this chapter, we review all background material on algebraic quantum field theory on curved spacetimes which is necessary for understanding the cosmological applications discussed in the next chapter. Starting with a brief account of globally hyperbolic curved spacetimes and related geometric notions, we then explain how the algebras of observables generated by products of linear quantum fields at different points are obtained by canonical quantization of spaces of classical observables. This discussion will be model--independent and will cover both Bosonic and Fermionic models with and without local gauge symmetries. Afterwards, we review the concept of Hadamard states which encompass all physically reasonable quantum states on curved spacetimes. The modern paradigm in QFT on curved spacetimes is that observables and their algebras should be constructed in a local and covariant way. We briefly review the theoretical formulation of this concept and explain how it is implemented in the construction of an extended algebra of observables of the free scalar field which also contains products of quantum fields at coinciding points. Finally, we discuss the quantum stress--energy tensor as a particular example of such an observable as well as the related semiclassical Einstein equation.}

\section{Globally Hyperbolic Spacetimes and Related Geometric Notions}
\label{sec:ghst}

The philosophy of algebraic quantum field theory in curved spacetimes is to set up a framework which is valid on all physically reasonable curved Lorentzian spacetimes and independent of their particular properties. Given this framework, one may then exploit particular properties of a given spacetime such as symmetries in order obtain specific results or to perform explicit calculations. A class of spacetimes which encompasses most cases which are of physical interest are {\em globally hyperbolic spacetimes}\index{globally hyperbolic spacetime}. These include Friedmann--Lema\^itre--Robertson--Walker spacetimes -- in particular Minkowski spacetime -- as well as Black Hole spacetimes such as Schwarzschild-- and Kerr--spacetime, whereas prominent examples of spacetimes which are not globally hyperbolic are Anti de Sitter--spacetime (see e.g. \cite[Chapter 3.5]{chap2_Bar}) and a portion of Minkowski spacetime obtained by restricting one of the spatial coordinates to a finite interval such as the spacetimes relevant for discussing the Casimir effect. The constructions we shall review in the following are well--defined on all globally hyperbolic spacetimes. The physically relevant spacetime examples which are not globally hyperbolic are usually such that sufficiently small portions still have this property. Consequently, the algebraic constructions on globally hyperbolic spacetimes can be extended to these cases by patching together local constructions, see 
for instance \cite{chap2_Ribeiro:2007hv,chap2_Dappiaggi:2014gea}. In this section we shall review the definition of globally hyperbolic spacetimes and a few related differential geometric notions which we shall use throughout this monograph.
 
To this end, in this work a {\em spacetime} $(M,g)$ is meant to be a Hausdorff, connected, smooth manifold $M$, endowed with a Lorentzian metric $g$, the invariant volume measure of which shall be denoted by $\vol \doteq \sqrt{|\det g|}dx$. We will mostly consider four--dimensional spacetimes. However, most notions and results can be formulated and obtained for Lorentzian spacetimes with a dimension $d$ differing from four and we will try to point out how the spacetime dimension affects them whenever it seems interesting and possible. We will follow the monograph by Wald \cite{chap2_WaldBook} regarding most conventions and notations and, hence, work with the metric signature $(-,+,+,+)$. It is often required that a spacetime be second countable, or, equivalently, paracompact, i.e. that its topology has a countable basis. Though, as proven by Geroch in \cite{chap2_GerochSpin1}, paracompactness already follows from the properties of $(M,g)$ listed above. In addition to the attributes already required, we demand that the spacetime under consideration is orientable and time--orientable and that an orientation has been chosen in both respects. We will often omit the spacetime metric $g$ and denote a spacetime by $M$ in brief.

For a point $x\in M$, $T_xM$ denotes the tangent space of $M$ at $x$ and $T^*_xM$ denotes the respective cotangent space; the tangent and cotangent bundles of $M$ shall be denoted by $TM$ and $T^*M$, respectively. If $\chi:M_1\to M_2$ is a diffeomorphism, we denote by $\chi^*$ the {\em pull--back} of $\chi$ and by $\chi_*$ the {\em push--forward} of $\chi$.  $\chi^*$ and $\chi_*$ map tensors on $M_2$ to tensors on $M_1$ and tensors on $M_1$ to tensors on $M_2$,  respectively; they furthermore satisfy $\chi_*=(\chi^{-1})^*$ \cite[Appendix C]{chap2_WaldBook}. In case $g_1$ and $g_2$ are the chosen Lorentzian metrics on $M_1$ and $M_2$ and $\chi_*g_1=g_2$, we call $\chi$ an {\em isometry}; if $\chi_*g_1=\Omega^2g_2$ with a strictly positive smooth function $\Omega$, $\chi$ shall be called a {\em conformal isometry} and $\Omega^2g$ a {\em conformal transformation}\index{conformal transformation} of $g$. Note that this definition differs from the one often used in the case of highly symmetric or flat spacetimes since one does not rescale coordinates, but the metric. A conformal transformation according to our definition is sometimes called {\em Weyl transformation} in the literature. If $\chi$ is an {\em embedding} $\chi:M_1\hookrightarrow M_2$, i.e. $\chi(M_1)$ is a submanifold of $M_2$ and $\chi$ a diffeomorphism between $M_1$ and $\chi(M_1)$, it is understood that a push--forward $\chi_*$ of $\chi$ is only defined on $\chi(M_1)\subset M_2$. In case an embedding $\chi:M_1\hookrightarrow M_2$ between the manifolds of two spacetimes $(M_1,g_1)$ and $(M_2,g_2)$ is an isometry between $(M_1,g_1)$ and $(\chi(M_1),g_2|_{\chi(M_1)})$, we call $\chi$ an {\em isometric embedding}, whereas an embedding which is a conformal isometry between $(M_1,g_1)$ and $(\chi(M_1),g_2|_{\chi(M_1)})$ shall be called a {\em conformal embedding}.

Some works make extensive use of the {\em abstract index notation}, i.e. they use Latin indices to denote tensorial identities which hold in any basis to distinguish them from identities which hold only in specific bases. As this distinction will not be necessary in the present work, we will not use abstract index notation, but shall use Greek indices to denote general tensor components in a coordinate basis $\{\partial_\mu\}_{\mu=0,\ldots,3}$ and shall reserve Latin indices for other uses. We employ the Einstein summation convention, e.g. $A_\mu^\mu\doteq\sum^{3}_{\mu=0}A_\mu^\mu$, and we shall lower Greek indices by means of $g_{\mu\nu}\doteq g(\pa_\mu,\pa_\nu)$ and raise them by $g^{\mu\nu}\doteq (g^{-1})_{\mu\nu}$.

Every smooth Lorentzian manifold admits a unique metric--compatible and torsion--free linear connection, the {\em Levi--Civita connection}, and we shall denote the associated {\em covariant derivative} along a vector field $v$, i.e. a smooth section of $TM$, by $\nabla_v$. We will abbreviate $\nabla_{\pa_\mu}$ by $\nabla_\mu$ and furthermore use the shorthand notation $T_{;\mu_1\cdots\mu_n}\doteq \nabla_{\mu_1}\cdots \nabla_{\mu_n} T$ for covariant derivatives of a tensor field $T$. Our definitions for the {\em Riemann tensor}\index{Riemann tensor@$R_{\alpha\beta\gamma\delta}$, Riemann tensor} $R_{\alpha\beta\gamma\delta}$, the {\em Ricci tensor} $R_{\alpha\beta}$\index{Ricci tensor@$R_{\alpha\beta}$, Ricci tensor}, and the {\em Ricci scalar}\index{Ricci scalar@$R$, Ricci scalar} $R$ are \beq\label{eq_DefinitionRicci} v_{\alpha;\beta\gamma}-v_{\alpha;\gamma\beta}\doteq R_{\alpha\phantom{\lambda}\beta\gamma}^{\phantom{\alpha}\lambda}v_\lambda\,,\qquad R_{\alpha\beta}\doteq R_{\alpha\phantom{\lambda}\beta\lambda}^{\phantom{\alpha}\lambda}\,,\qquad R\doteq R^\alpha_{\phantom{\alpha}\alpha}\,,\eeq where $v_\alpha$ are the components of an arbitrary covector. The Riemann tensor possesses the symmetries \beq \label{eq_RiemannSymmetries}R_{\alpha\beta\gamma\delta}=-R_{\beta\alpha\gamma\delta}=R_{\gamma\delta\alpha
\beta}\,,\qquad  R_{\alpha\beta\gamma\delta}+R_{\alpha\delta\beta\gamma}+R_{\alpha\gamma\delta\beta}=0\eeq
 and fulfils the {\em Bianchi identity}
\beq
\label{eq_BianchiIdentity}
R_{\alpha\beta\gamma\delta;\epsilon}+R_{\alpha\beta\epsilon\gamma;\delta}+R_{\alpha\beta\delta\epsilon;\gamma}=0\,.
\eeq
Moreover, its trace--free part, the {\em Weyl tensor}\index{Cabcd@$C_{\alpha\beta\gamma\delta}$, Weyl tensor}\index{Weyl tensor}, is defined as 
\begin{align*}
C_{\alpha\beta\gamma\delta}&=R_{\alpha \beta \gamma \delta}-\frac{1}{6}\left(g_{\alpha \delta}g_{\beta \gamma}-g_{
\alpha \gamma}g_{\beta \delta}\right)R\\
&\quad -\frac{1}{2}\left(g_{\beta \delta}R_{\alpha \gamma}-g_{\beta \gamma}R_{
\alpha \delta}-g_{\alpha \delta}R_{\beta \gamma}+g_{\alpha \gamma}R_{\beta\delta}\right),
\end{align*}
 where the appearing coefficients differ in spacetimes with $d\neq 4$.  In addition to the covariant derivative, we can define the notion of a {\em Lie derivative} along a vector field $v$: the integral curves $c(s)$ of $v$ with respect to a curve parameter $s$ define, in general only for small $s$ and on an open neighbourhood of $c(0)$, a one--parameter group of diffeomorphisms $\chi^v_s$ \cite[Chapter 2.2]{chap2_WaldBook}. Given a tensor field $T$ of arbitrary rank, we can thus define the Lie derivative of $T$ along $v$ as
$$\cL_v T\doteq \lim_{s\to 0}\left(\frac{(\chi^v_{-s})^*T-T}{s}\right).$$ If $\chi^v_s$ is a one-parameter group of isometries, we call $v$ a {\em Killing vector field}, while in case of $\chi^v_s$ being a one-parameter group of conformal isometries, we shall call $v$ a {\em conformal Killing vector field}. It follows that a Killing vector field $v$ fulfils $\cL_v g=0$, while a conformal Killing vector field $v$ fulfils $\cL_vg=fg$ with some smooth function $f$ \cite[Appendix C.3]{chap2_WaldBook}.

In order to define what it means for a spacetime to be globally hyperbolic, we need a few additional standard notions related to Lorentzian spacetimes. To wit, following our sign convention, we call a vector $v_x\in T_xM$ {\em timelike} if $g(v_x,v_x)<0$, {\em spacelike} if $g(v_x,v_x)>0$, {\em lightlike} or {\em null} if $g(v_x,v_x)=0$, and {\it causal} if it is either timelike or null. Extending this, we call a vector field $v: M\to TM$ spacelike, timelike, lightlike, or causal if it possesses this property at every point. Finally, we call a curve $c: \bR\supset I\to M$, with $I$ an interval, spacelike, timelike, lightlike, or causal if its tangent vector field bears this property. Note that, according to our definition, a trivial curve $c\equiv x$ is lightlike. As $(M,g)$ is time orientable, we can split the lightcones in $TM$ at all points in $M$ into `future' and `past' in a consistent way and say that a causal curve is {\em future directed} if its tangent vector field at a point is always in the future lightcone at this point; {\em past directed} causal curves are defined analogously.

For the definition of global hyperbolicity, we need the notion of inextendible causal curves; these are curves that `run off to infinity' or `run into a singular point'. Hence, given a future directed curve $c$ parametrised by $s$, we call $x$ a {\em future endpoint} of $c$ if, for every neighbourhood $\cO$ of $x$, there is an $s_0$ such that $c(s)\in\cO$ for all $s>s_0$. With this in mind, we say that a future directed causal curve is {\em future inextendible} if, for all possible parametrisations, it has {\em no} future endpoint and we define {\em past inextendible} past directed causal curves similarly. A related notion is the one of a {\em complete geodesic}. A geodesic $c$ is called complete if, in its {\em affine parametrisation} defined by $\nabla_{dc/ds}\frac{dc}{ds}=0$, the affine parameter $s$ ranges over all $\bR$. A manifold $M$ is called {\em geodesically complete} if all geodesics on $M$ are complete.

In the following, we are going to define the generalisations of flat spacetime lightcones in curved spacetimes. By $I^+(x, M)$ we denote the {\em chronological future} of a point $x$ relative to $M$, i.e. all points in $M$ which can be reached by a future directed timelike curve starting from $x$, while $J^+(x, M)$ denotes the {\em causal future} of a point $x$, {\it viz.} all points in $M$ which can be reached by future directed causal curve starting from $x$. Notice that, generally, $x\in J^+(x, M)$ and $I^+(x, M)$ is an open subset of $M$ while the situations $x\notin I^+(x, M)$ and $J^+(x, M)$ being a closed subset of $M$ are not generic, but for instance present in globally hyperbolic spacetimes \cite{chap2_WaldBook}. In analogy to the preceding definitions, we define the {\em chronological past} $I^-(x, M)$ and {\em causal past} $J^-(x, M)$ of a point $x$ by employing past directed timelike and causal curves, respectively. We extend this definition to a general subset $\cO\subset M$ by setting $$I^\pm(\cO, M)\doteq \bigcup\limits_{x\in\cO}I^\pm(x, M)\qquad J^\pm(\cO, M)\doteq \bigcup\limits_{x\in\cO}J^\pm(x, M)\,;$$
additionally, we define $I(\cO,M)\doteq I^+(\cO,M)\cup I^-(\cO,M)$ and $J(\cO,M)\doteq J^+(\cO,M)\cup J^-(\cO,M)$.\index{I@$I^\pm(\cO, M)$, chronological future/past of $\cO\subset M$}\index{J@$J^\pm(\cO, M)$, causal future/past of $\cO\subset M$}
As the penultimate prerequisite for the definition of global hyperbolicity, we say that a subset $\cO$ of $M$ is {\em achronal} if $I^+(\cO, M)\cap\cO$ is empty, i.e. an achronal set is such that every timelike curve meets it at most once. Given a closed achronal set $\cO$, we define its {\em future domain of dependence} $D^+(\cO, M)$ as the set containing all points $x\in M$ such that every past inextendible causal curve through $x$ intersects $\cO$. By our definitions, $D^+(\cO, M)\subset J^+(\cO, M)$, but note that $J^+(\cO, M)$ is in general considerably larger than $D^+(\cO, M)$. We define $D^-(\cO, M)$ analogously and set $D(\cO, M)\doteq D^+(\cO, M)\cup D^-(\cO, M)$. $D(\cO, M)$ is sometimes also called the {\em Cauchy development} of $\cO$. With this, we are finally in the position to state the definition of global hyperbolicity (valid for all spacetime dimensions).

\begin{definition}
\label{def_GloballyHyperbolic}
A {\em Cauchy surface}\index{Cauchy surface} is a closed achronal set $\Sigma\subset M$ with $D(\Sigma, M)=M$. A spacetime $(M,g)$ is called {\em globally hyperbolic \index{globally hyperbolic spacetime}} if it contains a Cauchy surface.
\end{definition}

Although the geometric intuition sourced by our knowledge of Minkowski spacetime can fail us in general Lorentzian spacetimes, it is essentially satisfactory in globally hyperbolic spacetimes. According to Definition \ref{def_GloballyHyperbolic}, a Cauchy surface is a `non--timelike' set on which every `physical signal' or `worldline' must register exactly once. This is reminiscent of a constant time surface in flat spacetime and one can indeed show that this is correct. In fact, Geroch has proved in \cite{chap2_GerochCauchy} that globally hyperbolic spacetimes are topologically $\bR\times\Sigma$ and Bernal and Sanchez \cite{chap2_Bernal1,chap2_Bernal2,chap2_Bernal3} have been able to improve on this and to show that every globally hyperbolic spacetime has a {\em smooth} Cauchy surface $\Sigma$ and is, hence, even diffeomorphic to $\bR\times\Sigma$. This implies in particular the existence of a (non--unique) smooth global {\em time function} $t:M\to \bR$\index{time function}, i.e. $t$ is a smooth function with a timelike and future directed gradient field $\nabla t$; $t$ is, hence, strictly increasing along any future directed timelike curve. In the following, we shall always consider smooth Cauchy surfaces, even in the cases where we do not mention it explicitly.

In the remainder of this chapter, we will gradually see that globally hyperbolic curved spacetimes have many more nice properties well--known from flat spacetime and, hence, seem to constitute the perfect compromise between a spacetime which is generically curved and one which is physically sensible. Particularly, it will turn out that second order, linear, hyperbolic partial differential equations have well--defined global solutions on a globally hyperbolic spacetime. Hence, whenever we speak of a spacetime in the following and do not explicitly demand it to be globally hyperbolic, this property shall be understood to be present implicitly.

On globally hyperbolic spacetimes, there can be no closed timelike curves, otherwise we would have a contradiction to the existence of a smooth and strictly increasing time function. There is a causality condition related to this which can be shown to be weaker than global hyperbolicity, namely, {\it strong causality}. A spacetime is called strongly causal if it can not contain almost closed timelike curves, i.e. for every $x\in M$ and every neighbourhood $\cO_1\ni x$, there is a neighbourhood $\cO_2\subset\cO_1$ of $x$ such that no causal curve intersects $\cO_2$ more than once. One might wonder if this weaker condition can be filled up to obtain full global hyperbolicity and indeed some references, e.g. \cite{chap2_Bar,chap2_HawkingBook}, define a spacetime $(M,g)$ to be globally hyperbolic if it is strongly causal and $J^+(x)\,\cap\,J^-(y)$ is compact for all $x,y\in M$. One can show that the latter definition is equivalent to Definition \ref{def_GloballyHyperbolic} \cite{chap2_Bar,chap2_WaldBook} which is, notwithstanding, the more intuitive one in our opinion.

We close this section by introducing a few additional sets with special causal properties. To this avail, we denote by $\exp_x$ the exponential map at $x\in M$. A set $\cO\subset M$ is called {\em geodesically starshaped} with respect to $x\in\cO$ if there is an open subset $\cO^\prime$ of $T_xM$ which is starshaped with respect to $0\in T_xM$ such that $\exp_x:\cO^\prime\to\cO$ is a diffeomorphism. We call a subset $\cO\subset M$ {\em geodesically convex} if it is geodesically starshaped with respect to all its points. This entails in particular that each two points $x$, $y$ in $\cO$ are connected by a unique geodesic which is completely contained in $\cO$. A related notion are {\em causal domains}, these are subsets of geodesically convex sets which are in addition globally hyperbolic. Finally, we would like to introduce {\em causally convex regions}, a generalisation of geodesically convex sets. They are open, non-empty subsets $\cO\subset M$ with the property that, for all $x, y \in \cO$, all causal curves connecting $x$ and $y$ are entirely contained in $\cO$. One can prove that every point in a spacetime lies in a geodesically convex neighbourhood and in a causal domain \cite{chap2_Friedlander} and one might wonder if the case of a globally hyperbolic spacetime which is geodesically convex is not quite generic. However, whereas Friedmann--Lema\^itre--Robertson--Walker spacetimes with flat spatial sections are geodesically convex, even de Sitter spacetime, which is both globally hyperbolic and maximally symmetric and could, hence, be expected to share many properties of Minkowski spacetime, is not.

\section{Linear Classical Fields on Curved Spacetimes}
\label{sec:linear_class}

As outlined in Section \ref{sec:pedagogical_aqft}, the `canonical' route to quantize linear classical field theories on curved spacetimes in the algebraic language is to first construct the canonical covariant classical Poisson bracket (or a symmetric equivalent in the case of Fermionic theories) and then to quantize the model by enforcing canonical (anti)commutation relations defined by this bracket. In this section, we shall first review how this is done for free field theories without gauge symmetry before discussing the case where local gauge symmetries are present.

\subsection{Models without Gauge Symmetry}
\label{sec:linear_nongauge}

We shall start our discussion of classical field theories without local gauge symmetries by looking at the example of the free Klein--Gordon field, which is the `harmonic oscillator' of QFT on curved spacetimes. In discussing this example it will become clear what the basic ingredients determining a linear field theoretic model are and how they enter the definition and construction of this model in the algebraic framework.

\subsubsection{The Free Neutral Klein--Gordon Field}
\label{ssec_EOMScalar}

In Physics, we are used to describe dynamics by (partial) differential equations and initial conditions. The relevant equation for the neutral scalar field is the free Klein--Gordon equation\index{Klein--Gordon equation}
$$(-\Box+ f)\phi\doteq(-\nabla_\mu\nabla^\mu+ f)\phi=0\,$$
with the {\em d'Alembert operator} $\Box$ and some scalar function $f$ of mass dimension 2. The function $f$ determines the `potential' $V(\phi) = \frac12 f \phi^2$ of the Klein--Gordon field and may be considered as a background field just like the metric $g$. Usually one considers the case $f=m^2 + \xi R$, i.e. 
\beq
\label{eq_KleinGordenMass}
P\phi\doteq\left(-\Box+\xi R+m^2\right)\phi=0
\eeq
such that $f$ is entirely determined in terms of a constant mass $m\ge 0$ and the Ricci scalar $R$ where the dimensionless constant $\xi$ parametrises the strength of the coupling of $\phi$ to $R$. In principle one could consider more non--trivial coupling terms with the correct mass dimension such as $f=(R^{\alpha\beta}R_{\alpha\beta})^2 /(m^4 R)$, however these may be ruled out by either invoking Occam's razor or by demanding analytic dependence of $f$ on the metric and $m$ like in \cite[Section 5.1]{chap2_HW01}.

The case \eqref{eq_KleinGordenMass} with $\xi=0$ is usually called {\em minimal coupling}, whereas, in four dimensions, the case $\xi=1/6$ is called {\em conformal coupling}\index{conformal coupling}. While the former name refers to the fact that the Klein--Gordon field is coupled to the background metric only via the covariant derivative, the reason for the latter is rooted in the behaviour of this derivative under conformal transformations. Namely, if we consider the conformally related metrics $g$ and $\widetilde g\doteq \Omega^2 g$ with a strictly positive smooth function $\Omega$, denote by $\nabla$, $\Box$, and $R$ the quantities associated to $g$ and by $\widetilde\nabla$, $\widetilde\Box$, and $\widetilde R$ the quantities associated to $\widetilde g$, then the respective metric compatibility of the covariant derivatives $\nabla$ and $\widetilde\nabla$ and their agreement on scalar functions imply \cite[Appendix D]{chap2_WaldBook}
\beq
\label{eq_ConformalInvarianceScalar}
\left(-\widetilde\Box+\frac16\widetilde R\right)\frac1\Omega = \frac{1}{\Omega^3}\left(-\Box+\frac16 R\right).
\eeq
This entails that a function $\phi$ solving $(-\Box+\frac16 R)\phi=0$ can be mapped to a solution $\widetilde\phi$ of $(-\widetilde\Box+\frac16 \widetilde R)\widetilde\phi=0$ by multiplying it with the conformal factor $\Omega$ to the power of the {\em conformal weight} $-1$, i.e. $\widetilde\phi=\Omega^{-1}\phi$. We shall therefore call a scalar field $\phi$ with an equation of motion $(-\Box+\frac16 R)\phi=0$ {\em conformally invariant}. In other spacetimes dimensions $d\neq 4$, the conformal weight and the magnitude of the conformal coupling are different, see \cite[Appendix D]{chap2_WaldBook}.

Having a partial differential equation for a free scalar field at hand, one would expect that giving sufficient initial data would determine a unique solution on all $M$. However, this is, in case of the Klein--Gordon operator at hand, in general only true for globally hyperbolic spacetimes. To see a simple counterexample, let us consider Minkowski spacetime with a compactified time direction and the massless case, i.e. the equation $(\pa_t^2-\pa_x^2-\pa_y^2-\pa_z^2)\phi=0$. Giving initial conditions $\phi|_{t=0}=0$, $\pa_t\phi|_{t=0}=1$, a possible {\em local} solution is $\phi\equiv t$. But this can of course never be a global solution, since one would run into contradictions after a full revolution around the compactified time direction.

In what follows, the {\em fundamental solutions}\index{fundamental solution} or {\em Green's functions}\index{Green's function} of the Klein--Gordon equation shall play a distinguished role. Before stating their existence, as well as the existence of general solutions, let us define the function spaces we shall be working with in the following, as well as their topological duals, see e.g. \cite[Chapter VI]{chap2_ChoquetBruhat} for an introduction.

\begin{definition}
 \label{def_FunctionSpacesScalar}
By $\Se(M)\doteq C^\infty(M,\bR)$ we denote the {\em smooth (infinitely often continuously differentiable), real--valued functions} on $M$ equipped with the usual {\em locally convex topology}, i.e. a sequence of functions $f_n\in\Se(M)$ is said to converge to $f\in\Se(M)$ if all derivatives of $f_n$ converge to the ones of $f$ uniformly on all compact subsets of $M$.

The space $\Sec(M)\doteq C_0^\infty(M,\bR)$ is the subset of $\Se(M)$ constituted by the {\em smooth, real--valued functions with compact support}. We equip $\Sec(M)$ with the locally convex topology determined by saying that a sequence of functions $f_n\in\Sec(M)$ converges to $f\in\Sec(M)$ if there is a compact subset $K\subset M$ such that all $f_n$ and $f$ are supported in $K$ and all derivatives of $f_n$ converge to the ones of $f$ uniformly in $K$.

By $\Se^\bC_0(M)\doteq  \Sec(M)\otimes_\bR \bC$, $\Se^\bC(M)\doteq  \Se(M)\otimes_\bR \bC$ we denote the complexifications of $\Sec(M)$ and $\Se(M)$ respectively.

The spaces $\Sesc(M)$ and $\Setc(M)$ denote the subspaces of $\Se(M)$ consisting of functions with {\em spacelike--compact} and {\em timelike--compact support} respectively\index{spacelike--compact}\index{timelike--compact}. I.e. $\supp f\cap \Sigma$ is compact for all Cauchy surfaces $\Sigma$ of $(M,g)$ and all $f\in \Sesc(M)$, whereas for all $f\in\Setc(M)$ there exist two Cauchy surfaces $\Sigma_1$, $\Sigma_2$ with $\supp f\subset J^-(\Sigma_1,M)\cap J^+(\Sigma_2,M)$.

By $\Se^\prime_0(M)$ we denote the space of {\em distributions}, i.e. the topological dual of $\Sec(M)$ provided by continuous,  linear functionals $\Sec(M)\to \bR$, whereas $\Se^\prime(M)$ denotes the topological dual of $\Se(M)$, i.e. the space of {\em distributions with compact support}. $\Se^{\prime \bC}(M)$ and $\Sec^{\prime \bC}(M)$ denote the complexified versions of the real--valued spaces.

For $f\in\Se(M)$ and $u\in\Se^\prime_0(M)\supset \Se(M)\supset \Sec(M)$ with compact overlapping support, we shall denote the (symmetric and non--degenerate) {\em dual pairing} of $f$ and $u$ by
$$\langle u, f\rangle\doteq \int\limits_M\vol u(x)f(x)\,.$$
\index{Gamma0(M)@$\Sec(M)$ ($\Se^\bC_0(M)$), space of smooth (complex--valued) functions on $M$}
\index{Gamma(M)@$\Se(M)$ ($\Se^\bC(M)$), space of smooth (complex--valued) functions on $M$}
\index{Gamma0'(M)@$\Se^\prime_0(M)$ ($\Se^{\prime \bC}_0(M)$), space of (complex--valued) distributions on $M$}
\index{Gamma'(M)@$\Se^\prime(M)$ ($\Se^{\prime \bC}(M)$), space of compactly supported (complex--valued) distributions on $M$}
\end{definition}

The physical relevance of the above spaces is that functions in $\Sec(M)$, so--called {\em test functions}, should henceforth essentially be viewed as encoding the localisation of some observable in space and time, reflecting the fact that a detector is of finite spatial extent and a measurement is made in a finite time interval. From the point of view of dynamics, initial data for a partial differential equation may be encoded by distributions or functions with both compact and non--compact support, whereas solutions of hyperbolic partial differential equations like the Klein--Gordon one are typically distributions or smooth functions which do {\em not} have compact support on account of the causal propagation of initial data; having a solution with compact support in time would entail that data `is lost somewhere'. Moreover, fundamental solutions of differential equations will always be singular distributions, as can be expected from the fact that they are solutions with a singular $\delta$--distribution as source. Finally, since (anti)commutation relations of quantum fields are usually formulated in terms of fundamental solutions, the quantum fields and their expectation values will also turn out to be singular distributions quite generically. Physically this stems from the fact that a quantum field has infinitely many degrees of freedom.

Let us now state the theorem which guarantees us existence and properties of solutions and fundamental solutions (also termed Green's functions or propagators) of the Klein--Gordon operator $P$. We refer to the monograph \cite{chap2_Bar} for the proofs.

\begin{theorem}
\label{thm_FundamentalSolutionScalar}
Let $P:\Se(M)\to\Se(M)$ be a {\em normally hyperbolic operator}\index{normally hyperbolic operator} on a globally hyperbolic spacetime $(M,g)$, i.e. in each coordinate patch of $M$, $P$ can be expressed as
$$P=-g^{\mu\nu}\pa_\mu\pa_\nu+A^\mu\pa_\mu+B$$\index{P@$P$, partial differential operator specifying the equation of motion}
with smooth functions $A^\mu$, $B$ and the {\em metric principal symbol} $-g^{\mu\nu}\pa_\mu\pa_\nu$. Then, the following results hold.
\begin{enumerate}
\item Let $f\in \Sec(M)$, let $\Sigma$ be a smooth Cauchy surface of $M$, let $(u_0,u_1)\in \Sec(\Sigma)\times\Sec(\Sigma)$, and let $N$ be the future directed timelike unit normal vector field of $\Sigma$. Then, the {\em Cauchy problem}\index{Cauchy problem}
$$Pu=f,\qquad u|_\Sigma\equiv u_0,\qquad \nabla_N u|_\Sigma\equiv u_1$$
has a unique solution $u\in\Se(M)$. Moreover, $$\supp u\subset J\left(\supp f\cup \supp u_0 \cup \supp u_1,M\right).$$ A unique solution to the Cauchy problem also exists if the assumptions on the compact support of $f$, $u_0$ and $u_1$ are dropped. 
\item There exist unique {\em retarded} $E_R$ and {\em advanced} $E_A$ {\em fundamental solutions (Green's functions, propagators)} of $P$. Namely, there are unique continuous maps $E_{R/A}:\Sec(M)\to\Se(M)$ satisfying $P\circ E_{R/A}=E_{R/A}\circ  P=\id_{\Sec(M)}$ and $\supp E_{R/A} f \subset J^\pm(\text{supp }f, M)$ for all $f\in\Sec(M)$.%
\index{E@$E_{R}$, $E_{A}$ retarded / advanced fundamental solution (propagator, Green's function) of the Klein--Gordon operator}
\index{E@$E$, causal propagator of the Klein--Gordon operator}

\item Let $f$, $g\in \Sec(M)$. If $P$ is {\em formally selfadjoint}\index{formally selfadjoint}, i.e. $\langle f,Pg\rangle=\langle Pf,g\rangle$, then $E_{R}$ and $E_{A}$ are the formal adjoints of one another, namely,
    $\langle f, E_{R/A} g\rangle=\langle E_{A/R} f, g\rangle$.
\item The {\em causal propagator (Pauli--Jordan function)}\index{causal propagator} of $P$ defined as $E\doteq E_{R}-E_{A}$ is a continuous map $\Sec(M)\to\Sesc(M)\subset \Se(M)$ satisfying: for all solutions $u$ of $Pu=0$ with compactly supported initial conditions on a Cauchy surface there is an $f\in\Sec(M)$ such that $u=E f$. Moreover, for every $f\in\Sec(M)$ satisfying $E f=0$ there is a $g\in\Sec(M)$ such that $f=Pg$. Finally if $P$ is formally self--adjoint, then $E$ is formally skew--adjoint, i.e.
    $\langle f, E g\rangle=-\langle E f, g\rangle$\,.
\end{enumerate}
\end{theorem}

\noindent The Klein--Gordon operator $P$ is manifestly normally hyperbolic. Moreover, one can check by partial integration that $P$ is also formally self--adjoint. Hence, all above--mentioned results hold for $P$. 

By continuity and the fact that $\Se(M)\subset \Se^\prime_0(M)$, the operators $E_{R/A}$ and $E$ define bi--distributions $E_{R/A}, E\in \Sec^\prime(M^2)$ which we denote by the same symbol via e.g.
$$E_{R/A}(f,g)\doteq \langle f,E_{R/A} g\rangle = \int\limits_{M^2}\vol  d_gy\;E_{R/A}(x,y)f(x)g(y)\,.$$
In terms of integral kernels of these distributions, some of the identities stated in Theorem \ref{thm_FundamentalSolutionScalar} read
$$P_x E_{R/A}(x,y) = \delta(x,y)\,,\qquad E_A(x,y)=E_R(y,x)\,,\qquad E(y,x)= -E(x,y)\,.$$

The support properties of $E_{R/A}$ entail that $E(f,g)$ vanishes if the supports of $f$ and $g$ are spacelike separated. On the level of distribution kernels, this implies that $E(x,y)$ vanishes for spacelike separated $x$ and $y$. In anticipation of the quantization of the free Klein--Gordon field, this qualifies $E(x,y)$ as a {\em commutator function}. In the classical theory instead, $E(x,y)$ defines a {\em Poisson bracket} or {\em symplectic form}. To see this, we first need to specify the vector space on which this bracket should be evaluated.

\begin{definition}\label{def_labellingspaceKleinGordon}By $\Sol $ ($\Solsc$) we denote the space of real (spacelike--compact) solutions of the Klein--Gordon equation
$$\Sol \doteq \{\phi\in \Se(M)\,|\,P\phi=0\}\,,\qquad \Solsc\doteq \Sol \cap \Sesc(M)\,.$$
\index{Sol@$\Sol$, space of smooth solutions of the equation of motion}
\index{Solsc@$\Solsc$, space of smooth solutions of the equation of motion with spacelike--compact support}
By $\cE $ we denote the quotient space 
$$\cE \doteq \Sec(M)/P[\Sec(M)]\,,$$
which is the {\em labelling space of linear on--shell observables of the free neutral Klein--Gordon field}.
\index{Ecal@$\cE$, labelling space of linear (gauge--invariant) on--shell observables}
\end{definition}
The fact that $\cE $ is the labelling space of (classical) linear on--shell observables of the free neutral Klein--Gordon field follows from the observation that each equivalence class $[f]\in\cE $ defines a linear functional on $\Sol $ by 
$$\Sol \ni\phi\mapsto \cO_{[f]}(\phi)\doteq \langle f, \phi \rangle\,,$$
where we note that, in the classical theory, $\Sol $ plays the role of the space of pure states of the model. As $\phi$ is a solution of the Klein--Gordon equation $\cO_{[f]}(\phi)$ does not depend on the representative $f\in[f]$ and is well--defined. The observable $ \langle f, \phi \rangle$ may be interpreted as the `smeared classical field' $\phi(f)\simeq \langle f, \phi \rangle$. The classical observable $\phi(x)$, i.e. the observable that gives the value of a configuration $\phi$ at the point $x$, may be obtained by formally considering $\phi(f)$ with $f=\delta_x$.

We know that every $\phi\in\Sol $ is in one--to--one correspondence with initial data given on an arbitrary but fixed Cauchy surface $\Sigma$ of $(M,g)$. Analogously the support of a representative $f\in[f]\in\cE $ can be chosen to lie in an arbitrarily small neighbourhood of an arbitrary Cauchy surface.

\begin{lemma}
\label{lem_TimeSliceClassicalScalar}
Let $[f]\in\cE $ be arbitrary and let $\Sigma$ be any Cauchy surface of $(M,g)$. Then, for any bounded neighbourhood $\cO(\Sigma)$ of $\Sigma$, we can find a $g\in\Sec(M)$ with $\supp g\subset \cO(\Sigma)$ and $g\in [f]$.
\end{lemma}

\begin{proof}
Let us assume that $\cO(\Sigma)$ lies in the future of $\supp f$, i.e. $J^-(\supp f,M)\cap \cO(\Sigma)=\emptyset$, the other cases can be treated analogously. Let us consider two auxiliary Cauchy surfaces $\Sigma_1$ and $\Sigma_2$ which are both contained in $\cO(\Sigma)$ and which are chosen such that $\Sigma_2$ lies in the future of $\Sigma$ whereas $\Sigma_1$ lies in the past of $\Sigma$. Moreover, let us take a smooth function $\chi\in\Gamma(M)$ which is identically vanishing in the future of $\Sigma_2$ and fulfils $\chi\equiv 1$ in the past of $\Sigma_1$ and let us define $g\doteq f-P\chi E_R f$. By construction and on account of the properties of both a globally hyperbolic spacetime $(M,g)$ and a retarded fundamental solution $E_R$ on $M$, $\chi E_R f$ has compact support, hence $g\in[f]$. Finally, $\supp g$ is contained in a compact subset of $J^+(\supp f,M)\cap \cO(\Sigma)$.
\end{proof}

We now observe that the causal propagator $E$ induces a meaningful Poisson bracket on $\cE $.

\begin{proposition}\label{prop_presymplectic} The tuple $(\cE ,\tau)$ with $\tau:\cE \times\cE \to\bR$ defined by 
$$\tau([f],[g])\doteq \langle f, E g\rangle$$
\index{tau@$\tau$, (pre--)symplectic or (pre--)inner product on $\cE$}
is a {\em symplectic space}\index{symplectic space}. In particular
\begin{enumerate}
\item $\tau$ is well--defined and independent of the chosen representatives, 
\item $\tau$ is antisymmetric,
\item $\tau$ is (weakly) non--degenerate, i.e. $\tau([f],[g])=0$ for all $[g]\in\cE $ implies $[f]=[0]$.
\end{enumerate}
\end{proposition}
                                                         
\begin{proof}
$\tau$ is independent of the chosen representatives because $P\circ E = 0$. $\tau$ is antisymmetric because $E$ is formally skew--adjoint, cf. the last item of Theorem \ref{thm_FundamentalSolutionScalar}, and because $\langle\cdot,\cdot\rangle$ is symmetric. The non--degeneracy of $\tau$ follows again from the last item of Theorem \ref{thm_FundamentalSolutionScalar} and the fact that $\langle\cdot,\cdot\rangle$ is non--degenerate.
\end{proof}

In standard treatments on scalar field theory, one usually defines Poisson brackets at `equal times', but as realised by Peierls in \cite{chap2_Peierls}, one can give a covariant version of the Poisson bracket which does not depend on a splitting of spacetime into space and time, and this is what we have given above. To relate the covariant form $\tau$ to an equal--time version, we need the definition of a `future part'\index{future part}\index{past part} of a function $f\in\Se(M)$. 

\begin{definition}\label{def_futurepart}
We consider a temporal cutoff function $\chi$ of the form discussed in the proof of Lemma \ref{lem_TimeSliceClassicalScalar}, i.e. a smooth function $\chi$ which is identically vanishing in the future of some Cauchy surface $\Sigma_2$ and identically one in the past of some Cauchy surface $\Sigma_1$ in the past of $\Sigma_2$. Given such a $\chi$, we define for an arbitrary $f\in\Se(M)$ the {\em future part} $f^+$ and the {\em past part} $f^-$ by 
$$f^+\doteq (1-\chi)f\,,\qquad f^- = \chi f\,.$$
\end{definition}

The relation of the covariant picture to the equal time--picture can be now shown in several steps.

\begin{theorem}\label{thm_covariantequaltime}
Let $\ip{\cdot,\cdot}_{\Sol}$ be defined on tuples of solutions with compact overlapping support by 
$$\Sol \times\Sol \ni(\phi_1,\phi_2)\mapsto \ip{\phi_1,\phi_2}_{\Sol}\doteq \ip{P\phi^+_1,\phi_2}\,.$$ 
\index{\$ipsol@$\ip{\cdot,\cdot}_{\Sol}$, bilinear form on $\Sol$}%
Moreover, let $\Sigma$ be an arbitrary Cauchy surface of $(M,g)$ with future--pointing unit normal vectorfield $N$ and canonical measure $d\Sigma$ induced by $\vol $. 
\begin{enumerate}
\item The causal propagator $E:\Sec(M)\to\Sesc(M)$ descends to a bijective map $E:\cE \to\Solsc$.
\item $\ip{\cdot,\cdot}_{\Sol}$ is antisymmetric and well--defined on all tuples of solutions with compact overlapping support, in particular this bilinear form does not depend on the choice of cutoff $\chi$ entering the definition of the future part.
\item For all $f\in\Sec(M)$ and all $\phi\in\Sol $, $\ip{f,\phi}=\ip{E f,\phi}_{\Sol}$. In particular, 
$\ip{\cdot,\cdot}_{\Sol}$ is well--defined on all tuples of solutions with spacelike--compact overlapping support.
\item For all $f,g\in\Sec(M)$, $\tau([f],[g])=\ip{E f,E g}_{\Sol}$, thus the causal propagator $E:\Sec(M)\to\Sesc(M)$ descends to an isomorphism between the symplectic spaces $(\cE ,\tau)$ and $(\Solsc,\ip{\cdot,\cdot}_{\Sol})$. 
\item For all $\phi_1,\phi_2\in\Sol$ with spacelike--compact overlapping support, 
$$\ip{\phi_1,\phi_2}_{\Sol}=\int\limits_\Sigma d\Sigma\; N^\mu j_\mu(\phi_1,\phi_2)\,,\qquad j_\mu(\phi_1,\phi_2) \doteq \phi_1 \nabla_\mu \phi_2 - \phi_2 \nabla_\mu \phi_1\,.$$
\item 
For all $f\in\Sec(\Sigma)$ it holds
$$\nabla_N E f|_\Sigma=f\,,\qquad E f|_\Sigma=0\,.$$ On the level of distribution kernels, this entails that
$$\nabla_NE(x,y)|_{\Sigma\times \Sigma}=\delta_\Sigma(x,y)\,,\qquad E(x,y)|_{\Sigma\times \Sigma}\equiv 0\,,$$
where $\delta_\Sigma$ is the $\delta$-distribution with respect to the canonical measure on $\Sigma$.
\end{enumerate}
\end{theorem}

\begin{proof}We sketch the proof. The first statement follows from the last item of Theorem \ref{thm_FundamentalSolutionScalar}. The fact that $\ip{\cdot,\cdot}_{\Sol}$ is well--defined follows from the observation that two different definitions $\phi^+$, $\phi^{+\prime}$ of the future part differ by a compactly supported smooth function $f=\phi_1^+-\phi_1^{+\prime}$; consequently the supposedly different definitions of the bilinear form  differ by $\ip{\phi_1,\phi_2}_{\Sol}-\ip{\phi_1,\phi_2}^\prime_{\Sol}=\ip{Pf,\phi_2} = \ip{f,P\phi_2}=0$. Note that this partial integration is only possible because $f$ has compact support, in particular, $\ip{\cdot,\cdot}_{\Sol}$ is non--vanishing in general. The antisymmetry of $\ip{\cdot,\cdot}_{\Sol}$ follows by similar arguments and $P\phi^+ = P(\phi-\phi^-) = -P\phi^-$. The third statement follows from the fact that $E_R f$ is a valid future part of $E f$, thus $\ip{E f,\phi}_{\Sol} = \ip{P E_R f, \phi} = \ip{f,\phi}$. The fourth statement follows immediately from the first and third one, whereas the fifth one follows from $\nabla^\mu j_\mu(\phi_1,\phi_2)=\phi_2P\phi_1 - \phi_1P\phi_2$ by an application of Stokes theorem, see e.g. \cite{chap2_DimockScalar}, where also a proof of the last statement can be found.
\end{proof}

We now interpret the previous results. As argued above, elements $[f]\in\cE $ label linear on--shell observables $\phi(f)\simeq \ip{f,\phi}$, i.e. the classical field $\phi$ smeared with the test function $f$. The causal propagator $E$ induces a non--degenerate symplectic form $\tau$ on $\cE $, which we may interpret as $\{\phi(f),\phi(g)\}\simeq \tau([f],[g])=\ip{f,E g}$, or, formally, as $\{\phi(x),\phi(y)\}=E(x,y)$. On the other hand, since $(\cE ,\tau)$ and $(\Solsc,\ip{\cdot,\cdot}_{\Sol})$ are symplectically isomorphic, we can equivalently label linear on--shell observables by $\Solsc\ni u$, i.e. by $\ip{u,\phi}_{\Sol}$, the classical field `symplectically smeared' with the test solution $u$, where this symplectic smearing consists of integrating a particular expression at equal times. The last result of the above theorem implies that the covariant Poisson bracket $\{\phi(x),\phi(y)\}=E(x,y)$ has the well--known equal--time equivalent 
$$\{\nabla_N\phi(x)|_\Sigma,\phi(y)|_\Sigma\}=\nabla_NE(x,y)|_{\Sigma\times \Sigma}=\delta_\Sigma(x,y)\,,$$
$$\{\phi(x)|_\Sigma,\phi(y)|_\Sigma\}=E(x,y)|_{\Sigma\times \Sigma}=0\,,$$ 
which may be interpreted as equal--time Poisson brackets of the field $\phi(x)$ and its `canonical momentum' $\nabla_N\phi(x)$. Further details on the relation between the equal--time and covariant picture can be found e.g. in \cite[Chapter 3]{chap2_WaldBook2}.

\subsubsection{General Models without Gauge Symmetry}
\label{sec_generalnongauge}

The previous discussion of the classical free neutral Klein--Gordon field revealed the essential ingredients defining this model. Following e.g. \cite{chap2_Sahlmann:2000zr,chap2_Bar:2011iu}, this can be generalised to define an arbitrary linear field--theoretic model on a curved spacetime.
\begin{definition}\label{def_linearnongauge}
A real Bosonic linear field--theoretic model without local gauge symmetries on a curved spacetime is defined by the data $(M,\cV,P)$, where
\begin{enumerate}
\item $M\simeq (M,g)$ is a globally hyperbolic spacetime,
\item $\cV$ is a real vector bundle over $M$, the space of smooth sections $\Se(\cV)$ of $\cV$ is endowed with a symmetric and non--degenerate bilinear form $\ip{\cdot,\cdot}_\cV$
\index{\$bilinearV@$\ip{\cdot,\cdot}_\cV$, non--degenerate (anti)symmetric bilinear form on $\Se(\cV)$}
\index{\$bilinearVfibre@$\ipp{\cdot,\cdot}_\cV$, non--degenerate (anti)symmetric fibrewise bilinear form on $\cV$}
which is well--defined on sections with compact overlapping support and given by the integral of a fibrewise symmetric and non--degenerate bilinear form $\ipp{\cdot,\cdot}_\cV:\Se(\cV)\times\Se(\cV)\to\Se(M)$,
\item $P:\Se(\cV)\to\Se(\cV)$ is a Green--hyperbolic\index{Green--hyperbolic operator} partial differential operator, i.e. there exist unique advanced $E^P_R$ and retarded $E^P_A$ fundamental solutions of $P$ 
\index{EP@$E^P_R$, $E^P_A$, retarded / advanced fundamental solution of a Green--hyperbolic operator $P$}
which satisfy $P\circ E^P_{R/A}=E^P_{R/A}\circ P=\id|_{\Sec(\cV)}$ and $\supp E^P_{R/A} f \subset J^\pm(\supp f,M)$ for all $f\in\Sec(\cV)$; moreover $P$ is formally self--adjoint with respect to $\ip{\cdot,\cdot}_\cV$.
\end{enumerate}
A real Fermionic linear field--theoretic model without local gauge symmetries on a curved spacetime is defined analogously with the only difference being that $\ipp{\cdot,\cdot}_\cV$ and $\ip{\cdot,\cdot}_\cV$ are not symmetric but antisymmetric. Complex theories can be obtained from the real ones by complexification.
\end{definition}

The relevance of the given data is as follows. Classical configurations $\Phi$ of the linear field model under consideration are smooth sections $\Phi\in\Se(\cV)$ of the vector bundle $\cV$. We recall that $\cV$ is locally of the form $M\times V$ with a real vector space $V$ which implies that locally $\Phi$ is a smooth function from $M$ to $V$, see e.g. \cite{chap2_Nakahara,chap2_Kobayashi} for background material on vector bundles. We shall denote by $\Sec(\cV), \Setc(\cV), \Sesc(\cV)$ the subspaces of $\Sec(\cV)$ consisting of smooth sections of $\cV$ with compact, timelike--compact and spacelike--compact support, respectively.%
\index{V@$\cV$, vector bundle}%
\index{Gamma(V)@$\Se(\cV)$, smooth sections of a vector bundle $\cV$}%
\index{Gamma0(V)@$\Sec(\cV)$, $\Setc(\cV)$, $\Sesc(\cV)$, smooth sections of a vector bundle $\cV$ with compact, timelike--compact, spacelike--compact support}%

The operator $P$ specifies the equation of motion for the field model, the formal self--adjointness of $P$ is motivated by the fact that equations of motion arising as Euler--Lagrange equations of a Lagrangean are generally given by a formally self--adjoint $P$. In fact the (formal) action $S(\Phi)=\frac12 \ip{\Phi,P\Phi}_\cV$ leads to the Euler--Lagrange equation $P\Phi=0$.

In the Klein--Gordon case we are dealing with an operator which is normally hyperbolic, i.e. the leading order term is of the form $-g^{\mu\nu}\partial_\mu\partial_\nu$. As reviewed in Theorem \ref{thm_FundamentalSolutionScalar}, this operator has a well--defined Cauchy problem, i.e. it is {\em Cauchy--hyperbolic}\index{Cauchy--hyperbolic operator}, and consequently unique advance and retarded fundamental solutions exist such that the operator is {\em Green--hyperbolic}. Example of partial differential operators which are Cauchy--hyperbolic, but not normally hyperbolic are the Dirac operator and the Proca operator which defines the equation of motion for a massive vector field, see e.g. \cite{chap2_Bar:2011iu}. On the other hand, the distinction between Cauchy--hyperbolic operators and Green--hyperbolic operators does not matter in most examples although one can construction operators which are Green--hyperbolic but not Cauchy--hyperbolic, cf. \cite{chap2_Bar:2011iu} for details.

Based on the data given in Definition \ref{def_linearnongauge}, a symplectic space (Bosonic case) or inner product space (Fermionic case) can be constructed in full analogy to the Klein--Gordon case, in particular, the following can be shown.

\begin{theorem}\label{thm_symplecticnongauge}Under the assumptions of Definition \ref{def_linearnongauge}, let $E^P\doteq E^P_R-E^P_A$ denote the causal propagator of $P$\index{EP@$E^P$, causal propagator of a Green--hyperbolic operator $P$}, and let $\Sol\subset \Se(\cV)$, $\Solsc \subset \Sesc(\cV)$ denote the space of smooth (smooth and spacelike--compact) solutions of $P\Phi=0$.
\index{Sol@$\Sol$, space of smooth solutions of the equation of motion}
\index{Solsc@$\Solsc$, space of smooth solutions of the equation of motion with spacelike--compact support}
\begin{enumerate}
\item The tuple $(\cE,\tau)$, where 
\index{Ecal@$\cE$, labelling space of linear (gauge--invariant) on--shell observables}
\index{tau@$\tau$, (pre--)symplectic or (pre--)inner product on $\cE$}
$$\cE\doteq \Sec(\cV)/P[\Sec(\cV)]\,,$$
$$\tau:\cE\times \cE\to \bR\,,\qquad \tau([f],[g])\doteq \ip{f,E^P g}_\cV\,,$$
is a well--defined symplectic (Bosonic case) or inner product (Fermionic case) space. In particular $\tau$ is well--defined and independent of the chosen representatives and moreover non--degenerate and antisymmetric (Bosonic case) or symmetric (Fermionic case).
\item Let $[f]\in\cE $ be arbitrary and let $\Sigma$ be any Cauchy surface of $(M,g)$. Then, for any bounded neighbourhood $\cO(\Sigma)$ of $\Sigma$, we can find a $g\in\Sec(\cV)$ with $\supp g\subset \cO(\Sigma)$ and $g\in [f]$.
\item The causal propagator $E^P:\Sec(\cV)\to\Sesc(\cV)$ descends to a bijective map $\cE\to\Solsc $ and for all $f\in\Sec(\cV)$ and all $\Phi\in\Sol$, $$\ip{f,\Phi}_\cV=\ip{E^P f,\Phi}_{\Sol}\,,$$ where for all $\Phi_1, \Phi_2\in\Sol$ with spacelike--compact overlapping support, the bilinear form $\ip{\cdot,\cdot}_{\Sol}$ is defined as 
$$\ip{\Phi_1,\Phi_2}_{\Sol}\doteq \ip{P\Phi^+_1,\Phi_2}_\cV\,.$$%
\item $\ip{\cdot,\cdot}_{\Sol}$ may be computed as a suitable integral over an arbitrary but fixed Cauchy surface $\Sigma$ of $(M,g)$ with future--pointing normal vector field $N$ and induced measure $d\Sigma$. If there exists a `current' $j:\Se(\cV)\times\Se(\cV)\to T^*M$ such that $\nabla^\mu j_\mu(\Phi_1,\Phi_2) = \langle\langle\Phi_1, P \Phi_2\rangle\rangle_\cV - \langle\langle\Phi_2, P \Phi_1\rangle\rangle_\cV$ for all $\Phi_1,\Phi_2\in\Se(\cV)$, then
$$\ip{\Phi_1,\Phi_2}_\Sol = \int\limits_\Sigma d\Sigma\; N^\mu j_\mu(\Phi_1,\Phi_2)\,.$$
\item The tuple $(\Solsc ,\ip{\cdot,\cdot}_{\Sol})$ is a well--defined symplectic (Bosonic case) or inner product (Fermionic case) space which is isomorphic to $(\cE,\tau)$.
\index{\$ipsol@$\ip{\cdot,\cdot}_{\Sol}$, bilinear form on $\Sol$}%
\end{enumerate}
\end{theorem}

As with the Klein--Gordon field, the last statement implies in physical terms that the symplectic respectively inner product space can be constructed both in a covariant and in an equal--time fashion, and that the two constructions give equivalent results. In many cases, the equal--time point of view is better suited for practical computations and for proving particular further properties of the bilinear form $\tau$, cf. the following discussion of theories with local gauge invariance.

\subsection{Models with Gauge Symmetry}
\label{sec:linear_gauge}

The discussion of linear field theoretic models with local gauge symmetries on curved spacetimes is naturally more involved than the case where such symmetries are absent. However, as in this monograph we will only be dealing with linear models and simple observables, it will not be necessary to introduce auxiliary fields like in the BRST/BV formalism \cite{chap2_Hollands:2007zg,chap2_Fredenhagen:2011mq}. Instead, we shall review an approach which has been developed in \cite{chap2_Dimock:1992ff} for the Maxwell field, used for linearised gravity in \cite{chap2_Fewster:2012bj} and then further generalised to arbitrary linear gauge theories in \cite{chap2_Hack:2012dm}. For linear models and simple observables this approach and the BRST/BV formalism give equivalent results, however, non--linear models and more general observables are not tractable in the way we shall review in the following.

\subsubsection{A Toy Model}

We outline the essential ideas of this approach at the example of a toy model. We consider as a gauge field $\Phi=(\phi_1,\phi_2)^t\subset \Se(\cV)$ a tuple of two scalar fields on a spacetime $(M,g)$ satisfying the equation of motion
$$P\Phi =\begin{pmatrix}-\Box & \phantom{-}\Box\\\phantom{-}\Box & -\Box\end{pmatrix}\begin{pmatrix}\phi_1\\\phi_2\end{pmatrix}=0\,,$$
where $\cV$ is the (trivial) vector bundle $\cV\doteq M\times \bR^2$.
The gauge transformations are given by the following translations on configuration space $\Phi\mapsto \Phi + K \epsilon$, where the gauge transformation operator $K:\Se(M)\to\Se(\cV)$ is the linear operator defined by $K \epsilon \doteq (\epsilon,\epsilon)^t$ for a smooth function $\epsilon\in\Se(M)$. One may check that $P\circ K=0$ holds which is equivalent to the gauge--invariance of the action $S(\Phi)\doteq \frac12\langle \Phi,P\Phi\rangle_\cV$ with $\langle \Phi,\Phi^\prime\rangle_\cV\doteq \int_M \vol ( \phi^{\phantom{\prime}}_1 \phi^\prime_1+\phi^{\phantom{\prime}}_2 \phi^\prime_2)$. 

Clearly, the linear combination $\psi\doteq \phi_1-\phi_2$ is gauge--invariant and satisfies $-\Box\psi=0$, and it would be rather natural to quantize $\Phi$ by directly quantizing $\psi$ as a massless, minimally coupled scalar field. This would be much in the spirit of the usual quantization of perturbations in Inflation, where gauge--invariant linear combinations of the gauge field components, e.g. the Bardeen--Potentials or the Mukhanov--Sasaki variable, are taken as the fundamental fields for quantization, see the last chapter of this monograph. However, in general it is rather difficult to directly identify a gauge--invariant fundamental field like $\psi$ whose classical and quantum theory is equivalent to the classical and quantum theory of the original gauge field. Notwithstanding, an indirect characterisation of such a gauge--invariant linear combination of gauge--field components, which can serve as a fundamental field for quantization, is still possible. In the toy model under consideration we consider a tuple  $f=(f_1,f_2)\in\Sec(\cV)$ of test functions $f_i\in\Sec(M)$. We ask that $K^\dagger f\doteq f_1+f_2=0$, where $K^\dagger:\Se(\cV)\to\Se(M)$ is the adjoint of the gauge transformation operator $K$ i.e. $\int_M \vol  \epsilon K^\dagger f=\langle K\epsilon, f\rangle_\cV$. Clearly, any $f$ satisfying this condition is of the form $f=(h,-h)^t$ for a test function $h$. We now observe that the pairing between a gauge field configuration $\Phi$ and such an $f$ is gauge--invariant, i.e. $\langle \Phi + K\epsilon, f\rangle_\cV = \langle f,\Phi\rangle_\cV + \int_M \vol  \epsilon K^\dagger f=\langle \Phi, f\rangle_\cV$. Thus we can consider the `smeared field' $\Phi(f)\simeq \langle f,\Phi\rangle_\cV$, with $f=(h,-h)^t$ and arbitrary $h$, as a gauge--invariant linear combination of gauge--field components which is suitable for playing the role of a fundamental field for quantization. We can compute $\langle f,\Phi\rangle_\cV=\int_M\vol \psi h$, and observe that, up to the `smearing' with $h$, this indirect choice of gauge--invariant fundamental field is exactly the one discussed in the beginning. If one chooses $h$ to be the delta distribution $\delta(x,y)$ rather than a test function, one even finds $\langle f,\Phi\rangle=\psi(x)$, whereas for general $h$, $\langle f,\Phi\rangle$ can be interpreted as a weighted, gauge--invariant measurement of the field configuration $\Phi$. Moreover, as already anticipated, in general gauge theories with more complicated gauge transformation operators $K$ it it usually extremely difficult to classify all solutions of $K^\dagger f=0$, which would be equivalent to a direct characterisation of one or several fundamental gauge--invariant fields such as $\psi$, whereas working implicitly with the condition $K^\dagger f=0$ is always possible.

\subsubsection{General Models}

From the previous discussion we can already infer most of the additional data which is needed in addition to the data mentioned in Definition \ref{def_linearnongauge} in order to specify a linear field theoretic model with local gauge symmetries on curved spacetimes.

\begin{definition}\label{def_lineargauge}
A real Bosonic linear field--theoretic model with local gauge symmetries on a curved spacetime is defined by the data $(M,\cV,\cW,P,K)$, where
\begin{enumerate}
\item $M\simeq (M,g)$ is a globally hyperbolic spacetime,
\item $\cV$ and $\cW$ are real vector bundles over $M$, the spaces of smooth sections $\cV$ and $\cW$ are endowed with symmetric and non--degenerate bilinear forms $\ip{\cdot,\cdot}_\cV$ and $\ip{\cdot,\cdot}_\cW$ which are well--defined on sections with compact overlapping support and given by the integral of  fibrewise symmetric and non--degenerate bilinear forms $\ipp{\cdot,\cdot}_\cV:\Se(\cV)\times\Se(\cV)\to\Se(M)$ and $\ipp{\cdot,\cdot}_\cW:\Se(\cW)\times\Se(\cW)\to\Se(M)$,
\index{Wcal@$\cW$, vector bundle}
\item $P:\Se(\cV)\to\Se(\cV)$ is a partial differential operator which is formally self--adjoint with respect to $\ip{\cdot,\cdot}_\cV$,
\item $K:\Se(\cW)\to\Se(\cV)$\index{K@$K$, gauge--transformation operator} is a partial differential operator such that $P\circ K=0$; moreover $R\doteq K^\dagger\circ K:\Se(\cW)\to\Se(\cW)$ is Cauchy--hyperbolic and there exists an operator $T:\Se(\cW)\to\Se(\cV)$\index{T@$T$, gauge--fixing operator} such that a) $\widetilde P\doteq P+T\circ K^\dagger$\index{Ptilde@$\widetilde{P}$, gauge--fixed equation of motion operator} is Green--hyperbolic and b) $Q\doteq K^\dagger \circ T$ is Cauchy--hyperbolic.
\end{enumerate}
A real Fermionic linear field--theoretic model with local gauge symmetries on a curved spacetime is defined analogously with the only difference being that $\ipp{\cdot,\cdot}_\cV$, $\ip{\cdot,\cdot}_\cV$, $\ipp{\cdot,\cdot}_\cW$ and $\ip{\cdot,\cdot}_\cW$ are not symmetric but antisymmetric. Complex theories can be obtained from the real ones by complexification.
\end{definition}

These data have the following meaning. Sections of $\cV$ are configurations of the gauge field $\Phi$, whereas local gauge transformations are parametrised via the gauge transformation operator $K$ by sections of $\cW$. The differential operator $P$ defines the equation of motion for the gauge field $\Phi$ via $P\Phi=0$. The formal self--adjointness of $P$ is motivated by $P$ being the Euler--Lagrange operator of a local action $S(\Phi)$, e.g. $S(\Phi)=\frac12 \langle \Phi,P\Phi\rangle_\cV$, whereas the gauge--invariance condition $P\circ K=0$ implies gauge--invariance of the action $S(\Phi)$. This condition implies (for $K\neq 0$) that $P$ can {\em not} be Cauchy--hyperbolic, because any `pure gauge configuration' $\Phi_\epsilon=K\epsilon$ with $\epsilon\in\Sec(\cW)$ of compact support solves the equation of motion $P\Phi_\epsilon=0$ with vanishing initial data in the distant past, whereas for Cauchy--hyperbolic $P$ the unique solution with vanishing initial data is identically zero.

The Cauchy--hyperbolicity of $R=K^\dagger\circ K$ implies that for every $\Phi\in\Se(\cV)$ there exists an $\epsilon\in\Se(\cW)$ such that $\Phi^\prime \doteq \Phi+K\epsilon$ satisfies the `canonical gauge--fixing condition' $K^\dagger \Phi^\prime =0$. The existence of the gauge--fixing operator $T$ such that the gauge--fixed equation of motion operator $\widetilde P= P+T\circ K^\dagger$ is Green--hyperbolic implies that every solution of $P\Phi=0$ in fact satisfies $\widetilde P \Phi=0$ up to gauge--equivalence; consequently, the dynamics of the `physical degrees of freedom' is ruled by a hyperbolic equation of motion even if $P$ is not hyperbolic. Finally, the condition that $Q=K^\dagger \circ T$ is Cauchy--hyperbolic implies that the gauge--fixing $K^\dagger \Phi=0$ is compatible with the hyperbolic dynamics of $\widetilde P \Phi=0$. $T$ is in general not canonical and the following constructions will not depend on the particular choice of $T$ in case several $T$ with the required properties exist, thus we do not consider the gauge--fixing operator $T$ as part of the data specifying the model.

Apart from the toy model discussed above, a simple example of a linear gauge theory which fits into Definition \ref{def_lineargauge} is the Maxwell field (on a trivial principal $U(1)$--bundle) which after all was the inspiration for the formulation of this definition. This model is specified by (in differential form notation)

$$\cW=M\times \bR\,,\qquad \cV=\cW\otimes T^*M = T^*M\,,$$
$$\ip{\epsilon_1, \epsilon_2}_\cW \doteq  \int\limits_M \epsilon_1\wedge
*\epsilon_2\,,\qquad \ip{\Phi_1, \Phi_2}_\cV \doteq  \int\limits_M \Phi_1\wedge
*\Phi_2\,,$$
$$P=d^\dagger d\,,\qquad K=T= d\,,$$
$$\widetilde P = d^\dagger d + d d^\dagger\,, \qquad K^\dagger K = K^\dagger T = d^\dagger d = \Box\,.$$

We would like to construct a (pre--)symplectic or (pre--)inner product space corresponding to the data given in Definition \ref{def_lineargauge} by following as much as possible the logic of the case without gauge symmetry. To this avail we need a few further definitions of section spaces.

\begin{definition}\label{def_sectionspacesgauge}
As before, we denote by $\Sol\subset\Se(\cV)$ ($\Solsc\subset\Sesc(\cV)$) the spaces of smooth solutions of the equation $P\Phi=0$ (with spacelike--compact support). By $\cG$ and $\cGsc$ we denote the space of gauge configurations (with spacelike--compact support), by $\cGscc$ we denote the gauge configurations induced by spacelike--compact gauge transformation parameters
$$\cG\doteq K\left[\Se(\cW)\right]\,,\qquad \cGsc\doteq \cG\cap \Sesc(\cW)\,,\qquad \cGscc\doteq K\left[\Sesc(\cW)\right]\,.$$
In general, $\cGscc\subsetneq\cGsc$. By $\ker_0(K^\dagger)$ we denote the space of gauge--invariant test--sections and by $\cE$ the {\em labelling space of linear gauge--invariant on--shell observables}
$$\ker_0(K^\dagger)\doteq \{f\in \Sec(\cV)\,|\,K^\dagger f =0\}\,,\qquad \cE\doteq \ker_0(K^\dagger)/P\left[\Sec(\cV)\right]\,.$$

\index{ker0@$\ker_0(K^\dagger)$, compactly supported kernel of $K^\dagger$}
\index{Ecal@$\cE$, labelling space of linear (gauge--invariant) on--shell observables}
\index{Gcal@$\cG$, space of gauge configurations}
\index{Gcalsc@$\cGsc$, space of gauge configurations with spacelike--compact support}
\index{Gcalscc@$\cGscc$, space of gauge configurations induced by spacelike--compact gauge transformation parameters}
\index{Sol@$\Sol$, space of smooth solutions of the equation of motion}
\index{Solsc@$\Solsc$, space of smooth solutions of the equation of motion with spacelike--compact support}
\end{definition}

Our discussion of the toy model in the previous subsection already indicated why $\cE$ defined above is a good candidate for a labelling space of linear gauge--invariant on--shell observables. First of all we observe that $\cE$ is well--defined because $P\left[\Sec(\cV)\right]\subset \ker_0(K^\dagger)$ owing to $P\circ K=0$. Moreover, we have by construction for arbitrary $\Phi\in\Sol$, $\epsilon\in\Se(\cW)$, $f\in\ker_0(K^\dagger)$ and $g\in\Sec(\cV)$
$$\ip{f+Pg,\Phi+K\epsilon}_\cV = \ip{f,\Phi}_\cV\,.$$
Consequently, every element $[f]$ of $\cE$ induces a well--defined linear functional on $\Sol/\cG$, i.e. on gauge--equivalence classes of on--shell configurations, by $\Sol/\cG\ni[\Phi]\mapsto \cO_{[f]}([\Phi])\doteq \ip{f,\Phi}_\cV$. Being gauge--invariant, these functionals correspond to meaningful (physical) observables. On the level of classical observables, the fact that the physical degrees of freedom of the gauge field propagate in a causal fashion is reflected in the following generalisation of Lemma \ref{lem_TimeSliceClassicalScalar} which is proved in \cite{chap2_Hack:2012dm}.

\begin{lemma}
\label{lem_TimeSliceClassicalGauge}
Let $[f]\in\cE $ be arbitrary and let $\Sigma$ be any Cauchy surface of $(M,g)$. Then, for any bounded neighbourhood $\cO(\Sigma)$ of $\Sigma$, we can find a $g\in\ker_0(K^\dagger)$ with $\supp g\subset \cO(\Sigma)$ and $g\in [f]$.
\end{lemma}

In constructing the classical bracket for models without gauge symmetry the last statement of Theorem \ref{thm_FundamentalSolutionScalar}, which in fact holds for the causal propagator $E^P$ of any Green--hyperbolic operator $P$ on an arbitrary vector bundle $\cV$, has been crucial. In the following, we review results obtained in \cite[Theorem 3.12+Theorem 5.2]{chap2_Hack:2012dm} which essentially imply that, although $P$ is not hyperbolic, the causal propagator $E^{\widetilde{P}}$ of the gauge--fixed equation of motion operator $\widetilde{P}=P+T\circ K^\dagger$ is effectively a causal propagator for $P$ up to gauge--equivalence. The crucial observation here is that $P$ and $\widetilde{P}$ coincide on $\ker_0(K^\dagger)$ which implies that $E^{\widetilde{P}}$ restricted to $\ker_0({K}^\dagger)$ is independent of the particular form of the gauge fixing operator $T$.

\begin{theorem}\label{prop_prop_G} The causal propagator $E^{\widetilde{P}}$ of $\widetilde{P}=P+T\circ {K}^\dagger$ satisfies the following relations.
\begin{enumerate}
\item $h\in\ker_0({K}^\dagger)$ and $E^{\widetilde{P}}h\in\cGscc$ if and only if $h\in P[\Sec(\cV)]$, with $\cGscc$ defined in Definition \ref{def_sectionspacesgauge}.
\item Every $h\in\Solsc$ 
can be split as $h=h_1+h_2$ with $h_1\in E^{\widetilde{P}}\left[\ker_0({K}^\dagger)\right]$ and $h_2\in \cGscc$.
\item $E^{\widetilde{P}}$ descends to a bijective map $\cE\to \Solsc/\cGscc$.
\item $E^{\widetilde{P}}$ is formally skew--adjoint w.r.t. $\ip{\cdot,\cdot}_{\cV} $ on $\ker_0({K}^\dagger)$, i.e. $$\ip{ h_1,E^{\widetilde{P}}h_2}_\cV =-\ip{ E^{\widetilde{P}}h_1,h_2}_\cV$$ for all $h_1,h_2\in\ker_0({K}^\dagger)$.
\item Let $T^\prime:\Se(\cW)\to\Se(\cV)$ be any differential operator satisfying the properties required for the operator $T$ in Definition \ref{def_lineargauge} and let $E^{\widetilde{P}^\prime}$ be the causal propagator of $\widetilde{P}^\prime\doteq P+T^\prime\circ {K}^\dagger$. Then $E^{\widetilde{P}^\prime}$ satisfies the four properties above.
\end{enumerate}
\end{theorem}

Given these results, we can now construct a meaningful bracket on $\cE$ by generalising Proposition \ref {prop_presymplectic}.

\begin{proposition}\label{prop_presymplectic_gauge} The tuple $(\cE ,\tau)$ with $\tau:\cE \times\cE \to\bR$ defined by 
$$\tau([f],[g])\doteq \ip{f, E^{\widetilde{P}} g}_\cV$$
\index{tau@$\tau$, (pre--)symplectic or (pre--)inner product on $\cE$}%
is a {\em pre--symplectic space} (Bosonic case) or {\em pre--inner product space} (Fermionic case). In particular, 
\begin{enumerate}
\item $\tau$ is well--defined and independent of the chosen representatives,
\item $\tau$ is antisymmetric (Bosonic case) or symmetric (Fermionic case).
\item Let $T^\prime:\Se(\cW)\to\Se(\cV)$ be any differential operator satisfying the properties required for the operator $T$ in Definition \ref{def_lineargauge} and define $\tau^{\prime}$ in analogy to $\tau$ but with the causal propagator $E^{{\widetilde{P}}^\prime}$ of $\widetilde{P}^\prime\doteq P+T^\prime\circ {K}^\dagger$ instead of $E^{{\widetilde{P}}}$. Then $\tau^{\prime}=\tau$.
\end{enumerate}
\end{proposition}

We stress that $\tau$ is in general {\em not} weakly non--degenerate, cf. the last statement of Theorem \ref{thm_proptau2}. This is a particular feature of gauge theories, cf. e.g. \cite{chap2_Sanders:2012sf,chap2_Benini:2013ita} for a discussion of the physical interpretation of this non--degeneracy in the case of the Maxwell field. The last statement above indicates that $\tau$ is independent of the gauge--fixing operator $T$ and in this sense, gauge--invariant. Indeed, we shall see in what follows that $\tau$ can be rewritten in a manifestly gauge--invariant form. The form of $\tau$ given here can be derived directly from the action $S(\Phi)=\frac12 \ip{\Phi,P\Phi}_\cV$ by Peierls' method  in analogy to the derivation for electromagnetism in \cite{chap2_Sanders:2012sf}, see also \cite{chap2_Khavkine:2012jf,chap2_Khavkine:2014kya}) for a broader context.

As in the case without local gauge symmetries it is interesting and useful to observe that the covariant pre--symplectic or pre--inner product space $(\cE,\tau)$ can be understood equivalently in an equal--time fashion. In fact the following statements have been proved in \cite[Proposition 5.1+Theorem 5.2]{chap2_Hack:2012dm} (or can be proved by slightly generalising the arguments used there).

\begin{theorem}\label{thm_proptau2}Under the assumptions of Definition \ref{def_lineargauge}, let $\ip{\cdot,\cdot}_{\Sol}$ be the bilinear form on $\Sol$ defined for $\Phi_1, \Phi_2\in\Sol$ with spacelike--compact overlapping support by
$$\ip{\Phi_1,\Phi_2}_\Sol \doteq \ip{P\Phi^+_1,\Phi_2}_\cV\,,$$
\index{\$ipsol@$\ip{\cdot,\cdot}_{\Sol}$, bilinear form on $\Sol$}%
where $\Phi^+$ denotes the future part of $\Phi$, see Definition \ref{def_futurepart}. This bilinear form has the following properties.
\begin{enumerate}
\item $\ip{\Phi_1,\Phi_2}_{\Sol}$ is well--defined for all $\Phi_1, \Phi_2\in\Sol$ with spacelike--compact overlapping support. In particular, it is independent of the choice of future part entering its definition.
\item $\ip{\cdot, \cdot}_{\Sol}$ is antisymmetric (Bosonic case) or symmetric (Fermionic case).
\item $\ip{\cdot, \cdot}_{\Sol}$ is gauge--invariant, i.e. 
$$\ip{ \Phi_1, \Phi_2+K\epsilon }_\Sol=\ip{\Phi_1, \Phi_2}_{\Sol}$$
for all $\Phi_1,\Phi_2\in\Sol$, $\epsilon\in \Se(\cW)$ s.t. $\Phi_1$ and $\epsilon$ have spacelike--compact overlapping support.
\item $\ip{\cdot,\cdot}_{\Sol}$ may be computed as a suitable integral over an arbitrary but fixed Cauchy surface $\Sigma$ of $(M,g)$ with future--pointing normal vector field $N$ and induced measure $d\Sigma$. If there exists a `current' $j:\Se(\cV)\times\Se(\cV)\to T^*M$ such that $\nabla^\mu j_\mu(\Phi_1,\Phi_2) = \langle\langle\Phi_1, P \Phi_2\rangle\rangle_\cV - \langle\langle\Phi_2, P \Phi_1\rangle\rangle_\cV$ for all $\Phi_1,\Phi_2\in\Se(\cV)$, then
$$\ip{\Phi_1,\Phi_2}_\Sol = \int\limits_\Sigma d\Sigma\; N^\mu j_\mu(\Phi_1,\Phi_2)\,.$$
\item For all $\Phi\in\Sol$ and all $h\in\ker_0(K^\dagger)$, 
$$\ip{  E^{\widetilde{P}}h,\Phi}_{\Sol}=\ip{h,\Phi}_\cV \,.$$
\item $E^{\widetilde{P}}$ descends to an isomorphism of pre--symplectic (Bosonic case) or pre--inner product spaces (Fermionic case) $E^{\widetilde{P}}:(\cE,\tau)\to (\Solsc /\cGscc,\langle\cdot,\cdot\rangle_{\Sol})$.
\item If $\cGscc\subsetneq\cGsc$, then $\tau$ is degenerate, i.e. there exists $[0]\neq[h]\in\cE$ s.t. $\tau([h],\cE)=0$.
\end{enumerate}
\end{theorem}

The fourth and fifth statement in the above theorem show that one can view the observable $\Sol/\cG\ni[\Phi]\mapsto \ip{h,\Phi}_\cV\simeq \Phi(h)$, i.e. the `covariantly smeared classical field', equivalently as an `equal--time smeared classical field' $\Sol/\cG\ni[\Phi]\mapsto \ip{H,\Phi}_\Sol$ with $H=E^{\widetilde{P}}h\in[H]\in \Solsc /\cGscc$.

\section{Linear Quantum Fields on Curved Spacetimes}
\label{sec:linear_quant}

Given the (pre--)symplectic or (pre--)inner product spaces of classical linear observables constructed in the previous section for Bosonic and Fermionic theories with or without local gauge symmetries, there are several `canonical' ways to construct corresponding algebras of observables in the associated quantum theories; these constructions differ mainly in technical terms. 

In the Bosonic case, one can consider the {\em Weyl algebra}\index{Weyl algebra} corresponding to the pre--symplectic space $(\cE,\tau)$, which essentially means to quantize exponentials $\exp(i\ip{f,\Phi}_\cV)$ of the smeared classical field $\Phi(f)\simeq \ip{f,\Phi}_\cV$ rather than the smeared classical field itself. This mainly has the technical advantage that one is dealing with a $C^*$--algebra corresponding to bounded operators on a Hilbert space, i.e. to operators with a bounded spectrum. However, sometimes it is also advisable in physical terms to consider exponential observables rather than linear ones as fundamental building blocks, e.g. in case one is dealing with {\em finite} gauge transformations, cf. \cite{chap2_Benini:2013ita}. The Weyl algebra is initially constructed under the assumption that the form $\tau$ is non--degenerate such that the the pre--symplectic space $(\cE,\tau)$ is in fact symplectic. However, the construction of the Weyl algebra is also well--defined in the degenerate case\cite{chap2_degenerateCCR} and even in case $\cE$ is not a vector space, but only an Abelian group \cite{chap2_Benini:2013ita}. In the Fermionic case it is not necessary to consider exponential observables in order to access the advantages of a $C^*$--algebraic framework, as a $C^*$--algebra can already be constructed based on linear observables due to the anticommutation relations of Fermionic quantum fields, see e.g. \cite{chap2_Bratteli2}.

As in this monograph we shall not make use of the $C^*$--algebraic framework, it will be sufficient to consider the algebra of quantum observables constructed by directly quantizing the smeared classical fields $\Phi(f)\simeq \ip{f,\Phi}_\cV$ themselves both in the Bosonic case and in the Fermionic case. We shall first construct the {\em Borchers--Uhlmann algebra} $\cA(M)$ corresponding to the linear model defined by the data $(M,\cV,\cW,P,K)$ cf. Definition \ref{def_lineargauge}, where we can consider a model $(M,\cV,P)$ without gauge invariance as the subclass $(M,\cV,\cW=\cV,P,K=0)$. The algebra $\cA(M)$ contains only the most simple observables, in particular it does not contain observables which correspond to pointwise powers of the field such as $\Phi(x)^2$ or the stress--energy tensor. We shall discuss a larger algebra which contains also these observables in Section \ref{sec:wick}. 

In order to construct the Borchers--Uhlmann algebra, we recall that the labelling space $\cE$ consists of equivalence classes of test sections $[f]$ corresponding to the classical observables $\Phi(f)\simeq \ip{f,\Phi}_\cV$ with $\Phi\in[\Phi]\in\Sol/\cG$ ($=\Sol$ in the absence of gauge symmetries, i.e. for $K=0$) and $f\in[f]\in \cE=\ker_0(K^\dagger)/P[\Sec(\cV)]$ ($=\Sec(\cV)/P[\Sec(\cV)]$ for $K=0$). We recall that omitting the equivalence classes in the notation $\Phi(f)\simeq \ip{f,\Phi}_\cV$ is meaningful because the latter expression is independent of the chosen representatives of the equivalence classes. For the quantum theory, we need complex expressions and therefore consider the complexification $\cE_\bC\doteq \cE\otimes_\bR \bC$ of the labelling space $\cE$.

With this in mind, we represent the quantum product of two smeared fields $\Phi(f_1)\Phi(f_2)$ by the tensor product $[f_1]\otimes [f_2]$. On the long run, we would like to represent $\Phi(f)$ as an operator on a Hilbert space, we therefore need an operation which encodes `taking the adjoint with respect to a Hilbert space inner product' on the abstract algebraic level. We define such a $*$--operation by setting $[\Phi(f)]^*\doteq\Phi(\overline{f})$, corresponding to $[f_1]^*\doteq [\overline{f_1}]$, and $[\Phi(f_1)\cdots\Phi(f_n)]^*=[\Phi(f_n)]^*\cdots[\Phi(f_1)]^*$. Observables would then be polynomials $\cP$ of smeared fields (tensor polynomials of elements of $\cE_\bC$) which fulfil $\cP^*=\cP$ (in the Fermionic case we need further conditions e.g. $\cP$ has to be an even polynomial). To promote $\Phi(f)$ to a proper quantum field with the correct (anti)commutation relations, we define
\beq\label{eq_CCR}\left[\Phi(f),\Phi(g)\right]_\mp\doteq\Phi(f)\Phi(g)\mp\Phi(g)\Phi(f)=i\tau([f],[g])\bI=i\ip{f,E^{\widetilde{P}},g}_\cV\bI\,,\eeq
where we recall that $E^{\widetilde{P}}$ is the causal propagator of the gauge--fixed equation of motion operator $\widetilde P = P + T\circ K^\dagger$ ($=P$ in the absence of gauge symmetries $K=0$), $\bI$ is the identity and the $-$ ($+$) sign applies for the Bosonic (Fermionic) case. Recall that $\tau([f],[g])$ vanishes if the supports of $f$ and $g$ are spacelike separated, the above {\em canonical (anti)commutation relations} (CCR/CAR)\index{CCR, canonical commutation relations}\index{CAR, canonical anticommutation relations}\index{canonical (anti)commutation relations} therefore assure that observables commute at spacelike separations.
In the case without gauge symmetries, one can write the CCR/CAR formally as $$\left[\Phi(x),\Phi(y)\right]_\mp=iE^P(x,y)\bI\,.$$ Recall that this is nothing but the covariant version of the well--known equal--time CCR/CAR. In the case where gauge symmetries are present, this expression does not make sense because the equality holds only when smeared with `gauge--invariant' test sections $f\in\ker_0(K^\dagger)$. Finally, we remark that dynamics is already encoded by the fact that $\cE_\bC$ consists of equivalence classes with $[Pg]=[0]\in\cE_\bC$ for all $g\in\Sec^\bC(\cV)$ and that it is convenient to have a topology on the algebra $\cA(M)$ in order to be able to quantify to which extent two abstract observables are `close', i.e. similar in physical terms. We subsume the above discussion in the following definition.

\begin{definition}
 \label{def_BorchersUhlmannScalar}\index{Borchers--Uhlmann algebra}
 Consider a linear Bosonic or Fermionic (gauge) field theory defined by $(M,\cV,\cW,P,K)$ (with $\cW=\cV$ and $K=0$ in the absence of gauge symmetries), cf. Definition \ref{def_linearnongauge} and Definition \ref{def_lineargauge} and let $(\cE,\tau)$ be the corresponding (pre--)symplectic or (pre--)inner product space constructed as in Theorem \ref{thm_symplecticnongauge} and Proposition \ref{prop_presymplectic_gauge}. The {\em Borchers--Uhlmann algebra} $\cA(M)$ of the model $(M,\cV,\cW,P,K)$ is defined as
$$\cA(M)\doteq \cA_0(M)/\cI \,,$$ where $\cA_0(M)$ is the direct sum
$$\cA_0(M)\doteq\bigoplus\limits_{n=0}^\infty \cE_\bC^{\otimes n}$$ ($\cE_\bC^{\otimes 0}\doteq \bC$) equipped with a product defined by the linear extension of the tensor product of $\cE_\bC^{\otimes n}$, a $*$--operation defined by the antilinear extension of $([f_1]\otimes\cdots\otimes [f_n])^*=[\overline{f_n}]\otimes\cdots\otimes [\overline{f_1}]$, and it is required each element of $\cA_0(M)$ is a linear combination of elements of $\cE_\bC^{\otimes n}$ with $n\le n_\mathrm{max}<\infty$. Additionally, we equip $\cA_0(M)$ with the topology induces by the locally convex topology of $\Sec(\cV)$. Moreover, $\cI$ is the closed $*$--ideal generated by elements of the form $-i\tau([f],[g])\oplus([f]\otimes [g]\mp [g]\otimes [f])$, where $-$ ($+$) stands for the Bosonic (Fermionic) case, and $\cA(M)$ is thought to be equipped with the product, $*$--operation, and topology descending from $\cA_0(M)$. 
\
If $\cO$ is an open subset of $M$, $\cA(\cO)$ denotes the algebra obtained by allowing only test sections with support in $\cO$.
\end{definition}

$\cA(M)$, in contrast to $\cA_0(M)$, depends explicitly on the metric $g$ of a spacetime $(M,g)$ via the causal propagator and the equation of motion. However, now and in the following we shall omit this dependence in favour of notational simplicity.

We recall that, by Lemma \ref{lem_TimeSliceClassicalGauge}, every equivalence class $[f]\in\cE_\bC$ is so large that it contains elements with support in an arbitrarily small neighbourhood of any Cauchy surface of $M$. This implies the following well--known result, which on physical grounds entails the predictability of observables.

\begin{lemma}
 \label{lem_TimeSliceAxiomQuantumScalar}
The Borchers--Uhlmann algebra $\cA(M)$ fulfils the {\em time--slice axiom}\index{time--slice axiom}. Namely, let $\Sigma$ be a Cauchy surface of $(M,g)$ and let $\cO$ be an arbitrary neighbourhood of $\Sigma$. Then $\cA(\cO)$=$\cA(M)$.
\end{lemma}

We now turn our attention to states\index{state}\index{w@$\omega$, (algebraic) state}. Let $\gA$ be a topological, unital $*$--algebra, i.e. $\gA$ is endowed with an operation $^*$ which fulfils $(AB)^*=B^*A^*$ and $(A^*)^*=A$ for all elements $A,B$ in $\gA$. A {\em state} $\omega$ on $\gA$ is defined to be a continuous linear functional $\gA\to\bC$ which is normalised, i.e. $\omega(\bI)=1$ and positive, namely, $\omega(A^*A)\ge 0$ must hold for any $A\in\gA$. Considering the special topological and unital $*$--algebra $\cA(M)$, a state on $\cA(M)$ is determined by its $n$--point correlation functions
$$\omega_n(f_1,\cdots,f_n)\doteq\omega\left(\Phi(f_1)\cdots\Phi(f_n)\right)\,.$$ 
which are distributions in $\Sec^{\prime\bC}(M^n)$. Given a state $\omega$ on $\gA$, one can represent $\gA$ on a Hilbert space $\cH_\omega$ by the so--called {\em GNS construction}\index{GNS construction} (after Gel'fand, Naimark, and Segal), see for instance \cite{chap2_Haag,chap2_Araki}. By this construction, algebra elements $A\in\gA$ are represented as operators $\pi_\omega(A)$ on a common dense and invariant subspace of $\cH_\omega$, while $\omega$ is represented as a vector of $|\Omega_\omega\rangle\in\cH_\omega$ such that for all $A\in\gA$
$$\omega(A) = \langle \Omega_\omega|\pi_\Omega(A)|\Omega_\omega\rangle\,.$$
Conversely, every vector in a Hilbert space $\cH$ gives rise to an algebraic state on the algebra of linear operators on $\cH$.

Among the possible states on $\cA(M)$ there are several special classes, which we collect in the following definition. Some of the definitions are sensible for general $*$--algebras, as we point out explicitly.

\begin{definition}
\label{def_StatesScalar}
Let $\gA$ denote a general $*$--algebra and let $\cA(M)$ denote the Borchers--Uhlmann algebra of a linear Bosonic or Fermionic (gauge) field theory.
\begin{enumerate}
\item A state $\omega$ on $\gA$ is called {\em mixed}\index{state!mixed}, if it is a convex linear combination of states, i.e. $\omega=\lambda\omega_1+(1-\lambda)\omega_2$, where $\lambda<1$ and $\omega_i\neq \omega$ are states on $\gA$. A state is called {\em pure}\index{state!pure} if it is not mixed.
\item A state $\omega$ on $\cA(M)$ is called {\em even}, if it is invariant under $\Phi(f)\mapsto -\Phi(f)$, i.e. it has vanishing $n$--point functions for all odd $n$.
\item An even state on $\cA(M)$ is called {\em quasifree} or {\em Gaussian}\index{state!quasifree}\index{state!Gaussian} if, for all even $n$,
$$\omega_n(f_1, \cdots, f_n)=\sum\limits_{\pi_n\in S^\prime_n}\prod\limits_{i=1}^{n/2}\omega_2\left(f_{\pi_n(2i-1)},f_{\pi_n(2i)}\right)\,.$$
Here, $S^\prime_n$ denotes the set of ordered permutations of $n$ elements, namely, the following two conditions are satisfied for $\pi_n\in S^\prime_n$:
$$
\pi_n(2i-1)<\pi_n(2i)\quad\text{for} \quad 1\leq i\leq n/2\,,
$$
$$
\pi_n(2i-1)<\pi_n(2i+1)\quad\text{for}  \quad 1\leq i <  n/2\,.
$$
\item Let $\alpha_t$ denote a one-parameter group of $*$--automorphisms on $\gA$, i.e. for arbitrary elements $A$, $B$ of $\gA$, $$\alpha_t\left(A^*B\right)=\left(\alpha_t(A)\right)^*\alpha_t(B)\,,\qquad \alpha_t\left(\alpha_s(A)\right)=\alpha_{t+s}(A)\,,\qquad \alpha_0(A)=A\,.$$ A state $\omega$ on $\gA$ is called {\em $\alpha_t$-invariant} if $\omega(\alpha_t(A))=\omega(A)$ for all $A\in\gA$.
\item An $\alpha_t$-invariant state $\omega$ on $\gA$ is said to satisfy the {\em KMS condition} for an inverse temperature $\beta=T^{-1}>0$ if, for arbitrary elements $A$, $B$ of $\gA$, the two functions
$$F_{AB}(t)\doteq \omega\left(B\alpha_t(A)\right)\,,\qquad G_{AB}(t)\doteq \omega\left(\alpha_{t}(A)B\right)$$
extend to functions $F_{AB}(z)$ and $G_{AB}(z)$ on the complex plane which are analytic in the strips $0<\text{Im}\; z< \beta$ and $-\beta<\text{Im}\; z< 0$ respectively, continuous on the boundaries $\text{Im}\;z\in\{0,\beta\}$, and fulfil
$$F_{AB}(t+i\beta)=G_{AB}(t)\,.$$
\end{enumerate}
\end{definition}

The KMS condition (after Kubo, Martin, and Schwinger) holds naturally for Gibbs states of {\em finite} systems in quantum statistical mechanics, i.e. for states that are given as $\omega_\beta(A)=\Tr \rho A$ with a density matrix $\rho=\exp{(-\beta H)} (\Tr \exp{(-\beta H)})^{-1}$, $H$ the Hamiltonian operator of the system, and $\Tr$ denoting the trace over the respective Hilbert space. This follows by setting $$\alpha_t(A)=e^{itH}Ae^{-itH}\,,$$ making use of the cyclicity of the trace, and considering that $\exp{(-\beta H)}$ is bounded and has finite trace in the case of a finite system. In the thermodynamic limit, $\exp{(-\beta H)}$ does not possess these properties any more, but the authors of \cite{chap2_HHW} have shown that the KMS condition is still a reasonable condition in this infinite-volume limit. Physically, KMS states are states which are in (thermal) equilibrium with respect to the time evolution encoded in the automorphism $\alpha_t$. In general curved spacetimes, there is no `time evolution' which acts as an automorphism on $\cA(M)$. One could be tempted to introduce a time evolution by a canonical time-translation with respect to some time function of a globally hyperbolic spacetime. However, the causal propagator $E^P$ will in general not be invariant under this time translation if the latter does not correspond to an isometry of $(M,g)$. Hence, such time--translation would not result in an automorphism of $\cA(M)$. There have been various proposals to overcome this problem and to define generalised notions of thermal equilibrium in curved spacetimes, see \cite{chap2_VerchRegensburg} for a review. Ground states (vacuum states) may be thought of as KMS states with inverse temperature $\beta = T^{-1} = \infty$.

Quasifree or Gaussian states which are in addition pure are closely related to the well--known Fock space picture in the sense that the the GNS representation of a pure quasifree state is an irreducible representation on Fock space, see e.g. \cite[Section 3.2]{chap2_Kay}. In this sense, a pure quasifree state is in one--to--one correspondence to a specific definition of a `particle'.   

For the remainder of this chapter and most of the remainder of this monograph, the only model we shall discuss is the free neutral Klein--Gordon field for simplicity. However, all concepts we shall review will be applicable to more general models.

\section{Hadamard States}
\label{sec:hadamard}

The power of the algebraic approach lies in its ability to separate the algebraic relations of quantum fields from the Hilbert space representations of these relations and thus in some sense to treat {\it all} possible Hilbert space representations at once. However, the definition of an algebraic state reviewed in the previous subsection is too general and thus further conditions are necessary in order to select the physically meaningful states among all possible ones on $\cA(M)$.

To this avail it seems reasonable to look at the situation in Minkowski spacetime. Physically interesting states there include the Fock vacuum state and associated multiparticle states as well as coherent states and states describing thermal equilibrium situations. All these states share the same ultraviolet (UV) properties, {\it i.e.}~the same high-energy behaviour, namely they satisfy the so--called {\em Hadamard condition}, which we shall review in a few moments. A closer look at the formulation of quantum field theory in Minkowski spacetime reveals that the Hadamard condition is indeed essential for the mathematical consistency of QFT in Minkowski spacetime, as we would like to briefly explain now. In the following we will only discuss real scalar fields for simplicity. Analyses of Hadamard states for fields of higher spin can be found  e.g. in \cite{chap2_Dappiaggi:2011cj,chap2_Fewster:2003ey,chap2_Hollands:1999fc,chap2_Sahlmann:2000zr}.

The Borchers--Uhlmann algebra $\cA(M)$ of the free neutral Klein--Gordon field $\phi$ contains only very basic observables, namely, linear combinations of products of free fields at separate points, e.g. $\phi(x)\phi(y)$. However, if one wants to treat interacting fields in perturbation theory, or the backreaction of quantum fields on curved spacetimes via their stress--energy tensor, ones needs a notion of normal ordering, i.e. a way to define field monomials like $\phi^2(x)$ at the same point. To see that this requires some work, let us consider the massless scalar field in Minkowski spacetime. Its two point function reads
\beq\label{eq_MinkowskiMassless}\omega_2(x,y)=\omega(\phi(x)\phi(y))=\lim_{\epsilon\downarrow 0}\frac{1}{4\pi^2}\frac1{(x-y)^2+i\epsilon(x_0-y_0)+\epsilon^2}\,,\eeq
where $(x-y)^2$ denotes the Minkowskian product induced by the Minkowski metric and the limit has to be understood as being performed after integrating $\omega_2$ with at least one test function (weak limit). $\omega_2(x,y)$ is a smooth function if $x$ and $y$ are spacelike or timelike separated. It is singular at $(x-y)^2=0$, but the singularity is `good enough' to give a finite result when smearing $\omega_2(x,y)$ with two test functions. Hence, $\omega_2$ is a well--defined (tempered) distribution. Loosely speaking, this shows once more that the product of fields $\phi(x)\phi(y)$ is `well--defined' at non--null related points. However, if we were to define $\phi^2(x)$ by some `limit' like $$\phi^2(x)\doteq\lim_{x\to y}\phi(x)\phi(y)\,,$$ the expectation value of the resulting object would `blow up' and would not be any meaningful object. The well--known solution to this apparent problem is to define field monomials by appropriate regularising subtractions. For the squared field, this is achieved by setting
$$\wick{\phi^2(x)}\;\doteq \lim_{x\to y}\left(\phi(x)\phi(y)-\omega_2(x,y)\bI\right)\,,$$
where of course one would have to specify in which sense the limit should be taken. Omitting the details of this procedure, it seems still clear that the {\em Wick square} $\wick{\phi^2(x)}$ is a meaningful object, as it has a sensible expectation value, i.e. $\omega(\wick{\phi^2(x)})=0$. In the standard Fock space picture, one heuristically writes the field (operator) in terms of creation and annihilation operators in momentum space, i.e.
$$\phi(x)=\frac{1}{\sqrt{2\pi}^3}\int\limits_{\bR^3} \frac{d\vec{k}}{\sqrt{2k_0}}\; \left(a^\dagger_{\vec{k}}\;e^{ikx}+a^{\phantom{\dagger}}_{\vec{k}}\;e^{-ikx}\right)\,,$$
and defines $\wick{\phi^2(x)}$ by writing the mode expansion of the product $\phi(x)\phi(y)$, `normal ordering' the appearing products of creation and annihilation operators such that the creation operators are standing on the left hand side of the annihilation operators, and then finally taking the limit $x\to y$. It is easy to see that this procedure is equivalent to the above defined subtraction of the vacuum expectation value. However, having defined the Wick polynomials is not enough. We would also like to multiply them, i.e., we would like them to constitute an algebra. Using the mode-expansion picture, one can straightforwardly compute
$$\wick{\phi^2(x)}\wick{\phi^2(y)}\;=\;\wick{\phi^2(x)\phi^2(y)}+4\wick{\phi(x)\phi(y)}\omega_2(x,y)+2\left(\omega_2(x,y)\right)^2\,,$$
which is a special case of the well--known {\it Wick theorem}, see for instance \cite{chap2_Zuber}. The right hand side of the above equation is a sensible object if the appearing square of of the two--point function $\omega_2(x,y)$ is well--defined. In more detail, we know that $\omega_2(x,y)$ has singularities, and that these are integrable with test functions. Obviously, $(\omega_2(x,y))^2$ has singularities as well, and the question is whether the singularities are still good enough to be integrable with test functions. In terms of a mode decomposition, one could equivalently wonder whether the momentum space integrals appearing in the definition of $\wick{\phi^2(x)\phi^2(y)}$ via normal ordering creation and annihilation operators converge in a sensible way. The answer to these questions is `yes' because of the energy positivity property of the Minkowskian vacuum state, and this is the reason why one usually never worries about whether normal ordering is well--defined in quantum field theory on Minkowski spacetime. In more detail, the Fourier decomposition of the massless two--point function $\omega_2$ reads
\beq\label{eq_FourierMinkowskivaccuum}
\omega_2(x,y)=\lim_{\epsilon\downarrow 0}\frac{1}{(2\pi)^3}\int\limits_{\bR^4} dk \;\Theta(k_0)\delta(k^2)\;e^{ik(x-y)}e^{-\epsilon k_0}\,,
\eeq
where $\Theta(k_0)$ denotes the Heaviside step function. We see that the Fourier transform of $\omega_2$ has only support on the forward lightcone (or the positive mass shell in the massive case); this corresponds to the fact that we have associated the positive frequency modes to the creation operator in the above mode expansion of the quantum field. This insight allows to determine (or rather, define) the square of $\omega_2(x,y)$ by a convolution in Fourier space
\begin{align*}(\omega_2(x,y))^2&=\lim_{\epsilon\downarrow 0}\frac{1}{(2\pi)^6}\int\limits_{\bR^4} dq \int\limits_{\bR^4} dp \;\Theta(q_0)\;\delta(q^2)\;\Theta(p_0)\;\delta(p^2)\;e^{i(q+p)(x-y)}e^{-\epsilon (q_0+p_0)}\\
&=\lim_{\epsilon\downarrow 0}\frac{1}{(2\pi)^6}\int\limits_{\bR^4} dk \int\limits_{\bR^4} dq \;\Theta(q_0)\;\delta(q^2)\;\Theta(k_0-q_0)\;\delta((k-q)^2)\;e^{ik(x-y)}e^{-\epsilon k_0}\,.\end{align*}
Without going too much into details here, let us observe that the above expression can only give a sensible result (a distribution) if the integral over $q$ converges, i.e. if the integrand is rapidly decreasing in $q$. To see that this is the case, note that for an arbitrary but fixed $k$ and large $q$ where here `large' is meant in the Euclidean norm on $\bR^4$, the integrand is vanishing on account of $\delta(q^2)$ and $\Theta(k_0-q_0)$ as $k_0-q_0<0$ for large $q$. Loosely speaking, we observe the following: by the form of a convolution, the Fourier transform of $\omega_2$ is multiplied by the same Fourier transform, but with negative momentum. Since the $\omega_2$ has only Fourier support in one `energy direction', namely the positive one, the intersection of its Fourier support and the same support evaluated with negative momentum is compact, and the convolution therefore well--defined. Moreover, as this statement only relies on the large momentum behaviour of Fourier transforms, it holds equally in the case of massive fields, as the mass shell approaches the light cone for large momenta.

The outcome of the above considerations is the insight that, if we want to define a sensible generalisation of normal ordering in curved spacetimes, we have to select states whose two--point functions are singular, but regular enough to allow for pointwise multiplication. Even though general curved spacetimes are not translationally invariant and therefore do not allow to define a global Fourier transform and a related global energy positivity condition, one could think that this task can be achieved by some kind of a `local Fourier transform' and a related `local energy positivity condition' because only the `large momentum behaviour' is relevant. In fact, as showed in the pioneering work of Radzikowski \cite{chap2_Radzikowski,chap2_Radzikowski2}, this heuristic idea can be made precise in terms of {\em microlocal analysis}, a modern branch of Mathematics. Microlocal analysis gives a rigorous way to define the `large momentum behaviour' of a distribution in a coordinate-independent manner and in the aforementioned works \cite{chap2_Radzikowski,chap2_Radzikowski2}, it has been shown that so--called {\em Hadamard states}, which have already been known to allow for a sensible renormalisation of the stress--energy tensor \cite{chap2_Wald2,chap2_Wald3,chap2_WaldBook2}, indeed fulfil a local energy positivity condition in the sense that their two--point function has a specific {\em wave front set}. Based on this, Brunetti, Fredenhagen, K\"ohler, Hollands, and Wald \cite{chap2_BFK, chap2_BF00, chap2_HW01, chap2_HW02, chap2_HW04} have been able to show that one can as a matter of fact define normal ordering and perturbative interacting quantum field theories based on Hadamard states essentially in the same way as on Minkowski spacetimes. Though, it turned out that there is a big conceptual difference to flat spacetime quantum field theories, namely, new regularisation freedoms in terms of curvature terms appear. Although these are finitely many, and therefore lead to the result that theories which are perturbatively renormalisable in Minkowski spacetime retain this property in curved spacetimes \cite{chap2_BF00,chap2_HW01}, the appearance of this additional renormalisation freedom may have a profound impact on the backreaction of quantum fields on curved spacetimes, as we will discuss in the last chapter of this monograph.

As we have already seen at the example of the massless Minkowski vacuum, Hadamard states can be approached from two angles. One way to discuss them is to look at the concrete realisation of their two--point function in `position space'. This treatment has lead to the insight that Hadamard states are the sensible starting point for the definition of a regularised stress--energy tensor \cite{chap2_Wald2,chap2_Wald3,chap2_WaldBook2}, and it is well--suited for actual calculations in particular. On the other hand, the rather abstract study of Hadamard states based on microlocal analysis is useful in order to tackle and solve conceptual problems. Following our discussion of the obstructions in the definition of normal ordering, we shall start our treatment by considering the microlocal aspects of Hadamard states. A standard monograph on microlocal analysis is the book of H\"ormander \cite{chap2_Hormander}, who has also contributed a large part to this field of Mathematics \cite{chap2_Hormander0,chap2_Hormander02}. Introductory treatments can be found in \cite{chap2_Reed, chap2_BF00, chap2_Kratzert, chap2_Strohmeier}.

Let us start be introducing the notion of a wave front set. To motivate it, let us recall that a smooth function  on $\bR^m$ with compact support has a rapidly decreasing Fourier transform. If we take an distribution $u$ in $\Gamma^\prime_0(\bR^m)$ and multiply it by an $f\in\Sec(\bR^m)$ with $f(x_0)\neq 0$, then $uf$ is an element of $\Se^\prime(\bR^m)$, i.e., a distribution with compact support. If $fu$ were smooth, then its Fourier transform $\widehat{fu}$ would be smooth and rapidly decreasing. The failure of $fu$ to be smooth in a neighbourhood of $x_0$ can therefore be quantitatively described by the set of directions in Fourier space where $\widehat{fu}$ is not rapidly decreasing. Of course it could happen that we choose $f$ badly and therefore `cut' some of the singularities of $u$ at $x_0$. To see the full singularity structure of $u$ at $x_0$, we therefore need to consider all test functions which are non--vanishing at $x_0$. With this in mind, one first defines the wave front set of distributions on $\bR^m$ and then extends it to curved manifolds in a second step.

\begin{definition}
 \label{def_WaveFrontSet}\index{wave front set}
A neighbourhood $\Gamma$ of $k_0\in\bR^m$ is called {\em conic} if $k\in\Gamma$ implies $\lambda k\in\Gamma$ for all $\lambda\in(0,\infty)$. Let $u\in\Gamma^\prime_0(\bR^m)$. A point $(x_0,k_0)\in \bR^m\times (\bR^m\setminus\{0\})$ is called a regular directed point of $u$ if there is an $f\in\Sec(\bR^m)$ with $f(x_0)\neq 0$ such that, for every $n\in\bN$, there is a constant $C_n\in\bR$ fulfilling
$$|\widehat{fu}(k)|\le C_n (1+|k|)^{-n}$$
for all $k$ in a conic neighbourhood of $k_0$. The {\em wave front set} $\WF(u)$ is the complement in $\bR^m\times (\bR^m\setminus\{0\})$ of the set of all regular directed points of $u$.
\end{definition}

We immediately state a few important properties of wave front sets, the proofs of which can be found in \cite{chap2_Hormander} (see also \cite{chap2_Strohmeier}).

\begin{theorem}
 \label{thm_PropertiesWavefront} Let $u\in\Gamma^\prime_0(\bR^m)$.
\begin{itemize}
 \item[a)] If $u$ is smooth, then $\WF(u)$ is empty.
 \item[b)] Let $P$ be an arbitrary partial differential operator. It holds $$\WF(Pu)\subset \WF(u)\,.$$
\item[c)] Let $U$, $V\subset \bR^m$, let $u\in\Gamma^\prime_0(V)$, and let $\chi:U\to V$ be a diffeomorphism. The pull-back $\chi^*(u)$ of $u$ defined by $\chi^*u(f)=u(\chi_*f)$ for all $f\in\Sec(U)$ fulfils
$$\WF(\chi^*u)=\chi^*\WF(u)\doteq \left\{(\chi^{-1}(x),\chi^*k)\;|\;(x,k)\in \WF(u)\right\}\,,$$
where $\chi^*k$ denotes the push--forward of $\chi$ in the sense of {\em cotangent vectors}. Hence, the wave front set transforms covariantly under diffeomorphisms as an element of $T^*\bR^m$, and we can extend its definition to distributions on general curved manifolds $M$ by patching together wave front sets in different coordinate patches of $M$. As a result, for $u\in\Gamma^\prime_0(M)$, $\WF(u)\subset T^*M\setminus\{{\mathbf 0}\}$, where ${\mathbf 0}$ denotes the zero section of $T^*M$.
\item[d)] Let $u_1$, $u_2\in\Gamma^\prime_0(M)$ and let
$$\WF(u_1)\oplus \WF(u_2)\doteq \left\{(x, k_1+k_2)\;|\;(x,k_1)\in \WF(u_1),\; (x,k_2)\in \WF(u_2)\right\}\,.$$
If $\WF(u_1)\oplus \WF(u_2)$ does not intersect the zero section, then one can define the product $u_1u_2$ in such a way that it yields a well--defined distribution in $\Gamma^\prime_0(M)$ and that it reduces to the standard pointwise product of smooth functions if $u_1$ and $u_2$ are smooth. Moreover, the wave front set of such product is bounded in the following way
$$\WF(u_1u_2)\subset \WF(u_1)\cup \WF(u_2)\cup \left(\WF(u_1)\oplus \WF(u_2)\right)\,.$$
\end{itemize}
\end{theorem}

Note that the wave front set transforms as a subset of the cotangent bundle on account of the covector nature of $k$ in $\exp(ikx)$. The last of the above statements is exactly the criterion for pointwise multiplication of distributions we have been looking for. Namely, from  \eqref{eq_FourierMinkowskivaccuum} and \eqref{eq_MinkowskiMassless} one can infer that the wave front set of the Minkowskian two--point function (for $m\ge0$) is \cite{chap2_Reed}
\beq\label{eq_MinkwoskiWF} \WF(\omega_2)=\left\{(x,y,k,-k)\in T^*\bM^2 \;|\; x\neq y,\;(x-y)^2=0,\;k||(x-y),\; k_0>0\right\}\eeq
$$\cup\left\{(x,x,k,-k)\in T^*\bM^2 \;|\; k^2=0,\; k_0>0\right\}\,,$$
particularly, it is the condition $k_0>0$ which encodes the energy positivity of the Minkowskian vacuum state. We can now rephrase our observation that the pointwise square of $\omega_2(x,y)$ is a well--defined distribution by noting that $\WF(\omega_2)\oplus \WF(\omega_2)$ does not contain the zero section. In contrast, we know that the $\delta$-distribution $\delta(x)$ is singular at $x=0$ and that its Fourier transform is a constant. Hence, its wave front set reads
$$\WF(\delta)=\{(0,k)\;|\;k\in\bR\setminus\{0\}\}\,,$$
and we see that the $\delta$-distribution does not have a `one-sided' wave front set and, hence, can not be squared. The same holds if we view $\delta$ as a distribution $\delta(x,y)$ on $\Sec(\bR^2)$. Then
$$\WF(\delta(x,y))=\{(x,x,k,-k)\;|\;k\in\bR\setminus\{0\}\}\,.$$

The previous discussion suggests that a generalisation of \eqref{eq_MinkwoskiWF} to curved spacetimes is the sensible requirement to select states which allow for the construction of Wick polynomials. We shall now define such a generalisation.

\begin{definition}
 \label{def_HadamardScalar}
Let $\omega$ be a state on $\cA(M)$. We say that $\omega$ fulfils the {\em Hadamard condition} and is therefore a {\em Hadamard state} if its two--point function $\omega_2$ fulfils
$$\WF(\omega_2)=\left\{(x,y,k_x,-k_y)\in T^* M^2\setminus \{\mathbf 0\}\;|\;
(x,k_x)\sim(y,k_y),\;k_x\triangleright 0\right\}\,.$$
Here, $(x,k_x)\sim(y,k_y)$ implies that there exists a null geodesic $c$ connecting $x$ to $y$ such
that $k_x$ is coparallel and cotangent to $c$ at $x$ and $k_y$ is the parallel transport of $k_x$ from $x$ to $y$ along $c$. Finally, $k_x\triangleright 0$ means that the covector $k_x$ is future--directed.
\end{definition}

Having discussed the rather abstract aspect of Hadamard states, let us now turn to their more concrete realisations. To this avail, let us consider a geodesically convex set $\cO$ in $(M,g)$, see Section \ref{sec:ghst}. By definition, there are open subsets $\cO^\prime_x\subset T_xM$ such that the exponential map $\exp_x: \cO^\prime_x\to\cO$ is well--defined for all $x\in \cO$, i.e. we can introduce Riemannian normal coordinates on $\cO$. For any two points $x$, $y\in \cO$, we can therefore define the {\em half squared geodesic distance}\index{sigma@$\sigma$, half squared geodesic distance, Synge's world function} $\sigma(x,y)$ as
$$\sigma(x,y)\doteq\frac12 g\left(\exp_x^{-1}(y), \exp_x^{-1}(y)\right)\,.$$
This entity is sometimes also called {\em Synge's world function} and is both smooth and symmetric on $\cO\times\cO$. Moreover, one can show that it fulfils the following identity
\beq\label{eq_PDESigma}\sigma_{;\mu}\sigma_{;}^{\phantom{;}\mu}=2\sigma\,,\eeq
where the covariant derivatives are taken with respect to $x$ (even though this does not matter by the symmetry of $\sigma$), see for instance \cite{chap2_Friedlander,chap2_Poisson}. Let us introduce a couple of standard notations related to objects on $\cO\times \cO$ such as $\sigma$. If $\cV$ and $\cW$ are vector bundles over $M$ with typical fibers constituted by the vector spaces $V$ and $W$ respectively, then we denote by $\cV\boxtimes \cW$ the {\em exterior tensor product} of $\cV $ and $\cW $. $\cV \boxtimes \cW $ is defined as the vector bundle over $M\times M$ with typical fibre $V\otimes W$. The more familiar notion of the tensor product bundle $\cV \otimes \cW $ is obtained by considering the pull-back bundle of $\cV \boxtimes \cW $ with respect to the map $M\ni x\mapsto (x,x)\in M^2$. Typical exterior product bundles are for instance the tangent bundles of Cartesian products of $M$, e.g. $T^*M\boxtimes T^*M = T^*M^2$. A section of $\cV \boxtimes \cW $ is called a {\em bitensor}\index{bitensor}. We introduce the {\em Synge bracket notation} for the coinciding point limits of a bitensor. Namely, let $B$ be a smooth section of $\cV \boxtimes \cW $. We define

$$[B](x)\doteq \lim_{y\to x}B(x,y)\,.$$
With this definition, $[B]$ is a section of $\cV \otimes \cW $. In the following, we shall denote by unprimed indices tensorial quantities at $x$, while primed indices denote tensorial quantities at $y$. As an example, let us state the well--known {\em Synge rule}, proved for instance in \cite{chap2_Christensen0,chap2_Poisson}.
\begin{lemma}
 \label{lem_SyngeRule} Let $B$ be an arbitrary smooth bitensor. Its covariant derivatives at $x$ and $y$ are related by {\em Synge's rule}. Namely,
$$[B_{;\mu^\prime}]=[B]_{\;\mu}-[B_{;\mu}]\,.$$
Particularly, let $\cV $ be a vector bundle,  let $f_a$ be a local frame of $\cV $ defined on $\cO\subset M$ and let $x$, $y\in\cO$. If $B$ is {\em symmetric}, i.e. the coefficients $B_{ab^\prime}(x,y)$ of $$B(x,y)\doteq B^{ab^\prime}(x,y)\;f_a(x)\otimes f_{b^\prime}(y)$$ fulfil $$B^{ab^\prime}(x,y)=B^{b^\prime a}(y,x)\,,$$
then
$$[B_{;\mu^\prime}]=[B_{;\mu}]=\frac12 [B]_{;\mu}\,.$$
\end{lemma}

The half squared geodesic distance is a prototype of a class of bitensors of which we shall encounter many in the following. Namely, $\sigma$ fulfils a partial differential equation \eqref{eq_PDESigma} which relates its higher order derivatives to lower order ones. Hence, given the initial conditions
$$[\sigma]=0\,,\qquad [\sigma_{;\mu}]=0\,,\qquad [\sigma_{;\mu\nu}]=g_{\mu\nu}$$
which follow from the very definition of $\sigma$, one can compute the coinciding point limits of its higher derivatives by means of an inductive procedure, see for instance \cite{chap2_DeWitt, chap2_Christensen0, chap2_FullingBook, chap2_Poisson}. As an example, in the case of $[\sigma_{;\mu\nu\rho}]$, one differentiates
\eqref{eq_PDESigma} three times and then takes the coinciding point limit. Together with the already known relations, one obtains $$[\sigma_{;\mu\nu\rho}]=0\,.$$ At a level of fourth derivative, the same procedure yields a linear combination of three coinciding fourth derivatives, though with different index orders. To relate those, one has to commute derivatives to rearrange the indices in the looked-for fashion, and this ultimately leads to the appearance of Riemann curvature tensors and therefore to
$$[\sigma_{;\mu\nu\varrho\tau}]=-\frac{1}{3}(R_{\mu\varrho\nu\tau}+R_{\mu\tau\nu\varrho})\,.$$

A different bitensor of the abovementioned kind we shall need in the following is the {\em bitensor of parallel transport} $g^\mu_{\rho^\prime}(x,y)$. Namely, given a geodesically convex set $\cO$, $x$,
$y\in \cO$, and a vector $v=v^{\mu^\prime}\pa_{\mu^\prime}$ in $T_yM$, the parallel transport of $v$ from $y$ to $x$ along the unique geodesic in $\cO$ connecting $x$ and $y$ is given by the vector $\tilde v$ in $T_xM$ with components
$$\tilde v^\mu = g^\mu_{\rho^\prime}v^{\rho^\prime}\,.$$
This definition of the bitensor of parallel transport entails
$$[g^\mu_{\rho^\prime}]=\delta^\mu_\rho\,,\qquad g^\mu_{\rho^\prime;\alpha}\sigma_{;}^{\phantom{;}\alpha}=0\,,\qquad g^\mu_{\rho^\prime}\sigma_{;}^{\phantom{;}\rho^\prime}=-\sigma_{;}^{\phantom{;}\mu}\,.$$
In fact, the first two identities can be taken as the defining partial differential equation of $g^\mu_{\rho^\prime}$ and its initial condition (one can even show that the mentioned partial differential equation is an ordinary one). Out of these, one can obtain by the inductive procedure outlined above
$$[g^\mu_{\rho^\prime;\alpha}]=0\,, \qquad[g^\mu_{\rho^\prime;\alpha\beta}]=\frac12 R^\mu_{\phantom{\mu}\nu\alpha\beta}\,.$$
With these preparations at hand, let us now provide the explicit form of Hadamard states.

\begin{definition}
 \label{def_HadamardFormScalar}
Let $\omega_2$ be the two--point function of a state on $\cA(M)$, let $t$ be a time function on $(M,g)$, let
$$\sigma_\epsilon(x,y)\doteq\sigma(x,y)+2i\epsilon(t(x)-t(y))+\epsilon^2\,,$$ and let $\lambda$ be an arbitrary length scale. We say that $\omega_2$ if of {\em local Hadamard form}\index{Hadamard form} if, for every $x_0\in M$ there exists a geodesically convex neighbourhood $\cO$ of $x_0$ such that $\omega_2(x,y)$ on $\cO\times \cO$ is of the form
\begin{align*}\omega_2(x,y)&=\lim_{\epsilon\downarrow 0}\frac{1}{8\pi^2}\left(\frac{u(x,y)}{\sigma_\epsilon(x,y)}+v(x,y)\log\left(\frac{\sigma_\epsilon(x,y)}{\lambda^2}\right)+w(x,y)\right)\\
&\doteq\lim_{\epsilon\downarrow 0}\frac{1}{8\pi^2}\left(h_\epsilon(x,y)+w(x,y)\right)\,.\end{align*}
Here, the {\em Hadamard coefficients}\index{Hadamard coefficients} $u$, $v$, and $w$ are smooth, real--valued biscalars, where $v$ is given by a series expansion in $\sigma$ as
$$v=\sum\limits^\infty_{n=0}v_n \sigma^n$$
with smooth biscalar coefficients $v_n$. The bidistribution $h_\epsilon$ shall be called {\em Hadamard parametrix}\index{Hadamard parametrix}, indicating that it solves the Klein--Gordon equation up to smooth terms.
\end{definition}

Note that the above series expansion of $v$ does not necessarily converge on general smooth spacetimes, however, it is known to converge on analytic spacetimes \cite{chap2_Garabedian}. One therefore often truncates the series at a finite order $n$ and asks for the $w$ coefficient to be only of regularity $C^n$, see \cite{chap2_Kay}. Moreover, the local Hadamard form is special case of the {\em global Hadamard form} defined for the first time in \cite{chap2_Kay}. The definition of the global Hadamard form in \cite{chap2_Kay} assures that there are no (spacelike) singularities in addition to the lightlike ones visible in the local form and, moreover, that the whole concept is independent of the chosen time function $t$. However, as proven by Radzikowski in \cite{chap2_Radzikowski2} employing the microlocal version of the Hadamard condition, the local Hadamard form already implies the global Hadamard form on account of the fact that $\omega_2$ must be positive, have the causal propagator $E$ as its antisymmetric part, and fulfil the Klein--Gordon equation in both arguments. It is exactly this last fact which serves to determine the Hadamard coefficients $u$, $v$, and $w$ by a recursive procedure.

To see this, let us omit the subscript $\epsilon$ and the scale $\lambda$ in the following, since they do not influence the result of the following calculations, and let us denote by $P_x$ the Klein--Gordon operator $P=-\Box + \xi R + m^2$ acting on the $x$-variable. Applying $P_x$ to $h$, we obtain potentially singular terms proportional to $\sigma^{-n}$ for $n=1,2,3$ and to $\log \sigma$, as well as smooth terms proportional to positive powers of $\sigma$. We know, however, that the total result is smooth because $P_x(h+w)=0$ since $\omega_2$ is a bisolution of the Klein--Gordon equation and $w$ is smooth. Consequently, the potentially singular terms in $P_x h$ have to vanish identically at each order in $\sigma$ and $\log \sigma$. This immediately implies
\beq\label{eq_PDEVScalar}P_xv= 0\,.\eeq
and, further, the following recursion relations
\begin{gather}
-P_xu+2v_{0;\mu}\sigma^\mu+(\square_x\sigma-2)v_0=0\,,\label{eq_PDEV0Scalar}\\
2u_{;\mu}\sigma^\mu+(\square_x\sigma-4)u=0\,,\label{eq_PDEUScalar}\\
-P_x v_n+2(n+1)v_{n+1;\mu}\sigma_{;}^{\phantom{;}\mu}+\left(n+1
\right)\left(\square_x\sigma+2n\right)v_{n+1}=0\,,\quad\forall n\geq 0\,.\label{eq_PDEVNScalar}
\end{gather}

To solve these recursive partial differential equations, let us now focus on \eqref{eq_PDEUScalar}. Since the only derivative appearing in this equation is the derivative along the geodesic connecting $x$ and $y$, \eqref{eq_PDEUScalar} is in fact an ordinary differential equation with respect to the affine parameter of the mentioned geodesic. $u$ is therefore uniquely determined once a suitable initial condition is given. Comparing the Hadamard form with the Minkowskian two--point function \eqref{eq_MinkowskiMassless}, the initial condition is usually chosen as
$$[u]=1\,,$$
which leads to the well--known result that $u$ is given by the square root of the so--called {\em Van Vleck-Morette} determinant, see for instance \cite{chap2_DeWitt, chap2_Christensen0, chap2_FullingBook, chap2_Poisson}. Similarly, given $u$, the differential equation \eqref{eq_PDEV0Scalar} is again an ordinary one with respect to the geodesic affine parameter, and it can be immediately integrated since taking the coinciding point limit of \eqref{eq_PDEV0Scalar} and inserting the properties of $\sigma$ yields the initial condition
$$[v_0]=\frac12 [P_xu]\,.$$
It is clear how this procedure can be iterated to obtain solutions for all $v_n$. Particularly, one obtains the initial conditions
$$[v_{n+1}]=\frac{1}{2(n+1)(n+2)}[P_xv_n]$$
 for all $n>0$. Moreover, one finds that $u$ depends only on the local geometry of the spacetime, while the $v_n$ and, hence, $v$ depend only on the local geometry and the parameters appearing in the Klein--Gordon operator $P$, namely, the mass $m$ and the coupling to the scalar curvature $\xi$. These observations entail that the state dependence of $\omega_2$ is encoded in the smooth biscalar $w$, which furthermore has to be symmetric because it is bound to vanish in the difference of two--point functions yielding the antisymmetric causal propagator $E$, {\it viz.} $$\omega_2(x,y)-\overline{\omega_2(x,y)}=\omega_2(x,y)-\omega_2(y,x)=iE(x,y)\,.$$ More precisely, this observation ensues from the following important result obtained in \cite{chap2_Moretti2,chap2_Moretti}.
\begin{theorem}
 \label{thm_SymmetryHadamard} The Hadamard coefficients $v_n$ are symmetric biscalars.
\end{theorem}
This theorem proves the folklore knowledge that the causal propagator $E$ is locally given by $$iE=\lim_{\epsilon\downarrow 0}\frac{1}{8\pi^2}(h_\epsilon-h_{-\epsilon})\,.$$

Even though we can in principle obtain the $v_n$ as unique solutions of ordinary differential equations, we shall only need their coinciding point limits and coinciding points limit of their derivatives in what follows. In this respect, the symmetry of the $v_n$ will prove very valuable in combination with Lemma \ref{lem_SyngeRule}. In fact, employing the Hadamard recursion relations, we find the following results \cite{chap2_Mo03}.

\begin{lemma}
 \label{lem_PHadamardScalar}
The following identities hold for the Hadamard parametrix $h(x,y)$
$$[P_xh]=[P_yh]=-6[v_1]\,,\quad [(P_xh)_{;\mu}]=[(P_yh)_{;\mu^\prime}]=-4[v_1]_{;\mu}\,,$$
$$\quad [(P_xh)_{;\mu^\prime}]=[(P_yh)_{;\mu}]=-2[v_1]_{;\mu}\,.$$
\end{lemma}
It is remarkable that these rather simple computations will be essentially sufficient for the construction of a conserved stress--energy tensor of a free scalar quantum field \cite{chap2_Mo03}. Particularly, the knowledge of the explicit form of, say, $[v_1]$ is not necessary to accomplish such a task. However, if one is interested in computing the actual backreaction of a scalar field on curved spacetimes, one needs the explicit form of $[v_1]$. One can compute this straightforwardly by the inductive procedure already mentioned at several occasions and the result is \cite{chap2_Mo03,chap2_Decanini2}

\begin{align}
[v_1]&=\frac{m^4}{8}+\frac{\left(6\xi-1\right)m^2R}{24}+\frac{\left(6\xi-1\right)R^2}{288}+\frac{(1-5\xi)\Box R}{120}\notag\\
&\qquad -\frac{R_{\alpha\beta}R^{\alpha\beta}}{720}+\frac{R_{\alpha\beta\gamma\delta}R^{\alpha\beta\gamma\delta}}{720}\notag\\
&=\frac{m^4}{8}+\frac{\left(6\xi-1\right)m^2R}{24}+\frac{\left(6\xi-1\right)R^2}{288}+\frac{(1-5\xi)\Box R}{120}\label{eq_V1Scalar}\\
&\qquad +\frac{C_{\alpha\beta\gamma\delta}C^{\alpha\beta\gamma\delta}+R_{\alpha\beta}R^{\alpha\beta}-\frac{R^2}{3}}{720}\,.\notag
\end{align}

The Hadamard coefficients are related to the so--called {\em DeWitt-Schwinger} coefficients, see for instance \cite{chap2_Moretti2,chap2_Moretti,chap2_Decanini1}, which stem from an {\it a priori} completely different expansion of two--point functions. The latter have been computed for the first time in \cite{chap2_Christensen0,chap2_Christensen}
  and can also be found in many other places like, e.g. \cite{chap2_Decanini1,chap2_FullingBook}.
  
Having discussed the Hadamard form to a large extent, let us state the already anticipated equivalence result obtained by Radzikowski in \cite{chap2_Radzikowski}. See also \cite{chap2_Sahlmann:2000zr} for a slightly different proof, which closes a gap in the proof of \cite{chap2_Radzikowski}.

\begin{theorem}
 \label{thm_HadamardScalar} Let $\omega_2$ be the two--point function of a state on $\cA(M)$. $\omega_2$ fulfils the Hadamard condition of Definition \ref{def_HadamardScalar} if and only if it is of global Hadamard form.
\end{theorem}

\noindent By the result of \cite{chap2_Radzikowski2}, that a state which is locally of Hadamard form is already of global Hadamard form, we can safely replace `global' by `local' in the above theorem. Moreover, from the above discussion it should be clear that the two--point functions of two Hadamard states differ by a smooth and symmetric biscalar.

In past works on (algebraic) quantum field theory in curved spacetimes, one has often considered only on quasifree Hadamard states. For non--quasifree states, a more general {\em microlocal spectrum condition} has been proposed in \cite{chap2_BFK} which requires certain wave front set properties of the higher order $n$--point functions of a non--quasifree state. However, as shown in \cite{chap2_Sanders,chap2_Sanders:2009sw}, the Hadamard condition of the two--point function of a non--quasifree state alone already determines the singularity structure of all higher order $n$--point functions by the CCR. It is therefore sufficient to specify the singularity structure of $\omega_2$ also in the case of non--quasifree states. Note however, that certain technical results on the structure of Hadamard states have up to now only been proven for the quasifree case  \cite{chap2_VerchQ}.
  
We close the general discussion of Hadamard states by providing examples and non--examples.
  
\paragraph{Examples of Hadamard states}

\begin{enumerate}
\item Given a Hadamard state $\omega$ on $\cA(M)$, any `finite excitation of $\omega$' is again Hadamard, i.e. for all $A\in\cA(M)$, $\omega_A$ defined for all $B\in\cA(M)$ by $\omega_A(B)\doteq \omega(A^* B A)/\omega(A^*A)$ is Hadamard \cite{chap2_Sanders:2009sw}.
\item All vacuum states and KMS (thermal equilibrium) states on ultrastatic spacetimes (i.e.  spacetimes with a metric $ds^2 = -dt^2 + h_{ij}dx^i dx^j$, with $h_{ij}$ not depending on time) are Hadamard states \cite{chap2_Fulling,chap2_Sahlmann:2000fh}.
\item Based on the previous statement, it has been proven in in \cite{chap2_Fulling} that Hadamard states exist on {\it any} globally hyperbolic spacetime by means of a spacetime deformation argument.
\item The {\it Bunch--Davies state} on de Sitter spacetime is a Hadamard state \cite{chap2_Allen,chap2_Dappiaggi:2007mx,chap2_Dappiaggi:2008dk}. It has been shown in \cite{chap2_Dappiaggi:2007mx,chap2_Dappiaggi:2008dk} that this result can be generalised to asymptotically de Sitter spacetimes, where distinguished Hadamard states can be constructed by means of a holographic argument; these states are generalisations of the Bunch--Davies state in the sense that the aforementioned holographic construction yields the Bunch--Davies state in de Sitter spacetime.
\item Similar holographic arguments have been used in \cite{chap2_Dappiaggi:2011cj,chap2_Dappiaggi:2005ci,chap2_Moretti:2005ty,chap2_Moretti:2006ks} to construct distinguished Hadamard states on asymptotically flat spacetimes, to rigorously construct the Unruh state in Schwarzschild spacetimes and to prove that it is Hadamard in \cite{chap2_Dappiaggi:2009fx}, to construct asymptotic vacuum and thermal equilibrium states in certain classes of Friedmann--Robertson--Walker spacetimes in -- see \cite{chap2_Dappiaggi:2010gt} and the next chapter -- and to construct Hadamard states in bounded regions of any globally hyperbolic spacetime in \cite{chap2_Dappiaggi:2010iq}.
\item A interesting class of Hadamard states in general Friedmann-Robertson-Walker are the {\it states of low energy} constructed in \cite{chap2_Olbermann:2007gn}. These states minimise the energy density integrated in time with a compactly supported weight function and thus loosely speaking minimise the energy in the time interval specified by the support of the weight function. This construction has been generalised to encompass almost equilibrium states in \cite{chap2_Kusku:2008zz} and to expanding spacetimes with less symmetry in \cite{chap2_Them:2013uka}. We shall review the construction of these states in the next chapter.
\item A construction of Hadamard states which is loosely related to states of low energy has been given in \cite{chap2_Brum:2013bia}. There the authors consider for a given spacetime $(M,g)$ with the spacetime $(N,g)$ where $N$ is a finite--time slab of $M$. Given a smooth timelike--compact function $f$ on $M$ which is identically 1 on $N$, one considers $A\doteq if Ef$ which can be shown to be a bounded and self--adjoint operator on $L^2(N,d_gx)$. The positive part $A^+$ of $A$ constructed with standard functional calculus defines a two--point function of a quasifree state $\omega$ on $\cA(N)$ via $\omega_2(f,g)\doteq \ip{f,A^+ g}$ which can be shown to be Hadamard (at least on classes of spacetimes) \cite{chap2_Brum:2013bia}. Taking for $f$ the characteristic function of $N$ gives the Sorkin--Johnston states proposed in \cite{chap2_Afshordi:2012jf} which are in general not Hadamard \cite{chap2_Fewster:2012ew}.
\item Hadamard states which possess an approximate local thermal interpretation have been constructed in \cite{chap2_Schlemmer}, see \cite{chap2_VerchRegensburg} for a review.
\item Given a Hadamard state $\omega$ on the algebra $\cA(M)$ and a smooth solution $\Psi$ of the field equation $P\Psi=0$, one can construct a coherent state by redefining the quantum field $\phi(x)$ as $\phi(x)\mapsto \phi(x) + \Psi(x)\mathbb{I}$. The thus induced coherent state has the two--point function $\omega_{\Psi,2}(x,y)=\omega_{2}(x,y)+\Psi(x)\Psi(y)$, which is Hadamard since $\Psi(x)$ is smooth.
\item A construction of Hadamard states via pseudodifferential calculus was developed in \cite{chap2_Gerard:2012wb}.
\end{enumerate}

\paragraph{Non-examples of Hadamard states}

\begin{enumerate}
\item The so--called {\em $\alpha$-vacua} in de Sitter spacetime \cite{chap2_Allen} violate the Hadamard condition as shown in \cite{chap2_Brunetti:2005pr}.
\item As already mentioned the  {\em Sorkin--Johnston states} proposed in \cite{chap2_Afshordi:2012jf} are in general not Hadamard \cite{chap2_Fewster:2012ew}.
\item A class of states related to Hadamard states, but in general not Hadamard, is constituted by {\it adiabatic states}. These have been introduced in \cite{chap2_Parker} and put on rigorous grounds by \cite{chap2_Luders}. Effectively, they are states which approximate ground states if the curvature of the background spacetime is only slowly varying. In \cite{chap2_Junker}, the concept of adiabatic states has been generalised to arbitrary curved spacetimes. There, it has also been displayed in a quantitative way how adiabatic states are related to Hadamard states. Namely, an adiabatic state of a specific order $n$ has a certain {\it Sobolev wave front set} (in contrast to the $C^\infty$ wave front set introduced above) and hence, loosely speaking, it differs from a Hadamard state by a biscalar of finite regularity $C^n$. In this sense, Hadamard states are adiabatic states of `infinite order'. We will review the concept of adiabatic states in the next chapter
\end{enumerate}  
  
Finally, let us remark that one can define the Hadamard form also in spacetimes with dimensions differing from 4, see for instance \cite{chap2_Sahlmann:2000zr,chap2_Mo03}. Moreover, the proof of the equivalence of the concrete Hadamard form and the microlocal Hadamard condition also holds in arbitrary spacetime dimensions, as shown in \cite{chap2_Sahlmann:2000zr}. A recent detailed exposition of Hadamard states may be found in \cite{chap2_Khavkine:2014mta}.

\section{Locality and General Covariance}
\label{sec:lcqft}

And important aspect of QFT on curved spacetimes is the backreaction of quantum fields in curved spacetimes, i.e. the effect of quantum matter--energy on the curvature of spacetime. This of course necessitates the ability to define quantum field theory on a curved spacetime {\em without knowing the curved spacetime beforehand}. It is therefore advisable to employ only generic properties of spacetimes in the construction of quantum fields. This entails that we have to formulate a quantum field theory in a {\em local} way, i.e. only employing local properties of the underlying curved manifold. In addition, we would like to take into account the diffeomorphism--invariance of General Relativity and therefore construct {\em covariant} quantum fields. This concept of a {\em locally covariant quantum field theory} goes back to many works, of which the first one could mention is \cite{chap2_DimockScalar}, followed by many others such as \cite[Chapter 4.6]{chap2_WaldBook2} and \cite{chap2_VerchSpin,chap2_HW01}. Building on these works, the authors of \cite{chap2_BFV} have given the first complete definition of a locally covariant quantum field theory.
  
As shown in \cite{chap2_BFV}, giving such a definition in precise mathematical terms requires the language of {\em category theory}, a branch of mathematics which basically aims to unify {\em all} mathematical structures into one coherent picture. A category is essentially a class $\gC$ of {\em objects} denoted by $\text{\it obj}(\gC)$, with the property that, for each two objects $A$, $B$ in $\gC$ there is (at least) one {\em morphism} or {\em arrow} $\phi :A\to B$ relating $A$ and $B$. The collection of all such arrows is denoted by $\text{\it hom}_\gC(A,B)$. Morphisms relating a chain of three objects are required to be associative with respect to compositions, and one demands that each object has an {\em identity morphism} $\text{\it id}_A:A\to A$ which leaves all morphisms $\phi :A\to B$ starting from $A$ invariant upon composition, i.e. $\phi\circ \text{\it id}_A =\phi$. An often cited simple example of a category is the category of sets $\mathfrak{Set}$. The objects of $\mathfrak{Set}$ are sets, while the morphism are maps between sets, the identity morphism of an object just being the identity map of a set. Given two categories $\gC_1$ and $\gC_2$, a {\em functor} $F:\gC_1\to\gC_2$ is a map between two categories which maps objects to objects and morphisms to morphisms such that identity morphisms in $\gC_1$ are mapped to identity morphism in $\gC_2$ and the composition of morphisms is preserved under the mapping. This paragraph was only a very brief introduction to category theory and we refer the reader to the standard monograph \cite{chap2_MacLane} and to the introduction in \cite[Section 1.7]{chap2_Szekeres} for further details. A locally covariant quantum field theory according to \cite{chap2_BFV} should be a functor from a category of spacetimes to a category of suitable algebras. The first step in understanding such a construction if of course the definition of a suitable category of spacetimes.  

We have already explained in the previous sections of this chapter that four--dimensional, oriented and time--oriented, globally hyperbolic spacetimes are the physically sensible class of spacetimes among all curved Lorentzian manifolds. It is therefore natural to take them as the objects of a potential category of spacetimes. Regarding the morphisms, one could think of various possibilities to select them among all possible maps between the spacetimes under consideration. However, to be able to emphasise the local nature of a quantum field theory, we shall take embeddings between spacetimes. This will allow us to require locality by asking that a quantum field theory on a `small' spacetime can be easily embedded into a larger spacetime without `knowing anything' about the remainder of the larger spacetime. Moreover, a sensible quantum field theory will depend on the orientation and time--orientation and the causal structure of the underlying manifold, we should therefore only consider embeddings that preserve these structures. To this avail, the authors in \cite{chap2_BFV} have chosen isometric embeddings with causally convex range (see Section \ref{sec:ghst} regarding an explanation of these notions), but since the causal structure of a spacetime is left invariant by conformal transformations, one could also choose conformal embeddings, as done in \cite{chap2_NicolaConformal}. We will nevertheless follow the choice of \cite{chap2_BFV}, since it will be sufficient for our purposes. Let us now subsume the above considerations in a definition.
  
\begin{definition}
\label{def_CategorySpacetimes}
The {\em category of spacetimes} $\mathfrak{Man}$ is the category having as its class of objects $\text{\it obj}(\mathfrak{Man})$ the globally hyperbolic, four-dimensional, oriented and time-oriented spacetimes $(M,g)$. Given two spacetimes $(M_1,g_1)$ and $(M_2,g_2)$ in $\text{\it obj}(\mathfrak{Man})$, the considered morphisms $\text{\it hom}_{\mathfrak{Man}}\left((M_1,g_1),(M_2,g_2)\right)$ are isometric embeddings $\chi:(M_1,g_1)\hookrightarrow (M_2,g_2)$ preserving the orientation and time-orientation and having causally convex range $\chi(M_1)$. Moreover, the identity morphism $\text{\it id}_{(M,g)}$ of a spacetime in $\text{\it obj}(\mathfrak{Man})$ is just the identity map of $M$ and the composition of morphisms is defined as the usual composition of embeddings.
\end{definition}

The just defined category is sufficient to discuss locally covariant Bosonic quantum field theories. However, for Fermionic quantum field theories, one needs a category which incorporates spin structures as defined in \cite{chap2_VerchSpin,chap2_Sanders}. At this point we briefly remark that our usage of the words `Boson' and `Fermion' for integer and half--integer spin fields respectively is allowed on account of the spin--statistics theorem in curved spacetimes proved in \cite{chap2_VerchSpin}, see also \cite{chap2_Fewster:2015yga} for a more recent and general work.  
  
To introduce the notion of a locally covariant quantum field theory and the related concept of a locally covariant quantum field, we need a few categories in addition to the one introduced above. By ${\mathfrak{TAlg}}$ we denote the category of unital topological $*$--algebras, where for two $\cA_1$, $\cA_2$ in $\obj(\mathfrak{TAlg})$, the considered morphisms $\hom_\mathfrak{TAlg}(\cA_1,\cA_2)$ are continuous, unit--preserving, injective $*$--homomorphisms. In addition, we introduce the category $\mathfrak{Test}$ of test function spaces $\Gamma_0(M)$ of objects $(M,g)$ in $\mathfrak{Man}$, where here the morphisms $\hom_\mathfrak{Test}(\Sec(M_1),\Sec(M_2))$ are push--forwards $\chi_*$ of the isometric embeddings $\chi:M_1\hookrightarrow M_2$. In fact, by $\cD$ we shall denote the functor between $\mathfrak{Man}$ and $\mathfrak{Test}$ which assigns to a spacetime $(M,g)$ in $\mathfrak{Man}$ its test function space $\Gamma_0(M)$ and to a morphism in $\mathfrak{Man}$ its push--forward. For reasons of nomenclature, we consider $\mathfrak{TAlg}$ and $\mathfrak{Test}$ as subcategories of the category $\mathfrak{Top}$ of all topological spaces with morphisms given by continuous maps. Let us now state the first promised definition.

\begin{definition}
 \label{def_LCQFT}
A {\em locally covariant quantum field theory}\index{locally covariant quantum field theory} is a (covariant) functor $\cA$ between the two categories $\mathfrak{Man}$ and $\mathfrak{TAlg}$. Namely, let us denote by $\alpha_\chi$ the mapping $\cA(\chi)$ of a morphism $\chi$ in $\mathfrak{Man}$ to a morphism in $\mathfrak{TAlg}$ and by $\cA(M,g)$ the mapping of an object in $\mathfrak{Man}$ to an object in $\mathfrak{TAlg}$, see the following diagram.
\begin{equation*}
\begin{CD}
(M_1,g_1) @>\chi>> (M_2,g_2)\\
@V{\cA}VV     @VV{\cA}V\\
\cA(M_1,g_1)@>{\alpha_\chi}>> \cA(M_2,g_2)
\end{CD}
\end{equation*}
Then, the following relations hold for all morphisms $\chi_{ij}\in \hom_\mathfrak{Man}((M_i,g_i),(M_j,g_j))$:
$$\alpha_{\chi_{23}}\circ \alpha_{\chi_{12}}=\alpha_{\chi_{23}\circ \chi_{12}}\,,\qquad \alpha_{\text{\it id}_M}=\text{\it id}_{\cA(M,g)}\,.$$
\indent A locally covariant quantum field theory is called {\em causal} if in all cases where $\chi_i\in \hom_\mathfrak{Man}((M_i,g_i),(M,g))$ are such that the sets $\chi_1(M_1)$ and $\chi_2(M_2)$ are spacelike separated in $(M,g)$,
$$\left[\alpha_{\chi_1}(\cA(M_1,g_1)),\alpha_{\chi_2}(\cA(M_2,g_2))\right]=\{0\}$$
in the sense that all elements in the two considered algebras are mutually commuting.\\\indent
Finally, one says that a locally covariant quantum field theory fulfils the {\em time--slice axiom}, if the situation that $\chi\in \hom_\mathfrak{Man}((M_1,g_1),(M_2,g_2))$ is such that $\chi(M_1,g_1)$ contains a Cauchy surface of $(M_2,g_2)$ entails
$$\alpha_\chi\left(\cA(M_1,g_1)\right)=\cA(M_2,g_2)\,.$$
\end{definition}

The authors of \cite{chap2_BFV} also give the definition of a state space of a locally covariant quantum field theory and this turns out to be dual to a functor, by the duality relation between states and algebras. One therefore chooses the notation of {\em covariant functor} for a functor in the strict sense, and calls such a mentioned dual object a {\em contravariant functor}. We stress once more that the term `local' refers to the size of spacetime regions. A locally covariant quantum field theory is such that it can be constructed on arbitrarily small (causally convex) spacetime regions without having any information on the remainder of the spacetime. In more detail, this means that the algebraic relations of observables in such small region are already fully determined by the information on this region alone. This follows by application of the above definition to the special case that $(M_1,g_1)$ is a causally convex subset of $(M_2,g_2)$.

As shown in \cite{chap2_BFV}, the quantum field theory given by assigning the Borchers--Uhlmann algebra $\cA(M)$ of the free Klein--Gordon field to a spacetime $(M,g)$ is a locally covariant quantum field theory fulfilling causality and the time--slice axiom. This follows from the fact that the construction of $\cA(M)$ only employs compactly supported test functions and the causal propagator $E$. The latter is uniquely given on any globally hyperbolic spacetime, particularly, the causal propagator on a causally convex subset $(M_1,g_1)$ of a globally hyperbolic spacetime $(M_2,g_2)$ coincides with the restriction of the same propagator on $(M_2,g_2)$ to $(M_1,g_1)$. Finally, causality follows by the causal support properties of the causal propagator, and the time-slice axiom follows by Lemma \ref{lem_TimeSliceAxiomQuantumScalar}.

Let us now discuss the notion of a {\em locally covariant quantum field}. These fields are particular observables in a locally covariant quantum field theory which transform covariantly, i.e. loosely speaking, as a tensor and are in addition constructed only out of local geometric data. In categorical terms, this means that they are {\em natural transformations} between the functors $\cD$ and $\cA$. We refer to \cite{chap2_MacLane} for the notion of a natural transformation, however, its meaning in our context should be clear from the following definition.  

\begin{definition}
 \label{def_LCQField}
A {\em locally covariant quantum field} $\Phi$ is a natural transformation between the functors $\cD$ and $\cA$. Namely, for every object $(M,g)$ in $\mathfrak{Man}$ there exists a morphism $\Phi_{(M,g)}:\Sec(M,g)\to \cA(M,g)$ in $\mathfrak{Top}$ such that, for each morphism $\chi\in \hom_\mathfrak{Man}((M_1,g_1),(M_1,g_1))$, the following diagram commutes.
\begin{equation*}
\begin{CD}
\Sec(M_1,g_1) @>\Phi_{(M_1,g_1)}>> \cA(M_1,g_1)\\
@V{\chi_*}VV     @VV{\alpha_{\chi}}V\\
\Sec(M_2,g_2)@>>\Phi_{(M_2,g_2)}> \cA(M_2,g_2)
\end{CD}
\end{equation*}
Particularly, this entails that
$$\alpha_\chi\circ \Phi_{(M_1,g_1)}=\Phi_{(M_2,g_2)}\circ \chi_*\,.$$
\end{definition}

It is easy to see that the Klein--Gordon field $\phi(f)$ is locally covariant. Namely, the remarks on the local covariance of the quantum field theory given by $\cA(M)$ after Definition \ref{def_LCQFT} entail that an isometric embedding $\chi:(M_1,g_1)\hookrightarrow (M_2,g_2)$ transforms $\phi(f)$ as
$$\alpha_\chi\left(\phi(f)\right)=\phi(\chi_* f)\,,$$ or, formally,
$$\alpha_\chi\left(\phi(x)\right)=\phi\left(\chi(x)\right)\,.$$
Hence, local covariance of the Klein--Gordon field entails that it transforms as a `scalar'. While the locality and covariance of the Klein--Gordon field itself are somehow automatic, one has to take care that all extended quantities, like Wick powers and time--ordered products thereof, maintain these good properties. The prevalent paradigm in algebraic quantum field theory is that {\em all} pointlike observables should be theoretically modelled by locally covariant quantum fields.

A comprehensive review of further aspects and results of locally covariant quantum field theory may be found in \cite{chap2_Fewster:2015kua}.

\section{The Quantum Stress--Energy Tensor and the Semiclassical Einstein equation}
\label{sec:stresstensor_see}

The aim of this section is to discuss the semiclassical Einstein equation and the quantum--stress--energy tensor\index{stress--energy tensor} $\wick{T_{\mu\nu}}$ which is the observable whose expectation value enters this equation. As argued in the previous section, all pointlike observables such as the quantum stress--energy tensor should be locally covariant fields in the sense of Definition \ref{def_LCQField}. Rather than discussing local covariance for non--linear pointlike observables only at the example of $\wick{T_{\mu\nu}}$, it is instructive to review the construction of general local and covariant Wick polynomials\index{Wick polynomials}.

\subsection{Local and Covariant Wick Polynomials}
\label{sec:wick}

The first construction of local and covariant general Wick polynomials was given in 
\cite{chap2_HW01} based on ideas already implemented for the stress--energy tensor in \cite{chap2_WaldBook2}. Here we would like to review a variant of the construction of \cite{chap2_HW01} in the spirit of the functional approach to perturbative QFT on curved spacetimes, termed {\em perturbative algebraic quantum field theory}, cf.
\cite{chap2_Brunetti:2009qc,chap2_FredenhagenRejzner,chap2_FredenhagenRejzner2}. Essentially, this point of view on local Wick polynomials was already taken in \cite{chap2_HW04}. We review here only the case of the neutral Klein--Gordon field, however, the functional approach is applicable to general field theories \cite{chap2_Rejzner:2011au,chap2_FredenhagenRejzner,chap2_FredenhagenRejzner2}.

In the functional approach to algebraic QFT on curved spacetimes, one considers observables as functionals on the classical field configurations. Upon quantization, these functionals are endowed with a particular non--commutative product which encodes the commutation relations of quantum observables. We have already taken this point of view in the discussion of the Borchers--Uhlmann algebra $\cA(M)$ as the result of quantizing a classical symplectic space constructed in Section \ref{sec:linear_class}. In particular, the smeared quantum scalar field $\phi(f)$ was considered as the quantization of the linear functional $F_f:\Se(M)\to \bC$, $F_f(\phi)=\ip{f,\phi}$ with $f\in\Sec^\bC(M)$. The new aspect in the approach we shall review now is to consider a much larger class of functionals on $\Se(M)$. To this avail, we view $\Se(M)$ as the space of off--shell configurations of the scalar field, whereas $\Sol\subset \Se(M)$ is the space of on--shell configurations. For the purpose of perturbation theory it is more convenient to perform all constructions off--shell first and to go on--shell only in the end, and we shall follow this route as well, even though perturbative constructions are not dealt with in this monograph.

To this end, we call a functional $F:\Se(M)\to\bC$ smooth if the $n$--th functional derivatives
\begin{gather}\label{eq: smooth functional}
\left\langle F^{(n)}(\phi),\psi_1\otimes\dots\otimes\psi_n\right\rangle\doteq \left.\frac{d^n}{d\lambda_1 \dots d\lambda_n}F\left(\phi+\sum^n_{j=1}\lambda_j\psi_j\right)\right|_{\lambda_1=\dots=\lambda_n=0}
\end{gather}
exist for all $n$ and all $\psi_1,\dots,\psi_n\in\Sec(M)$ and if $F^{(n)}(\phi)\in\Se^\prime(M^n)$, i.e. $F^{(n)}(\phi)$ is a distribution. By definition $F^{(n)}(\phi)$ is symmetric and we consider only polynomial functionals, i.e. $F^{(n)}(\phi)$ vanishes for $n>N$ and some $N$. We define the support of a functional as 
\begin{align}
\supp F\doteq\{ & x\in M\,|\,\forall \text{ neighbourhoods }U\text{ of }x\ \exists \phi,\psi\in\Se(M),\,\supp\,\psi\subset U,
\\ & \text{ such that }F(\phi+\psi)\neq F(\phi)\}\,.\nonumber
\end{align}
which coincides with the union of the supports of $F^{(1)}(\phi)$ over all $\phi\in\Se(M)$. The relevant space of functionals which encompasses all observables of the free neutral scalar field is the space of {\em microcausal functionals}\index{microcausal functionals}\index{Fmuc@$\Fmuc$, the space of microcausal functionals}
\begin{align}\label{def:microcausal functionals}
\Fmuc &\doteq
\left\{
F:\Se(M)\to\bC\, |\, F\textrm{ smooth, compactly supported, }\right.\notag\\ 
&\qquad\qquad \left.\WF\left(F^{(n)}\right)\cap\left(\overline{V}^n_+\cup\overline{V}^n_-\right)=\emptyset
\right\},
\end{align}
where $V_{+/-}$ is a subset of the cotangent space formed by the elements whose covectors are contained in the future/past light cones and $\overline{V}_{+/-}$ denotes is closure. $\Fmuc$ contains two subspaces of importance. On the one hand, it contains the space $\Floc$ of {\em local functionals}\index{local functionals}\index{Floc@$\Floc$, the space of local functionals} consisting of sums of functionals of the form 
$$
F(\phi)=\int\limits_M \vol \cP[\phi]_{\mu_1\dots\mu_n}(x)\; f^{\mu_1\dots\mu_n}(x)
$$
where $\cP[\phi]\in\Se(T^*M^n)$ is a tensor such that $\cP[\phi](x)$ is a (tensor) product of covariant derivatives of $\phi$ at the point $x$ with a total order of $n$ and $f\in\Sec(TM^n)$ is a test tensor. The prime example of a local functional is  a smeared field monomial
\begin{gather}\label{eq:smearedpolynomial}
F_{k,f}(\phi)\doteq
\int\limits_M \vol \phi(x)^k f(x)\simeq \phi^k(f)\,,\qquad f\in\Sec^\bC(M)\,.
\end{gather}
One the other hand, $\Fmuc$ contains the space $\Freg$ of {\em regular functionals}\index{regular functionals}\index{Freg@$\Freg$, the space of regular functionals}, i.e. all microcausal functionals whose functional derivatives are smooth such that $F^{(n)}(\phi)\in\Sec^\bC(M^n)$ for all $\phi$ and all $n$. A prime example of a regular functional is a functional of the form
\begin{gather}\label{eq:regularfunctional}
F(\phi)=\prod^n_{j=1}\ip{f_j,\phi}\,,\qquad f_1,\dots,f_n\in\Sec^\bC(M)\,.
\end{gather}
Linear functionals are the only functionals which are both local and regular.

Given a bidistribution $H\in\Sec^{\prime\bC}(M^2)$ which a) satisfies the Hadamard condition Definition \ref{def_HadamardScalar}, b) has the antisymmetric part $H(x,y)-H(y,x)=iE(x,y)$ defined by the causal propagator $E$ and a real symmetric part
$$H_\sym(x,y)\doteq \frac12\left(H(x,y)+H(y,x)\right)\,, $$
and c) is a bisolution of the Klein--Gordon equation $P_xH(x,y)=P_y H(x,y)=0$,
we define a product indicated by $\star_H$\index{\$starH@$\star_H$, $\star$--product on $\Fmuc$ induced by $H$} on $\Fmuc$ via
\beq\label{def_startproduct}
(F\star_H G)(\phi) \doteq \sum^\infty_{n=0}\frac{1}{n!}\ip{F^{(n)}(\phi),H^{\otimes n}G^{(n)}(\phi)},
\eeq
i.e. by the sum of all possible mutual contractions of $F$ and $G$ by means of $H$. This is just an elegant way to implement Wick's theorem by an algebraic product as we shall explain in the following and may be understood in terms of {\em deformation quantization}, cf. \cite{chap2_Fredenhagen:2015iia} for a review of this aspect. Owing to the Hadamard property of $H$ and the microlocal properties of microcausal functionals, the $\star_H$--product is well--defined on $\Fmuc$ by \cite[Theorem 8.2.13]{chap2_Hormander} and one can show that $F\star_H G\in\Fmuc$ for all $F,G\in\Fmuc$, i.e. $\Fmuc$ is closed under $\star_H$.

It is not necessary to require that $H$ is positive and thus the two--point function $\omega_2$ of a Hadamard state $\omega$ on $\cA(M)$. However, independent of whether or not we require $H$ to be positive, its real and symmetric part is not uniquely fixed by the conditions we imposed. Given a different $H^\prime$ satisfying these conditions, it follows that $d\doteq H^\prime - H=H^\prime_\sym-H_\sym$ is real, symmetric and smooth and that the product $\star_{H^\prime}$ is related to $\star_{H}$ by
\beq\label{eq_stariso} 
F\star_{H^\prime} G = \alpha_d\left(\alpha_{-d}(F)\star_H\alpha_{-d}(H)\right)\,,
\eeq
where $\alpha_d:\Fmuc\to\Fmuc$ is the `contraction exponential operator'  
\beq\label{def_alpha}
\alpha_d \doteq \exp\left(\int\limits_{M^2}\vol d_gy \; d(x,y) \;\frac{\delta}{\delta \phi(x)}\frac{\delta}{\delta \phi(y)}\right)\,.
\eeq

The previous discussion implies that, given a $H$ satisfying the above--mentioned conditions, we can define a meaningful off--shell algebra $\cW^0_H(M)\doteq (\Fmuc,\star_H)$, and a corresponding on--shell algebra $\cW_H(M)\doteq \cW^0_H(M)/\cI$, where $\cI$ is the ideal generated by $\Fmucsol\subset \Fmuc$, the microcausal functionals which vanish on on--shell configurations $\phi\in\Sol\subset \Se(M)$. In fact, we have $\cW_H(M)=(\Fmuc/\Fmucsol,\star_H)$ and in the following we shall indicate an equivalence class $[F]\in \Fmuc/\Fmucsol$ by $F$ for simplicity. 

As a further implication of the previous exposition we have that $\cW_H(M)$ and $\cW_{H^\prime}(M)$ constructed with a different $H$ of the required type are isomorphic via $\alpha_d:\cW_H(M)\to \cW_{H^\prime}(M)$ with $d=H^\prime-H$. In this sense, we can consider $\cW_H(M)$ as a particular representation of an abstract algebra $\cW(M)$ which is independent of $H$, cf. \cite{chap2_Brunetti:2009qc}. $\cW_H(M)$ is in fact a $*$--algebra and $\alpha_d:\cW_H(M)\to \cW_{H^\prime}(M)$ is a $*$--isomorphism if we define the involution ($*$--operation) on $\Fmuc$ via complex conjugation by
$$
F^*(\phi)\doteq \overline{F(\phi)}
$$
which implies
$$
(F\star_H G)^* \doteq G^* \star_H F^* 
$$
by the conditions imposed on $H$. $\cW_H(M)$ may be endowed with a topology induced by the so--called H\"ormander topology \cite{chap2_BF00,chap2_Dabrowski:2013fua}, and one can show that all continuous states on $\cW_H(M)$ are induced by Hadamard states on the Borchers--Uhlmann algebra $\cA(M)$, cf.  \cite{chap2_HollandsRuan} (in combination with \cite{chap2_Sanders:2009sw}). In fact, we shall now explain why $\cW_H(M)$ may be considered as the `algebra of Wick polynomials' and in particular as an extension of $\cA(M)$.

To this avail, we first consider two linear functionals $\phi(f_i)\doteq F_{1,f_i}(\phi)$, $i=1,2$, cf. \eqref{eq:smearedpolynomial}, i.e. the classical field smeared with $f_1,f_2\in\Sec^\bC(M)$. The definition of the $\star_H$--product \eqref{def_startproduct} implies 
$$
\left[\phi(f_1),\phi(f_2)\right]_{\star_H}\doteq \phi(f_1)\star_H\phi(f_2)-\phi(f_2)\star_H\phi(f_1)= iE(f_1,f_2)\,.
$$
This indicates that the product $\star_H$ encodes the correct commutation relations among quantum observables. If we consider instead the quadratic local functionals $\phi^2(f_i)\doteq F_{2,f_i}(\phi)$, $i=1,2$, cf. \eqref{eq:smearedpolynomial}, i.e. the pointwise square of the classical field smeared with $f_1,f_2\in\Sec^\bC(M)$, we find
\begin{align*}
\phi^2(f_1)\star_H \phi^2(f_2) &= \phi^2(f_1)\phi^2(f_2)\\
 &\qquad + 4\int\limits_{M^2}\vol d_gy \; d(x,y) \; \phi(x)\phi(y)H(x,y)f_1(x)f_2(y) + 2 H^2(f_1,f_2)\,,
\end{align*}
which we may formally write as
$$
\phi^2(x)\star_H \phi^2(y) = \phi^2(x)\phi^2(y)+ 4 \phi(x)\phi(y)H(x,y) + 2 H^2(x,y)\,.
$$
This expression may be compared by the expression obtain via Wick's theorem
$$
\wick{\phi^2(x)}_H \wick{\phi^2(y)}_H = \wick{\phi^2(x)\phi^2(y)}_H+ 4 \wick{\phi(x)\phi(y)}_H H(x,y) + 2 H^2(x,y)\bI\,,
$$
if we define $\wick{\cdot}_H$ to be the Wick--ordering w.r.t. the symmetric part $H_\sym$ of $H$, e.g.
$$
\wick{\phi(x)\phi(y)}_H\doteq \phi(x)\phi(y)-H_\sym(x,y)\bI\,, \qquad \wick{\phi^2(x)}_H = \lim_{x\to y }\wick{\phi(x)\phi(y)}_H\,.
$$
Consequently, local functionals $F\in\Floc$ considered as elements of $F\in\cW_H(M)$ correspond to Wick polynomials Wick--ordered with respect to $H_\sym$, formally one may write $\wick{F}_H = \alpha_{-H_\sym}(F)$ with $\alpha_{-H_\sym}$ defined as in \eqref{def_alpha}. In fact, the algebra $F\in\cW_H(M)$ contains also time--ordered products of Wick polynomials \cite{chap2_BF00,chap2_HW01,chap2_HW04} and the perturbative construction of QFT based on $\cW_H(M)$ implies that `tadpoles' are already removed.

We have anticipated that $\cW_H(M)$ is an extension of the Borchers--Uhlmann algebra $\cA(M)$, which contains only products of the quantum field $\phi$ at {\em different} points. To see this, we consider the algebra $\cA^\prime(M)\doteq \alpha_{-H_\sym}\left((\Freg/\Fregsol,\star_H)\right)$, where $\Fregsol\doteq \Freg\cap \Fmucsol$, $(\Freg/\Fregsol,\star_H)$ is a subalgebra of $\cW_H(M)=(\Fmuc/\Fmucsol,\star_H)$ because $\star_H$ closes on regular functionals and we note that $\alpha_{-H_\sym}$ is well--defined on $\Freg$ although $H_\sym$ is not smooth. One can check that $\cA^\prime(M)$ is in fact the algebra $\cA^\prime(M)=(\Freg/\Fregsol,\star_{E})$ with the product $\star_{ E}$ given by \eqref{def_startproduct} with $H$ replaced by $i E/2$, and that $\cA^\prime(M)$ is isomorphic to the Borchers--Uhlmann algebra $\cA(M)$ defined in Definition \ref{def_BorchersUhlmannScalar}.

We have mentioned that Hadamard states $\omega$ on $\cA(M)$ induce meaningful states on $\cW_H(M)$ and that in fact all reasonable states on $\cW_H(M)$ are of this form. To explain this in more detail, we consider a Gaussian Hadamard state on $\omega$ on $\cA(M)$ with two--point function $\omega_2$ and the algebra $\cW_{\omega_2}(M)$ constructed by means of $\omega_2$. Given this, we can define a Gaussian Hadamard state on $\cW_{\omega_2}(M)$ by
$$
\omega(F)\doteq F(\phi=0)\,,\qquad \forall \;F\in\Fmuc
$$
which corresponds to the fact that Wick polynomials Wick--ordered w.r.t. to $\omega_2$ have vanishing expectation values in the state $\omega$, e.g. $\omega(\wick{\phi^2(x)}_{\omega_2})=0$. Note that this definition implies in particular that 
$$
\omega\left(\phi(f_1)\star_{\omega_2}\phi(f_2)\right)=\omega_2(f_1,f_2)\,,
$$
i.e. $\omega_2$ is, as required by consistency, the two--point correlation function of $\omega$ also in the functional picture. If we prefer to consider the extended algebra $\cW_H(M)$ constructed with a different $H$, then we can define the state $\omega$ on $\cW_H(M)$ by a pull--back with respect to the isomorphism $\alpha_d:\cW_H(M)\to\cW_{\omega_2}(M)$ with $d=\omega_2-H$ and $\alpha_d$ as in \eqref{def_alpha}. In other words, $\omega\circ \alpha_d$ defines a state on $\cW_H(M)$, and this definition corresponds to e.g. 
$$
\omega(\wick{\phi^2(x)}_H) = \lim_{x\to y}(\omega_2(x,y)-H(x,y))\,,
$$
i.e. to a {\em point--splitting renormalisation}\index{point--splitting renormalisation} of the expectation value of the observable $\phi^2(x)$.

We recall that observable quantities should be local and covariant fields as discussed in the previous section. However, not all local elements of the algebra $\cW_H(M)$ satisfy this property, i.e. not all local functionals $F\in\Floc$ considered as elements of $\cW_H(M)$ correspond to local and covariant Wick polynomials. In particular the functional $\phi^2(f)=F_{2,f}(\phi)\simeq \wick{\phi^2(f)}_H$, cf. \eqref{eq:smearedpolynomial}, does not correspond to a local and covariant Wick--square because $H$ is by assumption a bisolution of the Klein--Gordon equation and thus $\wick{\phi^2(f)}_H$ does not only depend on the geometry, i.e. the metric and its derivatives, in the localisation region of the test function $f$, but also on the geometry of the spacetime $(M,g)$ outside of the support of $f$ \cite{chap2_HW01}. This is related to the observation that quite generally local and covariant Hadamard states do not exist \cite{chap2_Fewster:2011pe}. Notwithstanding, local and covariant elements of $\cW_H(M)$ do exist and, following \cite{chap2_HW01,chap2_HW04}, they can identified by means of a map $W_H:\Floc\to\Floc\subset\cW_H(M)$ which should satisfy a number of conditions:
\begin{enumerate}
\item $W_H$ commutes with functional derivatives, i.e. $(W_H(F))^{(1)}=W_H(F^{(1)})$ and with the involution $W_H(F)^*=W_H(F^*)$.
\item $W_H$ satisfies the Leibniz rule.
\item $W_H$ is local and covariant.
\item $W_H$ scales almost homogeneously with respect to constant rescalings $m\mapsto \lambda m$ and $g\mapsto \lambda^{-2} g$ of the mass $m$ and metric $g$.
\item $W_H$ depends smoothly or analytically on the metric $g$, the mass $m$ and the coupling $\xi$ to the scalar curvature present in the Klein--Gordon equation \eqref{eq_KleinGordenMass}.
\end{enumerate}
We refer to \cite{chap2_HW04} for a detailed discussion of these conditions, and only sketch their meaning and physical motivation. To this avail, we consider the smeared local polynomials $\phi^k(f)=F_{k,f}(\phi)$ defined in \eqref{eq:smearedpolynomial} and identify $W_H(\phi^k(x))$ by $\wick{\phi^k(x)}$ omitting the smearing function for simplicity. The first of the above axioms then imply that $[\wick{\phi^k(x)},\phi(y)]=i m E(x,y)\wick{\phi^{m-1}(x)}$, i.e. the Wick--ordered fields satisfy standard commutation relations. The Leibniz rule further demands that $\nabla_\mu\wick{\phi^k(x)}=\wick{\nabla_\mu(\phi^k(x))}$, i.e. Wick--ordering  commutes with covariant derivatives, and the locality and covariance condition demands that $\wick{\phi^k(x)}$ is a local and covariant field in the sense of Definition \ref{def_LCQField}. Finally, the scaling condition requires that under constant rescalings $m\mapsto \lambda m$ and $g\mapsto \lambda^{-2} g$, $\wick{\phi^k(x)}$ scales (in four spacetime dimensions) as $\wick{\phi^k(x)}\mapsto \lambda^k \wick{\phi^k(x)} + \cO(\log \lambda)$, which among other things implies that $\wick{\phi^k(x)}$ has the correct `mass dimension', and the smoothness / analyticity requirement implies that e.g. $\wick{\phi^2(x)}$ may not contain a term like e.g. $\exp(\xi^{-1})m^{-4}R^{-1}(x)(R_{\mu\nu}(x)R^{\mu\nu}(x))^2 $ which would be allowed by the previous conditions.

It has been demonstrated in \cite{chap2_HW01,chap2_HW04} (see also \cite{chap2_Mo03}) that a prescription of defining local and covariant Wick polynomials exists, but that this prescription is not unique. In fact, if we consider (in a geodesically convex neighbourhood) a $H$ of the form \eqref{def_HadamardFormScalar} with $w=0$, i.e. a purely geometric Hadamard parametrix, then Wick--ordering w.r.t. to this $H$, e.g. $\wick{\phi^k(x)}\doteq \wick{\phi^k(x)}_H = \alpha_{-H_\sym}(\phi^k(x))$ satisfies all conditions reviewed above. However, this prescription is not the only possibility, but one can consider e.g.
$$
\wick{\phi^2(x)}\,=\, \wick{\phi^2(x)}_H + \alpha R(x) + \beta m^2
$$
with arbitrary real and dimensionless constants $\alpha$ and $\beta$ which are analytic functions of $\xi$. These constants parametrise the {\em renormalisation freedom} of Wick polynomials, or, in the context of perturbation theory, the renormalisation freedom inherent in removing tadpoles. Note that a change of scale $\lambda$ in \eqref{def_HadamardFormScalar} can be subsumed in this renormalisation freedom as a particular one--parameter family.

The coefficients parametrising the renormalisation freedom of local Wick poly-nomials may not be fixed within QFT on curved spacetimes, but have to be determined in a more general framework or by comparison with experiments. We will comment further on this point when discussing the renormalisation freedom of the stress--energy tensor in the context of cosmology in the next chapter. Note that, by local covariance, these coefficients are universal and may be fixed once and for all in all globally hyperbolic spacetimes (of the same dimension). Admittedly, in view of the above presentation of locality and general covariance, one might think that this holds only for spacetimes with isometric subregions (or spacetimes with conformally related subregions on account of \cite{chap2_NicolaConformal}). However, given two spacetimes $(M_1,g_1)$ and $(M_2,g_2)$ with not necessarily isometric subregions, one can employ the deformation argument of \cite{chap2_Fulling0} to deform, say, $(M_1,g_1)$ such that it contains a subregion isometric to a subregion of $(M_2,g_2)$. As the renormalisation freedom is parametrised by constants multiplying curvature terms or dimensionful constants which maintain their form under such a deformation, one can require that the mentioned constants are the same on $(M_1,g_1)$ and $(M_2,g_2)$ in a meaningful way.

\subsection{The Semiclassical Einstein Equation and Wald's Axioms}
\label{sec:see}

The central equation in describing the influence of quantum fields on the background spacetime -- i.e. their backreaction -- is the {\em semiclassical Einstein equation}\index{semiclassical Einstein equation}. It reads
\beq\label{eq_Semiclassical}G_{\mu\nu}(x)=8\pi G \;\omega(\wick{T_{\mu\nu}(x)})\,,\eeq
where the left hand side is given by the standard Einstein tensor $G_{\mu\nu}=R_{\mu\nu}-\frac12 R g_{\mu\nu}$, $G$ denotes Newton's gravitational constant, and we have replaced the stress--energy tensor of classical matter by the expectation value of a suitable Wick polynomial $\wick{T_{\mu\nu}(x)}$ representing the quantum stress--energy tensor evaluated in a state $\omega$. Considerable work has been invested in analysing how such equation can be derived via a suitable semiclassical limit from some potential quantum theory of gravity. We refer the reader to \cite[Section II.B]{chap2_FW96} for a review of several arguments and only briefly mention that a possibility to derive \eqref{eq_Semiclassical} is constituted by starting from the sum of the {\it Einstein--Hilbert action} $S_\mathrm{EH}(g)$ and the matter action $S_\text{matter}(g,\Phi)$\,,
\beq
\label{eq_EinsteinHilber} S(g,\Phi)\doteq S_\mathrm{EH}(g) - S_\text{matter}(g,\Phi)\,,
\eeq
$$S_\mathrm{EH}(g)\doteq \frac{1}{16\pi G}\int\limits_M \vol R =\frac{1}{16\pi G}\int\limits_M dx \sqrt{|\det g|}\;R$$
\noindent formally expanding a quantum metric and a quantum matter field around a classical (background) vacuum solution of Einstein's equation, and computing the equation of motion for the expected metric while keeping only `tree--level' ($\hbar^0$) contributions of the quantum metric and `loop--level' ($\hbar^1$) contributions of the quantum matter field. In this work, we shall not contemplate on whether and in which situation the above mentioned `partial one--loop approximation' is sensible, but we shall take the following pragmatic point of view: \eqref{eq_Semiclassical} seems to be the simplest possibility to couple the background curvature to the stress--energy of a quantum field in a non--trivial way. We shall therefore consider \eqref{eq_Semiclassical} as it stands and only discuss for which quantum states and Wick polynomial definitions of $\wick{T_{\mu\nu}(x)}$ it is a self-consistent equation. A rigorous proof that solutions of the semiclassical Einstein equation actually exist in the restricted case of cosmological spacetimes has been given in \cite{chap2_Pinamonti:2010is,chap2_Pinamonti:2013wya}.

We first observe that in \eqref{eq_Semiclassical} one equates a `sharp' classical quantity on the left hand side with a `probabilistic' quantum quantity on the right hand side. The semiclassical Einstein equation can therefore only be sensible if the fluctuations of the stress--energy tensor $\wick{T_{\mu\nu}(x)}$ in the considered state $\omega$ are small. In this respect, we already know that we should consider $\omega$ to be a Hadamard state and $\wick{T_{\mu\nu}(x)}$ to be a Wick polynomial Wick--ordered by means of a Hadamard bidistribution. Namely, the discussion in the previous sections tells us that this setup at least assures {\em finite fluctuations} of $\wick{T_{\mu\nu}(x)}$ as the pointwise products appearing in the computation of such fluctuations are well--defined distributions once their Hadamard property is assumed. In fact, this observation has been the main motivation to consider Hadamard states in the first place \cite{chap2_Wald2,chap2_WaldBook2}. However, it seems one can {\it a priori} not obtain more than these qualitative observations, and that quantitative statements on the actual size of the fluctuations can only be made {\it a posteriori} once a solution of \eqref{eq_Semiclassical} is found. An extended framework where the left hand side of the semiclassical Einstein equation is also interpreted stochastically is discussed in \cite{chap2_Pinamonti:2013zba}.   

Having agreed to consider only Hadamard states and Wick polynomials constructed by the procedures outlined in the last section, two questions remain. Which Hadamard state and which Wick polynomial should one choose to compute the right hand side of \eqref{eq_Semiclassical}? The first question can ultimately only be answered by actually solving the semiclassical Einstein equation. Observe that this actually poses a non--trivial problem as the formulation of the semiclassical Einstein equation in principle requires to specify a map which assigns to a metric $g$ a Hadamard state $\omega_g$, whereas we know that a covariant assignment of Hadamard states to spacetimes does not exist \cite{chap2_Fewster:2011pe}. This problem can be partially overcome by defining such a map only on a particular subset of globally hyperbolic spacetimes as we shall see in Section \ref{sec:cosmostates}. The question of which Wick polynomials should be taken as the definition of a quantum stress--energy tensor is also non--trivial, as Wick--ordering turns out to be ambiguous in curved spacetimes, see \cite{chap2_HW01,chap2_HW04} and the last section. We have already pointed out at several occasions that one should define the Wick polynomial $\wick{T_{\mu\nu}(x)}$ in a local, and, hence, state--independent way. In the context of the semiclassical Einstein equation the reason for this is the simple observation that one would like to solve \eqref{eq_Semiclassical} without knowing the spacetime which results from this procedure beforehand, but a state solves the equation of motion and, hence, already `knows' the full spacetime, thus being a highly non--local object. On account of the above considerations, we can therefore answer the question for the correct Wick polynomial representing $\wick{T_{\mu\nu}(x)}$ without having to choose a specific Hadamard state $\omega$ beforehand. The following review of the quantum stress--energy tensor will be limited to the case of the free neutral Klein--Gordon field. An analysis of the case of Dirac fields from the perspective of algebraic QFT on curved spacetimes may be found in \cite{chap2_Dappiaggi:2009xj}, whereas the case of interacting scalar fields in treated in \cite{chap2_HW04}.

To this avail, let us consider the stress--energy tensor of classical matter fields. Given a classical action $S_\text{matter}$, the related (Hilbert) stress--energy tensor can be computed as \cite{chap2_WaldBook,chap2_Forger}
\beq\label{eq_ClassicalStressTensor}T_{\mu\nu}\doteq \frac{2}{\sqrt{|\det g|}}\frac{\delta S_\text{matter}}{\delta g^{\mu\nu}}\,.\eeq
For the Klein--Gordon action $S(g,\phi)=\frac12 \ip{\phi,P\phi}$ we find 
\begin{align}\label{eq_STKG}
T_{\mu\nu}=&(1-2\xi)\phi_{;\mu}\phi_{;\nu}-2\xi\phi_{;\mu\nu}\phi+\xi G_{\mu\nu}\phi^2\\
&\qquad +g_{\mu\nu}\left\{2\xi (\square \phi)\phi+\left(2\xi-\frac12\right)\phi_{;\rho} \phi^{\phantom{;}\rho}_{;}-\frac12 m^2\phi^2\right\}\notag\,.\end{align}

A straightforward computation shows that the classical stress--energy tensor is covariantly conserved on--shell, i.e.
$$\nabla^{\mu}T_{\mu\nu}=-(\nabla_\nu \phi) P \phi=0\,.$$
Moreover, a computation of its trace yields
$$g^{\mu\nu}T_{\mu\nu}=(6\xi-1)\left(\phi\square\phi+\phi_{;\mu}\phi_{;}^{\phantom{;}\mu}\right)-m^2\phi^2-\phi P\phi\,.$$
Particularly, we see that in the conformally invariant situation, that is, $m=0$ and $\xi=\frac16$, the classical stress--energy tensor has vanishing trace on--shell. In fact, one can show that this is a general result, namely, the trace of a classical stress--energy tensor is vanishing on--shell if and only if the respective field is conformally invariant \cite[Theorem 5.1]{chap2_Forger}.

In view of the discussion of local and covariant Wick polynomials in the previous section, it seems natural to define the quantum stress--energy tensor $\wick{T_{\mu\nu}(x)}$ just as the Wick polynomial $\wick{T_{\mu\nu}(x)}_H$ Wick--ordered with respect to a purely geometric Hadamard parametrix $H=\frac{h}{8\pi^2}$ of the form \eqref{def_HadamardFormScalar} with $w=0$. Given a Hadamard state $\omega$ whose two--point function is (locally) of the form $\omega_2 = H + \frac{w}{8\pi^2}$, the expectation value of $\wick{T_{\mu\nu}(x)}_H$ in this state may be computed as
$$
\omega\left(\wick{T_{\mu\nu}(x)}_H\right)=\lim_{x\to y}\frac{D^\mathrm{can}_{\mu\nu}(x,y)w(x,y)}{8\pi^2}=\frac{\left[D^\mathrm{can}_{\mu\nu} w\right]}{8\pi^2}\,,
$$
where the bidifferential operator $D^\mathrm{can}_{\mu\nu}$ may be obtained from the classical stress--energy tensor as 
\begin{align*}
D^\mathrm{can}_{\mu\nu}(x,y)=\frac12 \frac{\delta^2 T_{\mu\nu}}{\delta \phi(x)\delta \phi(y)}&=
 (1-2\xi)g^{\nu^\prime}_\nu\nabla_{\mu}\nabla_{\nu^\prime}-2\xi\nabla_{\mu}\nabla_\nu+\xi G_{\mu\nu}\\
 &\qquad +g_{\mu\nu}\left\{2\xi \square_x+\left(2\xi-\frac12\right)g^{\rho^\prime}_\rho\nabla^\rho\nabla_{\rho^\prime}-\frac12 m^2\right\}\,,
\end{align*}
and where we recall the Synge bracket notation for coinciding point limits of bitensors, cf. Section \ref{sec:hadamard}. Here, unprimed indices denote covariant derivatives at $x$, primed indices denote covariant derivatives at $y$ and $g^{\nu^\prime}_\nu$ is the bitensor of parallel transport, cf. Section \ref{sec:hadamard}. However, as we shall discuss in a bit more detail in the following, this particular definition of a quantum stress--energy tensor is not satisfactory because it does not yield a covariantly conserved quantity which is a necessary condition for the semiclassical Einstein equation to be well--defined.

The first treatment of the quantum stress--energy tensor from an algebraic point of view was the analysis of Wald in \cite{chap2_Wald2}. At the time \cite{chap2_Wald2} appeared, workers in the field had computed the expectation value of the quantum stress--energy tensor by different renormalisation methods like adiabatic subtraction, dimensional regularisation, $\zeta$-function regularisation and DeWitt--Schwinger point--splitting regularisation (see \cite{chap2_Birrell} and also \cite{chap2_MorettiZeta,chap2_Hack:2012dm2}) and differing results had been found. From the rather modern point of view we have reviewed in the previous sections, it is quite natural and unavoidable that renormalisation in curved spacetimes is ambiguous. However, the axioms for the (expectation value of) the quantum stress--energy tensor introduced in \cite{chap2_Wald2} helped to clarify the case at that time  and to understand that in principle {\em all} employed renormalisation schemes were correct in physical terms. Consequently, the apparent differences between them could be understood on clear conceptual grounds. These axioms (in the updated form presented in \cite{chap2_WaldBook2}) are:\index{Wald's axioms}

\begin{enumerate}
 \item Given two (not necessarily Hadamard) states $\omega$ and $\omega^\prime$ such that the difference of their two--point functions $\omega_2-\omega^\prime_2$ is smooth, $\omega(\wick{T_{\mu\nu}(x)})-\omega^\prime(\wick{T_{\mu\nu}(x)})$ is equal to
$$\left[D^\text{can}_{\mu\nu}\;\left(\omega_2(x,y)-\omega^\prime_2(x,y)\right)\right]\,.$$
\item $\omega(\wick{T_{\mu\nu}(x)})$ is locally covariant in the following sense: let
$$\chi: (M,g) \hookrightarrow (M^\prime,g^\prime)$$ be defined as in Section \ref{sec:lcqft} and let $\alpha_\chi$ denote the associated continuous, unit-preserving, injective $*$--morphisms between the relevant enlarged (abstract) algebras $\cW(M,g)$ and $\cW(M^\prime,g^\prime)$. If two states $\omega$ on $\cW(M,g)$ and $\omega^\prime$ on $\cW(M^\prime,g^\prime)$ respectively are related via $\omega=\omega^\prime\circ \alpha_\chi$, then
$$\omega^\prime\left(\wick{T_{\mu^\prime\nu^\prime}(x^\prime)}\right)=\chi_*\omega\left(\wick{T_{\mu\nu}(x)}\right)\,,$$
where $\chi_*$ denotes the push--forward of $\chi$ in the sense of covariant tensors.
\item Covariant conservation holds, i.e.
$$\nabla^\mu\omega\left(\wick{T_{\mu\nu }(x)}\right)=0\,.$$
\item In Minkowski spacetime $\bM$, and in the relevant Minkowski vacuum state $\omega_\bM$
$$\omega_\bM\left(\wick{T_{\mu\nu}(x)}\right)=0\,.$$
\item $\omega(\wick{T_{\mu\nu}(x)})$ does not contain derivatives of the metric of order higher than 2.
\end{enumerate}

Some of these axioms are just special cases of the axioms for local Wick polynomials reviewed in Section \ref{sec:wick}. In fact, the first of these axioms is just a variant of the requirement that local Wick polynomials have standard commutation relations. In the case of $\wick{T_{\mu\nu}}$ this implies that two valid definitions of this observable can only differ by a term proportional to the identity. The second axiom is just the locality and covariance of Wick polynomials here formulated on the level of expectation values. The condition that the quantum--stress energy tensor has vanishing expectation value on Minkowski spacetime in the corresponding vacuum state is not compatible with the requirement that local Wick polynomials depend smoothly on the mass $m$, cf. \cite[Theorem 2.1 (e)]{chap2_Mo03} and can be omitted because it is not essential. The last axiom is motivated by the wish to ensure that the solution theory of the semiclassical Einstein equation does not depart `too much' from the one of the classical Einstein equation. In particular one would like to have that all solutions of the semiclassical Einstein equation behave well in the classical limit $\hbar \to 0$. Wald himself had realised, however, that this axiom could not be satisfied for massless theories without introducing an artificial length scale into the theory; therefore, the axiom has been discarded.

Using these axioms as well as a variant of the scaling requirement for local Wick polynomials, Wald could prove that a uniqueness result for $\omega(\wick{T_{\mu\nu}})$ (and thus for $\wick{T_{\mu\nu}}$ itself by the first axiom) can be obtained, namely that two valid definitions $\wick{T_{\mu\nu}}$ and $\wick{T_{\mu\nu}}^\prime$  of the quantum stress--energy tensor can only differ by a term of the form\index{renormalisation freedom of the quantum stress--energy tensor}
\beq\label{eq_renfreedom}
\wick{T_{\mu\nu}}^\prime - \wick{T_{\mu\nu}} = \alpha_1 m^4 g_{\mu\nu} + \alpha_2 m^2 G_{\mu\nu} + \alpha_3 I_{\mu\nu} + \alpha_4 J_{\mu\nu} + \epsilon K_{\mu\nu}\,,
\eeq
where $\alpha_i$ are real and dimensionless constants and the last three tensors appearing above are the conserved local curvature tensors
\begin{align} I_{\mu\nu}& \doteq \frac{1}{\sqrt{|\det g|}}\frac{\delta}{\delta g^{\mu\nu}}\int\limits_M\vol  \;R^2  = -g_{\mu\nu}\left(\frac{1}{2}R^2+2\square R\right)+2R_{;\mu\nu}+2R R_{\mu\nu}\notag\,,\\
J_{\mu\nu}&\doteq \frac{1}{\sqrt{|\det g|}}\frac{\delta}{\delta g^{\mu\nu}}\int\limits_M \vol \; R_{\alpha\beta} R^{\alpha\beta}\notag\\
&= -\frac{1}{2}g_{\mu\nu}(R_{\mu\nu} R^{\mu\nu}+\square R)+R_{;\mu\nu}-\square R_{\mu\nu} +2 R_{\alpha\beta}R^{\alpha\,\,\,\beta}_{\,\,\,\mu\,\,\,\nu}\,,\label{eq_ConservedLocal}\\
 K_{\mu\nu}&\doteq \frac{1}{\sqrt{|\det g|}}\frac{\delta}{\delta g^{\mu\nu}}\int\limits_M \vol  \;R_{\alpha\beta\gamma\delta} R^{\alpha\beta\gamma\delta} \notag\\ 
 &= -\frac{1}{2}g_{\mu\nu}R_{\alpha\beta\gamma\delta} R^{\alpha\beta\gamma\delta}+2R_{\alpha\beta\gamma\mu} R^{\alpha\beta\gamma}_{\phantom{\alpha\beta\gamma}\nu}+4 R_{\alpha\beta}R^{\alpha\phantom{\mu}\beta}_{\phantom{\alpha}\mu\phantom{\beta}\nu}\notag\\
 &\qquad\qquad -4R_{\alpha\mu}R^{\alpha}_{\,\,\,\nu}-4\square R_{\mu\nu}+2R_{;\mu\nu}\,.\notag
 \end{align}
This uniqueness result follows by using the first of Wald's axioms in order to observe that $\wick{T_{\mu\nu}}^\prime - \wick{T_{\mu\nu}}$ is a $c$--number. From the locality and covariance axiom  in combination with conservation it follows that this $c$--number must be a local and conserved curvature tensor whereas the scaling condition implies that it has the correct mass dimension. As \cite{chap2_Tichy} pointed out, these requirements do not fix $\wick{T_{\mu\nu}}^\prime - \wick{T_{\mu\nu}}$ to be of the above form, but demanding in addition that $\wick{T_{\mu\nu}}$ depends in a smooth or analytic way on $m$ and $g$ as in \cite{chap2_HW01,chap2_HW04} is sufficient to rule out the additional terms mentioned in \cite{chap2_Tichy} so that \eqref{eq_renfreedom} indeed classifies the full renormalisation freedom compatible with the axioms of local Wick polynomials introduced in \cite{chap2_HW01,chap2_HW04} and the conservation of $\wick{T_{\mu\nu}}$.

In fact, we will see in the next section that changing the scale $\lambda$ in the regularising Hadamard bidistribution amounts to changing $\wick{T_{\mu\nu}}$ exactly by a tensor of this
form and, furthermore, the attempt to renormalise perturbative Einstein--Hilbert quantum gravity at one loop order
automatically yields a renormalisation freedom in form of such a tensor as well \cite{chap2_tHooft0}\footnote{In fact, at least in the case of scalar fields, the combination of the local curvature tensors appearing as the finite renormalisation freedom in \cite{chap2_tHooft0} is, up to a term which seems to be an artifact of the dimensional regularisation employed in that paper, the same that one gets via changing the scale in the regularising Hadamard bidistribution.}. Moreover, using
the {\em Gauss--Bonnet--Chern theorem} in four dimensions, which states that  
$$\int\limits_{M}\vol \left( R_{\mu\nu\rho
\tau}R^{\mu\nu\rho\tau}-4R_{\mu\nu}R^{\mu\nu}+R^2\right)$$ 
is a topological invariant and, therefore, has a vanishing
functional derivative with respect to the metric \cite{chap2_Alty,chap2_tHooft0}, one can restrict the freedom even
further by removing $K_{\mu\nu}$ from the list of allowed local curvature tensors as $K_{\mu\nu}=4 J_{\mu\nu}-I_{\mu\nu}$. Finally, the above tensors all have a trace proportional to $\square R$ and thus the linear combination $I_{\mu\nu}-3 J_{\mu\nu}$ is traceless.

\subsection{A Conserved Quantum Stress--Energy Tensor}
\label{sec:stresstensor}

After some dispute about computational mistakes (see the discussion in \cite{chap2_Wald3}) it had soon be realised that the quantum stress--tensor $\wick{T_{\mu\nu}}_H$ defined by Wick--ordering the canonical expression by means of a purely geometric Hadamard bidistribution $H$ is not conserved although it satisfies the other conditions for local Wick polynomials mentioned at the end of Section \ref{sec:wick}. The reason for this is the fact that, in contrast to the two--point function of a Hadamard state, a purely geometric Hadamard bidistribution fails to satisfy the equation of motion. Consequently $\nabla^\mu\wick{T_{\mu\nu}}_H=\wick{\nabla^\mu T_{\mu\nu}}_H$ is (in four spacetime dimensions) the covariant divergence of a non--vanishing and non--conserved local curvature term $\nabla^\mu\wick{T_{\mu\nu}}_H = \nabla^\mu C_{\mu\nu}\neq 0$. The obvious solution to this problem has been to compute  $C_{\mu\nu}$ and then to define a conserved stress--tensor by 
\beq\label{eq_subtractconserved}
\wick{T_{\mu\nu}}\;\doteq\; \wick{T_{\mu\nu}}_H - C_{\mu\nu}\bI\,,
\eeq
i.e. by just subtracting this {\em conservation anomaly}. It was found that this conserved stress tensor still has a {\em trace anomaly}\index{trace anomaly}, i.e. the trace of any stress--energy tensor which satisfies the Wald axioms is non--vanishing in the conformally invariant case $m=0$, $\xi=\frac16$ \cite{chap2_Wald3}.

Two proposals have been made in order to motivate the ad--hoc subtraction \eqref{eq_subtractconserved} on conceptual grounds. In \cite{chap2_Mo03} it has been suggested that the Wick--ordering prescription should be kept, but that the classical expression $T_{\mu\nu}$ entering the definition of $\wick{T_{\mu\nu}}_H$ should be modified in such a way that it coincides with the canonical classical stress--energy tensor on--shell but gives a conserved observable upon quantization. In \cite{chap2_HW04} instead it has been argued that the Wick--ordering prescription should be modified without changing the classical expression. This point of view has the advantage that it fits into the general framework of defining local and covariant Wick polynomials introduced in \cite{chap2_HW01,chap2_HW04}. In particular it is possible to alter uniformly the definition of all local Wick polynomials induced by Wick--ordering with respect to a purely geometric Hadamard parametrix $H$ in such a way that a) all axioms for local Wick polynomials mentioned at the end of Section \ref{sec:wick} are still satisfied and b) the canonical stress--energy tensor Wick--ordered with respect to this prescription is conserved \cite{chap2_HW04}. In the aforementioned reference it has also been argued that this point of view has the further advantage that it is applicable even to perturbatively constructed interacting models. Notwithstanding, we shall review the approach of \cite{chap2_Mo03} in the following for ease of presentation.

To this avail, we modify the classical stress--energy tensor by setting 
$$T^c_{\mu\nu}\doteq T^\text{can}_{\mu\nu} + c g_{\mu\nu}\phi P \phi\,,$$
where $T^\text{can}_{\mu\nu}$ is the canonical expression and $c$ is a suitable constant to be fixed later. We then define $\wick{T_{\mu\nu}}$ by
$$
\wick{T_{\mu\nu}}\;\doteq\; \wick{T^c_{\mu\nu}}_H\,,
$$ 
where $\wick{\cdot}_H$ indicates Wick--ordering (i.e. point--splitting regularisation) with respect to a purely geometric $H=\frac{h}{8\pi^2}$ of the form \eqref{def_HadamardFormScalar} with $w=0$. Following the arguments of the previous section, the expectation value of $\wick{T_{\mu\nu}}$ in a Hadamard state $\omega$ whose two--point function is (locally) of the form $\omega_2 = H + \frac{w}{8\pi^2}$ may be computed as
\beq\label{def_improvedexpval}
\omega\left(\wick{T_{\mu\nu}(x)}\right)=\frac{\left[D^\mathrm{c}_{\mu\nu} w\right]}{8\pi^2}\,,
\eeq
where
\begin{align}\label{eq_ImprovedKGDiff}D^{c}_{\mu\nu}&\doteq D^\text{can}_{\mu\nu}+c g_{\mu\nu}P_x\\
 &=(1-2\xi)g^{\nu^\prime}_\nu\nabla_{\mu}\nabla_{\nu^\prime}-2\xi\nabla_{\mu}\nabla_\nu+\xi G_{\mu\nu}\notag\\&\quad+g_{\mu\nu}\left\{2\xi \square_x+\left(2\xi-\frac12\right)g^{\rho^\prime}_\rho\nabla^\rho\nabla_{\rho^\prime}-\frac12 m^2\right\}+c g_{\mu\nu}P_x\,.\notag
\end{align}
The following result can now be shown \cite{chap2_Mo03}.

\begin{theorem}
\label{thm_TensorScalar}
 Let $\omega(\wick{T_{\mu\nu}(x)})$ be defined as in \eqref{def_improvedexpval} with $c=\frac13$.
\begin{itemize}
 \item[a)] $\omega(\wick{T_{\mu\nu}(x)})$ is covariantly conserved, i.e. $$\nabla^\mu\omega\left(\wick{T_{\mu\nu}(x)}\right)=0\,.$$
 \item[b)] The trace of $\omega(\wick{T_{\mu\nu}(x)})$ equals
$$g^{\mu\nu}\omega\left(\wick{T_{\mu\nu}(x)}\right)=\frac{1}{4\pi^2}\left[v_1\right]-\frac{1}{8\pi^2}\left(3\left(\frac16-\xi\right)\Box+m^2\right)[w]$$
$$=\frac{1}{2880\pi^2}\left(\frac52\left(6\xi-1\right)R^2+6(1-5\xi)\Box R+ C_{\alpha\beta\gamma\delta}C^{\alpha\beta\gamma\delta}+R_{\alpha\beta}R^{\alpha\beta}-\frac{R^2}{3}\right)$$$$+\frac{1}{4\pi^2}\left(\frac{m^4}{8}+\frac{\left(6\xi-1\right)m^2R}{24}\right)-\frac{1}{8\pi^2}\left(3\left(\frac16-\xi\right)\Box+m^2\right)[w]\,,$$
which, for $m=0$ and $\xi=\frac16$, constitutes the {\em trace anomaly} of the quantum stress--energy tensor.
\item[c)] The conservation and trace anomaly are independent of the chosen scale $\lambda$ in the Hadamard parametrix $h$. Namely, a change
$$\lambda\;\mapsto\;\lambda^\prime$$
results in 
$$\omega(\wick{T_{\mu\nu}(x)})\;\mapsto\; \omega(\wick{T_{\mu\nu}(x)})^\prime=\omega(\wick{T_{\mu\nu}(x)})+\delta T_{\mu\nu}\,,
$$
where
\beq\label{eq_lchange}
\delta T_{\mu\nu}\doteq\frac{2\log{\lambda/\lambda^\prime}}{8\pi^2}\left[D^{c}_{\mu\nu}\;v\right]=\frac{2\log{\lambda/\lambda^\prime}}{8\pi^2}\left[D^\mathrm{can}_{\mu\nu}\;v\right]
\eeq
$$=\frac{2\log{\lambda/\lambda^\prime}}{8\pi^2}\left(\frac{m^2(6\xi-1)G_{\mu\nu}}{12}-\frac{m^4}{8}g_{\mu\nu}+\frac{1}{360}(I_{\mu\nu}-3J_{\mu\nu})-\frac{(6\xi-1)^2}{144}I_{\mu\nu}\right)
$$
is a conserved tensor which has vanishing trace for $m=0$ and $\xi=\frac16$.
\end{itemize}
\end{theorem}

\begin{proof}
\begin{enumerate}
\item Leaving $c$ undetermined and employing Synge's rule (cf.  Lemma \ref{lem_SyngeRule}), we compute
$$8\pi^2\nabla^\mu\omega\left(\wick{T_{\mu\nu}(x)}\right)=\nabla^\mu\left[D^{c}_{\mu\nu}\;w\right]=\left[(\nabla^\mu+g^\mu_{\mu^\prime}\nabla^{\mu^\prime})D^{c}_{\mu\nu}\;w\right]$$
$$=\left[-\left(g^{\nu^\prime}_\nu\nabla_{\nu^\prime}P_x+c(g^{\nu^\prime}_\nu\nabla_{\nu^\prime}P_x+\nabla_{\nu}P_x)\right)w\right]\,.$$
Let us now recall that $P_x(h+w)=0$ and, hence, $P_x w=-P_xh$. Inserting this and the identities found in Lemma \ref{lem_PHadamardScalar}, we obtain
$$8\pi^2\nabla^\mu\omega\left(\wick{T_{\mu\nu}(x)}\right)=-\left[\left(-g^{\nu^\prime}_\nu\nabla_{\nu^\prime}P_x+c(g^{\nu^\prime}_\nu\nabla_{\nu^\prime}P_x+\nabla_{\nu}P_x)\right)h\right]$$
$$=(6c-2)\left[v_1\right]_{;\nu}\,.$$
This proves the conservation for $c=\frac13$.
\item It is instructive to leave $c$ undetermined also in this case. Employing Synge's rule and the results of Lemma \ref{lem_PHadamardScalar}, we find
$$8\pi^2g^{\mu\nu}\omega\left(\wick{T_{\mu\nu}(x)}\right)=8\pi^2g^{\mu\nu}\left[D^{c}_{\mu\nu}\;w\right]$$
$$=(4c-1)[P_x w]-\left(3\left(\frac16-\xi\right)\Box+m^2\right)[w]$$
$$=(4c-1)6[v_1]-\left(3\left(\frac16-\xi\right)\Box+m^2\right)[w]\,.$$
Inserting $c=\frac13$ yields the wanted result.
\item The proof ensues without explicitly computing $\delta T_{\mu\nu}$ in terms of the stated conserved tensors from the following observation. Namely, a change of scale as considered transforms $w$ by a adding a term $2\log \lambda/\lambda^\prime v$. Hence, our computations for proving the first two statements entail
$$\frac{8\pi^2}{2\log \lambda/\lambda^\prime}\nabla^{\mu}\delta T_{\mu\nu}=\left[\left(-g^{\nu^\prime}_\nu\nabla_{\nu^\prime}P_x+c(g^{\nu^\prime}_\nu\nabla_{\nu^\prime}P_x+\nabla_{\nu}P_x)\right)v\right]\,,$$
$$\frac{8\pi^2}{2\log \lambda/\lambda^\prime}g^{\mu\nu}\delta T_{\mu\nu}=-(4c-1)[P_x v]-\left(3\left(\frac16-\xi\right)\Box+m^2\right)[v]\,.$$
The former term vanishes because $P_x v=0$ as discussed in Section \ref{sec:hadamard}, and the same holds for the latter term if we insert $\xi=\frac16$ and $m=0$.
\end{enumerate}
\end{proof}

The proof the second statement clearly shows that there is a possibility to assure vanishing trace in the conformally invariant case, but this possibility is not compatible with conservation. Moreover, we have stated the last result in explicit terms in order to show how a change of scale in the Hadamard parametrix is compatible with the renormalisation freedom of the quantum stress--energy tensor. Finally, we stress that the term added to the canonical stress--energy tensor is compatible with local covariance because the corresponding change of the quantum stress--energy tensor is proportional to $g_{\mu\nu}[v_1]$, i.e. a local curvature tensor.

We would also like to point out that the above explicit form of the trace anomaly has also been known before Hadamard point--splitting had been developed. Particularly, the same result had been obtained by means of so--called {\it DeWitt--Schwinger point--splitting} in \cite{chap2_Christensen0}. This renormalisation prescription is a priori not rigorously defined on Lorentzian spacetimes and the Hadamard point--splitting computation in \cite{chap2_Wald3} had therefore been the first rigorous derivation of the trace anomaly of the stress--energy tensor. However, DeWitt--Schwinger point--splitting can be reformulated on rigorous grounds, cf. \cite{chap2_Hack:2012dm2}.

\section{Further Reading}

The review of algebraic quantum field theory on curved spacetimes in this chapter has covered aspects of this framework which are relevant for the applications discussed in the following chapter. Recent reviews which deal with aspects and constructions not covered in the present chapter, or provide further details, are \cite{chap2_FredenhagenRejzner2, chap2_Hollands:2014eia} and  \cite{chap2_Fewster:2015kua, chap2_Benini:2015bsa, chap2_Khavkine:2014mta, chap2_Fredenhagen:2015iia}, which are part of \cite{chap2_bible}. A historical account of quantum field theory in curved spacetimes may be found in \cite{chap2_Wald:2006ty}.

\chapter{Cosmological Applications}
\label{sec:cosmo}

\abstract{In this chapter we discuss two cosmological applications of algebraic quantum field theory in curved spacetimes. In the Standard Model of Cosmology -- the $\Lambda$CDM--model --
the matter--energy content of the universe on large scales is modelled by a classical stress--energy tensor of perfect fluid form. Motivated by the fact that this matter--energy is considered to have a microscopic description in terms of a quantum field theory, we demonstrate as a first application how the classical perfect fluid stress--energy tensor in the the $\Lambda$CDM--model may be derived within quantum field theory on curved spacetimes by showing that there exist quantum states on cosmological spacetimes in which the expectation value of the quantum stress--energy tensor is qualitatively and quantitatively of the form assumed in the $\Lambda$CDM--model up to corrections which may have interesting phenomenological implications. In the simplest models of Inflation, it is assumed that a classical scalar field on a cosmological spacetime coupled to the metric via the Einstein equations drives an exponential phase of expansion in the early universe. As a second application, the standard approach to the quantization of the perturbations of this coupled system, which makes heavy use of the symmetries of cosmological spacetimes, is re--examined by comparing it with a more fundamental approach which consists of quantizing the perturbations of a scalar field and the metric field in a gauge--invariant manner on general backgrounds and then considering the symmetric cosmological backgrounds as a special case.}

\section{A Brief Introduction to Cosmology}
\label{sec:cosmobrief}

According to the well--known {\em cosmological principle}, our universe is {\em homogeneous} and {\em isotropic}. This postulate implies that, on large scales, the cosmos looks `the same' everywhere and in all directions, see \cite[Chapter 5]{chap3_WaldBook} for a precise definition and a discussion of these issues. A remarkable confirmation of the isotropy of our universe is the fact that the temperature of the {\em Cosmic Microwave Background} (CMB) is isotropic up to relative fluctuations of the order $10^{-5}$ \cite{chap3_Ade:2013sjv}. Based on the cosmological principle, we shall regard {\em Friedmann-Lema\^itre-Robertson-Walker (FLRW) spacetimes}\index{FLRW (Friedmann-Lema\^itre-Robertson-Walker) spacetimes} as the curved manifolds describing our universe on large scales, where here `large' means scales of the size of galaxy superclusters, i.e. around $10^8$ light--years or $10^{24}$m, which is just about a thousandth of the diameter of the observable universe! The underlying manifold of such spacetimes is $I_t\times \Sigma_\kappa$, where $I_t$ denotes an open interval in $\bR$ and $\Sigma_\kappa$ is a three--dimensional manifold of constant curvature $\kappa$, and their metric is given by the line element
$$ds^2=-dt^2+a^2(t)\left(\frac{dr^2}{1-\kappa r^2}+r^2d\bS^2(\theta,\varphi)\right).$$
Here, $a(t)$ is a strictly positive smooth function called the {\em scale factor}\index{a@$a(t)$, scale factor}\index{scale factor}, $t\in I_t$ denotes {\em cosmological time}\index{cosmological time}, and $\theta, \varphi$ are coordinates on the 2--sphere $\bS^2$, the canonical line element of which is denoted by $d\bS^2$. If we recall the discussion in Section \ref{sec:ghst}, we immediately realise that $\Sigma_\kappa$ (or rather $\{t\}\times \Sigma_\kappa$ for all $t\in I_t$) is a Cauchy surface, and hence all FLRW spacetimes are globally hyperbolic. In the following, we shall restrict attention to $\kappa=0$ as this is the situation strongly favoured by experimental data, the CMB measurements in particular \cite{chap3_Ade:2013zuv}. In this case, one speaks of {\em flat} FLRW spacetimes and $\Sigma_\kappa=\bR^3$, while $r$ denotes the Euclidean distance in $\bR^3$, which in the cosmological context is called the {\em comoving distance}\index{comoving distance}. 

Obviously, a possible time function on a FLRW spacetime is given by the cosmological time itself, but there are further possibilities, which are often more convenient: the {\em conformal time}\index{conformal time} $\eta$, the scale factor $a$ itself, and the {\em redshift}\index{redshift} $z$. These time variables are related by
\beq\label{eq_timevariables}
dt = a d\eta = \frac{da}{aH} = -\frac{dz}{(1+z)H}\,,
\eeq
$$
\Leftrightarrow\qquad \eta = \int^t_{t_0} \frac{1}{a(\tilde t)}d\tilde t\,,\qquad H\doteq \frac{\partial_t a}{a}\,, \qquad z=\frac{a_0}{a}-1
$$
where $H$ denotes the {\em Hubble rate}\index{Hubble rate} and $t_0$ and $a_0$ are conventional, as is the interval $I_\eta\ni \eta$ determined by $I_t$ and the above integral expression for $\eta$. In cosmology, $a_0$ is interpreted as the scale factor of today and usually chosen as $a_0=1$ so that the present redshift is $z=0$. The conformal time is always a well--defined time variable, whereas using $a$ and $z$ as time variables is only meaningful if the Hubble rate $H$ has a definite sign. Observations indicate $H>0$, i.e. an {\em expanding} universe, and we will make this assumption when discussing the cosmological evolution. In fact, as the name suggests, $z$ is a useful time variable in cosmology because it is actually a {\em direct observable} that can be measured by comparing the observed spectra of distant luminous objects with known spectra and assigning the `time--label' $z$ to the time when the light observed today was emitted from these objects. From the theoretical point of view, using $z$ as a time variable has the advantage that -- as we shall discuss in the following -- the (semiclassical) Einstein equation in the simplest cases becomes an ordinary differential equation for $H(z)$, whereas by using $t$ or $\eta$ as variables one has to deal with differential equations where both $H$ and $a$ appear as functions of $t$ or $\eta$. The latter variable is called conformal time because the metric of flat FLRW spacetimes in the Cartesian coordinate system $(\eta,\vec{x})\in I_\eta\times\bR^3$ reads
$$ds^2=a(\eta)^2\left(-d\eta^2+d\vec{x}^2\right)\,,$$
which indicates that flat FLRW spacetimes are {\em conformally flat}. The conformal time is a convenient time variable because one can eliminate terms which are first order in $\partial_t$ in the Klein--Gordon equation on FLRW by passing to $\eta$ as we shall see. In the following, we shall use $\dot{\;}$ and $^\prime$ to indicate derivatives with respect to $t$ and $\eta$ respectively.

The functional behaviour of the scale factor $a$ describes the `history' of our universe, which, according to General Relativity, is completely determined by the specification of the matter--energy  content of our universe in terms of the stress--energy tensor $T_{\mu\nu}$ and its coupling to gravity via the Einstein equation
\beq
\label{eq_EinsteinEquation}
G_{\mu\nu}=8\pi G T_{\mu\nu}\,.
\eeq
In order to obtain flat FLRW spacetimes as solutions of this equation, $T_{\mu\nu}$ must have the have the form of a stress--energy tensor of a {\em perfect fluid}\index{perfect fluid}, namely, in comoving coordinates $(t,\vec{x})$,
\beq
\label{eq_StressTensorPerfectFluid}T_{\mu\nu}=\mathrm{diag}(\varrho,p,p,p)
.\eeq
In \eqref{eq_StressTensorPerfectFluid}, the {\em energy density}\index{energy density} $\varrho$ and the {\em pressure}\index{pressure} $p$ of matter-energy are related by the {\em equation of state}\index{equation of state} \beq
\label{eq_EquationOfState}p=p(\rho)=w(\varrho)\varrho\,,
\eeq
where the last form of the equation of state is convenient because in cosmology $w(\varrho)$ is often a constant. 

In fact, according to the Standard Model of Cosmology\index{Standard Model of Cosmology} -- the $\Lambda$CDM--model\index{LCDM@$\Lambda$CDM--model} -- our universe contains matter, radiation, and Dark Energy, whose combined energy density determines the expansion of the universe. In the $\Lambda$CDM--model, these three kinds of matter--energy are modelled macroscopically by perfect fluids with equations of state $p = w \rho$, $w=0, \frac13, -1$ for matter, radiation and Dark Energy (assuming that the latter is just due to a cosmological constant), respectively. In the context of cosmology, the terms `matter' and `radiation' subsume all matter--energy with the respective macroscopic equation of state such that e.g. `radiation' does not comprise only electromagnetic radiation, but also the three left--handed neutrinos present in Standard Model of Particle Physics (SM) and possibly so--called Dark Radiation,  and `matter' includes both the baryonic matter which is well--understood in the SM and Dark Matter. Here, Dark Matter and Dark Radiation both quantify contributions to the macroscopic matter and radiation energy densities which exceed the ones expected from our knowledge of the SM and are believed to originate either from particles respectively fields not present in the SM or from geometric effects such as modifications of General Relativity. Observations indicate that the current matter--energy content of the universe is composed of roughly $30\%$ matter and $70\%$ Dark Energy, while the relative contribution of radiation is only of order $10^{-4}$, see e.g. \cite{chap3_Ade:2013zuv} for the exact numbers from the 2013 data release of the Planck satellite. The measured $w$ of Dark Energy is in good agreement with a constant $w=-1$ \cite{chap3_Ade:2013zuv}, where in this context Dark Energy is mostly just taken to be all matter--energy which does not behave macroscopically like matter or radiation.

In flat FLRW spacetimes, the Einstein equation reduces to the {\em Friedmann equations}\index{Friedmann equations}
\beq
\label{eq_FriedmannEquations}
H^2=\frac{8\pi G}{3}\varrho\,,\qquad \frac{\ddot{a}}{a}=-\frac{4\pi G}{3}(\varrho+3p)\,.\eeq
The Einstein equation is only consistent if the stress--energy tensor $T_{\mu\nu}$ is covariantly conserved $\nabla^\mu T_{\mu\nu}=0$. In (not necessarily flat) FLRW spacetimes, the covariant conservation of the stress-energy tensor implies
\beq
\label{eq_ConservationPerfectFluid}
\dot{\varrho}+3H(\varrho+p)=0\,,
\eeq
and this equation can be obtained directly from the Friedmann equations since $\nabla^\mu T_{\mu\nu}=0$ has been implicitly assumed in their derivation. Alternatively, given conservation of $T_{\mu\nu}$, it is sufficient to solve only the first Friedmann equation in \eqref{eq_FriedmannEquations}.

Under the assumption that interactions between the individual matter--energy components are negligible on large scales -- which is well--motivated for redshifts $z<10^9$, cf. Section \ref{sec:cosmo_outline_lcdm} --, the stress--energy tensor of each matter--energy component is conserved on its own and \eqref{eq_ConservationPerfectFluid} implies for the energy densities of radiation, matter and a cosmological constant $\Lambda$ respectively
\beq
\label{eq_ScalingMatter}
\frac{\varrho_\rad}{\varrho_0}=\frac{\Omega_\rad}{a^4}\,,\qquad\frac{\varrho_\mat}{\varrho_0}=\frac{\Omega_\mat}{a^3}\,,\qquad\frac{\varrho_\Lambda}{\varrho_0}=\Omega_\Lambda\,.\qquad
\eeq
Here, by convention, $\varrho_0$ is the energy density of today and thus the constants $\Omega_\rad\simeq 10^{-4}$, $\Omega_\mat\simeq 0.3$, $\Omega_\Lambda\simeq 0.7$ quantify the relative contributions of radiation, matter and the cosmological constant to the present energy density. Clearly, matter, radiation, and the cosmological constant have very different scaling behaviours with respect to $a$. The first Friedmann equation implies that, if $\varrho>0$ for all times and $\dot{a}>0$ at one instant of time, then $a$ will be strictly increasing for all times. Consequently, if we consider the present matter--energy content described above and assume that Dark Energy is a cosmological constant, then our universe must have had two phases of evolution preceding the present era dominated by Dark Energy: a phase where radiation has determined the behaviour of $a$ followed by a matter--dominated era. This motivates examining the solutions of the Friedmann equations separately for each matter--energy component and one finds
\beq
\label{eq_SolutionFriedmann}
a_\rad(t)\propto (t-t_0)^{\frac12}\,,\qquad a_\mat(t)\propto (t-t_0)^{\frac23}\,,\qquad  a_\Lambda(t)\propto e^{\sqrt{\frac{\Lambda}{3}}t}\,.
\eeq

The outcome of the preceding discussion is that, under the mentioned assumptions, our universe must have inevitably faced a {\em Big Bang}\index{Big Bang} at some point of time in the past, {\em  i.e.} there has been a $t_0>-\infty$ with $a(t_0)=0$ and we shall set $t_0=0$ in the following. Note that the occurrence of a Big Bang follows already from the second Friedmann equation and the assumptions $\dot{a}>0$, $\varrho>0$, since then $\varrho+3p>0$ and therefore $\ddot{a}<0$ if we take into account the sum of all three matter--energy constituents in the radiation--dominated era at early times.

The Big Bang scenario is known to have (at least) one hitch, usually termed {\em horizon problem}\index{horizon problem}. We refer to e.g. \cite{chap3_Dodelson,chap3_Ellis,chap3_Mukhanov:2005sc} for a quantitative discussion of this issue and only consider its qualitative aspects here. To wit, the isotropy of the temperature of the CMB radiation entails that the so--called {\em last scattering surface}, i.e. the region from where the CMB photons we see today have been emitted, must have lied in the forward lightcone of some event responsible for the thermal equilibrium of such region. The size of the last scattering surface is the radius $r_\text{em}$ of our past lightcone at the time $t_\text{em}$ of CMB photon emission, namely, the speed of light times the conformal time difference $\eta(t_\text{now})-\eta(t_\text{em})$. The isotropy of the CMB therefore entails that the following inequality must hold
\beq\label{eq_HorizonProblem}\eta(t_\text{em})-\eta(0)=\int\limits^{t_\text{em}}_{0}\frac{d\tilde{t}}{a(\tilde{t})}\ge r_\text{em}\eeq
and one can compute that this is not the case in the standard model of cosmology; this is precisely the horizon problem. A prominent possibility to solve the horizon problem is {\em Inflation}\index{Inflation}, see for instance \cite{chap3_Dodelson, chap3_Ellis, chap3_Mukhanov:2005sc, chap3_Straumann:2005mz}. In this scenario, one usually assumes that, in the very early universe, there has been an additional matter--energy component mimicking a large cosmological constant and thus leading to phase of exponential expansion. Inserting this assumption into \eqref{eq_HorizonProblem} leads to a large negative $\eta(0)$ and therefore allows for \eqref{eq_HorizonProblem} to be fulfilled. In the simplest models of Inflation the matter--energy component responsible for the exponential expansion is a classical scalar field with a suitable potential.

\section{The Cosmological Expansion in QFT on Curved Spacetimes}
\label{sec:lcdmqft}

Following the program outlined in Section \ref{sec:cosmo_outline_lcdm}, we shall now explain how the post--BBN cosmological evolution can be understood within QFT on curved spacetimes. In particular, we shall argue that there exist phenomenologically well--motivated Hadamard states for free quantum fields on FLRW spacetimes which approximately solve the semiclassical Einstein equation to a good degree in such a way that the energy density in these states and on the spacetime provided by the solution of the semiclassical equation is qualitatively and quantitatively of the form assumed in the $\Lambda$CDM--model up to small corrections whose possible interpretation we shall discuss. As argued in Section \ref{sec:cosmo_outline_lcdm}, we shall make the simplified assumption that matter and radiation are microscopically modelled by a pair of conformally coupled free neutral scalar fields which are massive and massless respectively.

\subsection{The Renormalisation Freedom of the Quantum Stress--Energy Tensor in the Context of Cosmology}
\label{sec_freedomcosmo}

We start our analysis by analysing the renormalisation freedom of the quantum stress--energy tensor in the cosmological context. To this avail, we recall the discussion in Section \ref{sec:stresstensor_see}. We consider a Hadamard state $\omega$ whose two--point function $\omega_2$ is (locally) of the form $\omega_2(x,y) = H(x,y) + \frac{w(x,y)}{8\pi^2}$, where $H$ is the purely geometric, singular and state--independent part and $w$ encodes the smooth and state--dependent part in \eqref{def_HadamardFormScalar}. The expectation value of $\wick{T_{\mu\nu}(x)}$ in the state $\omega$ reads
\beq\label{eq_deftmunu}
\omega\left(\wick{T_{\mu\nu}}\right)=\frac{[D_{\mu\nu}w]}{8\pi^2}+\alpha_1 m^4 g_{\mu\nu} + \alpha_2 m^2 G_{\mu\nu} + \alpha_3 I_{\mu\nu} + \alpha_4 J_{\mu\nu}\,,
\eeq
where $D_{\mu\nu}$ is the bidifferential operator
\begin{align}\label{eq_ImprovedKGDiffCOsmo}D_{\mu\nu}&=(1-2\xi)g^{\nu^\prime}_\nu\nabla_{\mu}\nabla_{\nu^\prime}-2\xi\nabla_{\mu}\nabla_\nu+\xi G_{\mu\nu}\notag\\&\quad+g_{\mu\nu}\left\{2\xi \square_x+\left(2\xi-\frac12\right)g^{\rho^\prime}_\rho\nabla^\rho\nabla_{\rho^\prime}-\frac12 m^2\right\}+\frac13 g_{\mu\nu}P_x\,,\notag
\end{align}
the bracket $[\cdot]$ denotes the coinciding point limit, $I_\mu$ and $J_\mu$ are conserved local curvature tensors \eqref{eq_ConservedLocal} and $\alpha_i$ are real dimensionless constants which are analytic functions of the coupling to the scalar curvature $\xi$. One should think of the $\alpha_i$ as encoding the {\em combined} renormalisation freedom of all quantum matter fields. In our model, only the massive scalar field `contributes to' $\alpha_1$ and $\alpha_2$, whereas $\alpha_3$ and $\alpha_4$ encode the combined freedom of the massive and massless conformally coupled scalar fields.

The parameters $\alpha_i$ are free parameters of the theory which are independent of the spacetime $(M,g)$ and can in principle be fixed by experiment, just like the mass $m$. The general physical interpretation of the occurrence of these a priori undetermined parameters is as follows. In usual particle physics experiments we always measure the {\em difference} of the expectation value of $\wick{T_{\mu\nu}}$ in two states, e.g. the vacuum and a many--particle state. However, gravity is sensitive to the {\em absolute} value of $\omega\left(\wick{T_{\mu\nu}}\right)$, thus the unambiguous specification of $\omega\left(\wick{T_{\mu\nu}}\right)$ would correspond to a specification of a `zero point' in the absolute energy scale, but this is impossible within quantum field theory in curved spacetime.

$\alpha_1$ and $\alpha_2$ can be interpreted as a renormalisation of the cosmological constant and a renormalisation of Newton's constant, respectively. In the following we will take the point of view that $\alpha_2$ is not a free parameter because Newton's constant has been measured already. In order to do this, we have to fix a value for the length scale $\lambda$ in the Hadamard parametrix $H$ \eqref{def_HadamardFormScalar} as the Definition of the smooth part $w$ and thus the expression $[D_{\mu\nu}w]$ depend on $\lambda$, cf. Theorem \ref{thm_TensorScalar} and in particular \eqref{eq_lchange}. We do this by confining $\lambda$ to be a scale in the range in which the strength of gravity has been measured. Because of the smallness of the Planck length, the actual value of $\lambda$ in this range does not matter as changing $\lambda$ in this interval gives a negligible contribution to $\omega\left(\wick{T_{\mu\nu}}\right)$. Moreover, in the case of conformal coupling $\xi=\frac16$, which we shall assume most of the time, $\alpha_2$ is independent of $\lambda$ as one can infer from \eqref{eq_lchange}. One could also take a more conservative point of view and consider $\alpha_2$ to be a free parameter, in this case comparison with cosmological data, e.g. from Big Bang Nucleosynthesis, would presumably constrain $\alpha_2$ to be very small once $\lambda$ is in the discussed range. Omitting the freedom parametrised by $\alpha_2$, we are left with three free parameters in the definition of $\omega\left(\wick{T_{\mu\nu}}\right)$. One of them corresponds to the cosmological constant which is already a free parameter in classical General Relativity, whereas the other two parameters do not appear there and thus will by themselves lead to an extension of the $\Lambda$CDM--model.

At this point, we would like to highlight the point of view on the so--called {\em cosmological constant problem}\index{cosmological constant problem} taken in this work, as well as in most works on algebraic QFT on curved spacetimes and e.g. the review \cite{chap3_Bianchi:2010uw}. It is often said that QFT {\em predicts} a value for the cosmological constant which is way too large in comparison to the one measured. This conclusion is reached by computing one or several contributions to the vacuum energy density in Minkowski spacetime $\Lambda_\text{vac}$ and finding them all to be too large, such that, at best, a fine--tuned subtraction in terms of a negative bare cosmological constant $\Lambda_\text{bare}$ is necessary in order to obtain the small value $\Lambda_\text{vac}+\Lambda_\text{bare}$ we observe. In this work, we assume as already mentioned the point of view that it is not possible to provide an {\em  absolute} definition of energy density within QFT on curved spacetimes, and thus neither $\Lambda_\text{vac}$ nor $\Lambda_\text{bare}$ have any physical meaning by themselves; only $\Lambda_\text{vac}+\Lambda_\text{bare}$ is physical and measurable and any cancellation which happens in this sum is purely mathematical. The fact that the magnitude of $\Lambda_\text{vac}$ depends on the way it is computed, e.g. the loop or perturbation order, cf. e.g. \cite{chap3_Sola:2013gha}, is considered to be unnatural following the usual intuition from QFT on flat spacetime. However, it seems more convincing to us to accept that $\Lambda_\text{vac}$ and $\Lambda_\text{bare}$ have no relevance on their own, which does not lead to any contradiction between theory and observations, rather than the opposite.

In the recent work \cite{chap3_Holland:2013xya} it is argued that a partial and unambiguous relevance can be attributed to $\Lambda_\text{vac}$ by demanding $\Lambda_\text{bare}$ to be analytic in all coupling constants, as required by the axioms for local and covariant Wick polynomials reviewed in Section \ref{sec:wick}; taking this point of view, one could give the contribution to $\Lambda_\text{vac}$ which is non--analytic in these constants an unambiguous meaning. Indeed the authors of \cite{chap3_Holland:2013xya} compute a non--perturbative and hence non--analytic contribution to $\Lambda_\text{vac}$ which turns out to be small. In the view of this, one could reformulate our statement in the above paragraph and say that contributions to $\Lambda_\text{vac}$ and $\Lambda_\text{bare}$ which are analytic in masses and coupling constants have no physical relevance on their own.

\subsection{States of Interest in Cosmological Spacetimes}
\label{sec:cosmostates}

We shall now discuss Hadamard states of interest in flat FLRW spacetimes. As argued in Section \ref{sec:cosmo_outline_lcdm}, in the context of describing the cosmological expansion within QFT on curved spacetimes we are interested in Hadamard states which may be phenomenologically interpreted as generalised thermal excitations of generalised vacuum states.

The spatial translation and rotation invariance of flat FLRW spacetimes allows to give a Fourier decomposition of solutions of the Klein--Gordon equation in terms of {\em modes}. To this end, we recall that the metric $g$ of a flat FLRW spacetime is conformally related to the Minkowski metric $g_0$ via $g=a^2 g_0$. This implies that the Klein--Gordon operator on FLRW spacetimes can be written as
$$P=-\Box+\xi R+m^2=\frac{1}{a^3}\left(\pa^2_\eta-\vec{\nabla}^2+a^2\left(\xi-\frac16\right)R+a^2m^2\right)a\,,$$
where $\vec{\nabla}$ denotes the gradient with respect to the (comoving) spatial coordinates. We thus see that scalar field in a FLRW spacetime is conformally equivalent to a scalar field with time--varying mass in Minkowski spacetimes. This leads us to define a mode solution $\phi_{\vec k}$ of $P\phi=0$ as
\beq
\label{eq_DefinitionModeScalar}
\phi_{\vec k}(\eta,\vec{x})\doteq \frac{\chi_k(\eta) e^{i\vec{k}\vec{x}}}{(2\pi)^{\frac32}a(\eta)}\,,
\eeq
where $k\doteq|\vec{k}|$ and the temporal mode $\chi_k(\eta)$ is a solution of the ordinary differential equation
\beq\label{eq_TimeODEScalar}\left(\pa^2_\eta+k^2+a(\eta)^2\left(\xi-\frac16\right)R(\eta)+a(\eta)^2m^2\right)\chi_k(\eta)=0\,,\eeq
which depends on $\vec k$ only via $k$. For each $k$, the solution space of this equation is two--dimensional and can be parametrised without loss of generality by any solution $\chi_k(\eta)$ and its complex conjugate $\overline{\chi_k(\eta)}$ which satisfies the normalisation condition 
\beq\label{eq_NormalisationScalarMode}\overline{\chi_k(\eta)}\pa_\eta \chi_k(\eta)-\chi_k(\eta)\pa_\eta\overline{\chi_k(\eta)}\equiv i\,.\eeq
Note that this condition is well--defined because the {\em Wronskian} on the right hand side of \eqref{eq_NormalisationScalarMode} is constant in $\eta$ on account of \eqref{eq_TimeODEScalar} and furthermore purely imaginary. Given a family of temporal modes $\chi_k$ which satisfy \eqref{eq_NormalisationScalarMode}, one can define a new family by a {\em Bogoliubov transformation}
\beq\label{eq_bolog}
\widetilde {\chi_k}= \lambda(k)\chi_k + \mu(k) \overline{\chi_k}\,,\qquad |\lambda(k)|^2-|\mu(k)|^2=1
\eeq
where the {\em Bogoliubov coefficients}\index{Bogoliubov coefficients} $\lambda(k)$ and $\mu(k)$ have to satisfy the normalisation condition above in order for \eqref{eq_NormalisationScalarMode} to hold. 

The normalisation condition \eqref{eq_NormalisationScalarMode} is chosen in such a way that a {\em pure} and {\em Gaussian} isotropic and homogeneous state for the Klein--Gordon field on a flat FLRW spacetime (satisfying mild regularity conditions) is determined by a two-point correlation function of the form \cite{chap3_Lueders:1990np}
$$\omega_2(x,y)=\lim_{\epsilon\downarrow 0}\frac{1}{8\pi^3 a(\eta_x)a(\eta_y)}\int\limits_{\mathbb{R}^3} d{\vec k}\, \chi_k(\eta_x)\overline{\chi_k(\eta_y)}e^{i\vec{k}(\vec{x}-\vec{y})}e^{-\epsilon k}\,.$$
In particular, we recall that, by Theorem \ref{thm_covariantequaltime}, the causal propagator $E(x,y)=i^{-1}(\omega_2(x,y)-\omega_2(y,x))$ satisfies 
$$
\nabla^x_N E(x,y)|_{\Sigma\times \Sigma}=\delta_\Sigma(\vec{x},\vec{y})
$$
for any Cauchy surface $\Sigma$ of the spacetime $(M,g)$ with future pointing normal vector field $N$, where $\delta_\Sigma(\vec{x},\vec{y})$ is the $\delta$--distribution with respect to the measure $d\Sigma$ on $\Sigma$ induced by the covariant volume measure $\vol$. Observing that $\{\eta\}\times \bR^3$ is a Cauchy surface with $N=a(\eta)^{-1}\partial_\eta$ and $d\Sigma = a(\eta)^3 d\vec{x}$, \eqref{eq_NormalisationScalarMode} ensures that the causal propagator satisfies
$$
\frac{1}{a(\eta_x)}\partial_{\eta_x}E(x,y)|_{\eta_x=\eta_y} = \frac{1}{a(\eta_x)^3}\delta(\vec{x}-\vec{y})\,.
$$

Choosing a pure, Gaussian, homogeneous and isotropic state $\omega$ of the quantized free Klein--Gordon field on a spatially flat FLRW spacetime amounts to choosing a solution of \eqref{eq_TimeODEScalar} and \eqref{eq_NormalisationScalarMode} for each $k$. In order for $\omega$ to be a Hadamard state, the temporal modes $\chi_k$ have to satisfy certain conditions in the limit of large $k$ which are difficult to formulate precisely. Heuristically, a necessary but not sufficient condition is that the dominant part of $\chi_k$ for large $k$, when the mass and curvature terms in \eqref{eq_TimeODEScalar} are dominated by $k^2$, is $\frac{1}{\sqrt{2k}} e^{-ik\eta}$, i.e. a {\em positive frequency} solution. A Bogoliubov transformation \eqref{eq_bolog} of a Hadamard state defined by a family of modes $\chi_k$ is Hadamard if and only if $\mu(k)$ is rapidly decreasing in $k$ \cite{chap3_Pinamonti:2010is,chap3_Zschoche:2013ola}.

A particular class of states often discussed in the literature are the {\em  adiabatic states} introduced in \cite{chap3_Parker:1969au}. They are specified by modes of the form
\beq\label{eq_adiabatic}
\chi_k(\eta)=\frac{1}{\sqrt{2\Omega(k, \eta)}}\exp\left(-i \int^\eta_{\eta_0}\Omega(k,\tilde\eta) d\tilde\eta\right),
\eeq
where $\Omega(k,\eta)$ solves a non--linear differential equation in $\eta$ obtained by inserting this ansatz into \eqref{eq_TimeODEScalar} and finding
$$\Omega(k,\eta)^2=f(\Omega(k,\eta)^{\prime\prime}, \Omega(k,\eta)^\prime, \Omega(k,\eta), a(\eta))$$
for a suitable function $f$. While this ansatz in principle holds for any state, the adiabatic states are specified by solving the differential equation for $\Omega(k,\eta)$ iteratively as
$$\Omega_{n+1}(k,\eta)^2\doteq f(\Omega_n(k,\eta)^{\prime\prime}, \Omega_n(k,\eta)^\prime, \Omega_n(k,\eta), a(\eta))$$
starting from $\Omega_0(k,\eta)=\sqrt{k^2+m^2a^2 + \left(\xi-\frac16\right)R a^2}$. Truncating this iteration after $n$ steps defines the adiabatic states of order $n$. Note that, while the resulting modes satisfy the normalisation condition \eqref{eq_NormalisationScalarMode} exactly, they satisfy \eqref{eq_TimeODEScalar} only up to terms which vanish in the limit of constant $a$ or of infinite $k$ and/or $m$. Thus they constitute only approximate states. This can be cured by using the adiabatic modes of order $n$ only for the specification of the initial conditions for exact solutions of \eqref{eq_TimeODEScalar}, see \cite{chap3_Lueders:1990np}. Regarding the regularity properties of such defined `proper' adiabatic states, it has been shown in \cite{chap3_Junker} (for spacetimes with compact spatial sections) that they are in general not as UV--regular as Hadamard states, but that they approach the UV--regularity of Hadamard states in a certain microlocal sense in the limit of large $n$. Consequently, while Hadamard states have sufficient regularity in order to compute products and expectation values of Wick polynomials with arbitrarily many derivatives, one has to consider adiabatic states of higher and higher order when dealing with Wick polynomials with an increasing number of derivatives, which is conceptually unsatisfactory. In the following we shall often use the `improper' adiabatic modes of order $0$, $\chi_{0,k}(\eta)\doteq\exp(-i\int^\eta_{\eta_0}\Omega_0(k,\tilde\eta) d\tilde\eta))/\sqrt{2 \Omega_0(k,\eta)}$. Adiabatic states have also been constructed for Dirac fields, see \cite{chap3_Hollands:1999fc,chap3_Landete:2013lpa}, and general curved spacetimes \cite{chap3_Hollands:1999fc,chap3_Junker}. 

A further class of states of interest in cosmology, and in fact our candidates for generalised vacuum states, are the {\em  states of low energy}\index{states of low energy} (SLE) introduced in \cite{chap3_Olbermann:2007gn}, motivated by results of \cite{chap3_Fewster:1999gj}. These states are defined by minimising the energy density per mode $\varrho_k$
$$\varrho_k\doteq\frac{1}{16 a^4 \pi^3}\left(|\chi_k^\prime|^2+\left(6\xi -1\right)aH\left(|\chi_k|^2\right)^\prime+\left(k^2+m^2a^2-\left(6\xi -1\right)H^2a^2\right)|\chi_k|^2\right)$$
integrated in (cosmological) time with a sampling function $f$ and thus loosely speaking minimise the energy in the time interval where the sampling function is supported. The minimisation is performed by choosing arbitrary basis modes $\chi_k$ and then determining the Bogoliubov coefficients $\lambda(k)$, $\mu(k)$ with respect to these modes, such that (for $\xi\in[0,\frac16]$) the resulting modes of the state of low energy are 
$$\chi_{f,k}=\lambda(k) \chi_k+\mu(k) \overline{\chi_k}$$
with
\begin{align*}
\lambda(k)\doteq e^{i(\pi -\arg c_2(k))}\sqrt{\frac{c_1(k)}{\sqrt{c_1(k)^2-|c_2(k)|^2}}+\frac12}\qquad\mu(k) \doteq \sqrt{|\lambda(k)|^2-1}
\end{align*}
\begin{align*}
c_1(k)&\doteq \frac{1}{2}\int dt f(t) \frac{1}{a^4}\left(|\chi_k^\prime|^2+\left(6\xi -1\right)aH\left(|\chi_k|^2\right)^\prime\right.\\
&\left.\qquad\qquad  +\left(k^2+m^2a^2-\left(6\xi -1\right)H^2a^2\right)|\chi_k|^2\right)
\end{align*}
\begin{align*}
c_2(k)&\doteq \frac{1}{2}\int dt f(t) \frac{1}{a^4}\left({\chi_k^\prime}^2+\left(6\xi -1\right)aH\left(\chi_k^2\right)^\prime\right.\\
&\left.\qquad\qquad  +\left(k^2+m^2a^2-\left(6\xi -1\right)H^2a^2\right)\chi_k^2\right)\,.
\end{align*}
\cite{chap3_Olbermann:2007gn} only discusses the case of minimal coupling, i.e. $\xi=0$ and proves that the corresponding SLE satisfy the Hadamard condition for sampling functions $f$ which are smooth and of compact support in time. However, we shall use these states for the case of conformal coupling $\xi=\frac16$, and, although we do not prove that they satisfy the Hadamard condition, we shall find them to be at least regular enough for computing the energy density. Moreover, it is not difficult to see that the SLE construction yields the {\em  conformal vacuum}\index{conformal vacuum}
$$\chi_{f,k}(\eta)=\frac{1}{\sqrt{2k}}e^{-ik\eta}\,,$$
and thus a Hadamard state \cite{chap3_Pinamonti:2008cx}, for all sampling functions $f$ in the massless case. This demonstrates both that the SLE construction for $\xi=\frac16$ yields Hadamard states at least in special cases and that states of low energy deserve to be considered as generalised vacuum states on curved spacetimes. In \cite{chap3_Daniel} it is conjectured that states of low energy in fact only exist for $\xi\in[0,\frac16]$ and that they are Hadamard in all these cases. The SLE construction has recently been generalised to spacetimes with less symmetry in \cite{chap3_Them:2013uka}.

A conceptual advantage of states of low energy is the fact that they can be consistently defined an all FLRW spacetimes at once just by specifying the sampling function $f$ once and for all (with respect to e.g. cosmological time and a fixed origin of the time axis). Thus, they solve the conceptual problem mentioned in Section \ref{sec:see}, namely the necessity to specify a state in a way which does not depend on the spacetime in order for the semiclassical Einstein equation to be a priori well--defined. Moreover, in  \cite{chap3_Degner} it has been proven that, on spacetimes which are asymptotically de Sitter towards $\eta\to-\infty$, every state of low energy converges to the Bunch--Davies state (the unique maximally symmetric Hadamard state) upon sending the support of $f(t(\eta))$ in $\eta$ to negative infinity. This is a rigorous variant of the statement that every state on de Sitter spacetimes converges to the Bunch--Davies state for positive asymptotic times.

We now proceed to construct the anticipated generalised thermal states on the basis of states of low energy. To this avail, we recall a result of \cite{chap3_Dappiaggi:2010gt}: given a pure, isotropic and homogeneous state, i.e. a set of modes $\chi_k$, one can construct generalised thermal states with a two--point correlation function of the form
\beq\label{eq_thermalfourier}
\omega_2(x,y)=\frac{1}{8\pi^3 a(\eta_x)a(\eta_y)}\int\limits_{\mathbb{R}^3} d{\vec k}\,e^{i\vec{k}(\vec{x}-\vec{y})}\left(\frac{\chi_k(\eta_x)\overline{\chi_k(\eta_y)}}{1-e^{-\beta k_0}}+\frac{\overline{\chi_k(\eta_x)}\chi_k(\eta_y)}{e^{\beta k_0}-1}\right)\,,
\eeq
with
$$k_0\doteq\sqrt{k^2+m^2a^2_F}\,.$$
It has been shown in \cite{chap3_Dappiaggi:2010gt} that for the case of conformal coupling, special FLRW spacetimes and particular generalised vacuum modes $\chi_k$ on these spacetimes, these states satisfy certain generalised thermodynamic laws and the Hadamard condition, and one can show that they satisfy the Hadamard condition on general FLRW spacetimes if the pure state specified by $\chi_k$ is already a Hadamard state by using results of \cite{chap3_Pinamonti:2010is,chap3_Zschoche:2013ola}. Essentially, this follows from the fact that the difference of the momentum space integrand in \eqref{eq_thermalfourier} and the corresponding integrand of the pure state induced by $\chi_k$ is rapidly decreasing for large $k$.

We shall assume in the following that the quantum fields in our model are in a generalised thermal state of the form as above, with generalised vacuum modes $\chi_k$ specified by a state of low energy with suitable sampling function $f$. If $m>0$, the phenomenological interpretation of these states is that they are the quantum state of a massive field which has been in thermal equilibrium in the hot early universe and has departed from this equilibrium at the `freeze--out time' $a=a_F$ and `freeze--out temperature' $T_F=a_F/\beta$. In the massless conformally coupled case, these states are just conformal rescalings of the thermal equilibrium (KMS) state with temperature $1/\beta$ in Minkowski spacetime\index{conformal KMS state}. 

The generalised thermal states we use here have also been constructed and analysed for Dirac fields, see \cite{chap3_Dappiaggi:2010gt}. Moreover, we would like to mention that several definitions of generalised thermal states on curved spacetimes have been proposed so far, including {\em  almost equilibrium states} \cite{chap3_Kusku:2008zz} and {\em  local thermal equilibrium states} \cite{chap3_Buchholz:2006iv,chap3_Schlemmer,chap3_VerchRegensburg}. A comparison of these different proposals in the context of cosmological applications would certainly be interesting, but is beyond the scope of this work.

\subsection{Setup and (Computational) Strategy for Approximately Solving the Semiclassical Einstein Equation}
\label{sec:comp_energydensity}

We now approach the first main result of Section \ref{sec:lcdmqft}, a demonstration that the energy density in the $\Lambda$CDM--model can be reproduced from first principles within quantum field theory in curved spacetime. To this avail, we consider the following setup: we model radiation by a conformally coupled massless scalar quantum field and matter by a conformally coupled massive scalar quantum field. We choose the conformal coupling also for the massive scalar field because this considerably simplifies analytical computations and we also found numerical computations to be more stable with this value of non--minimal coupling to the curvature.  Moreover, both quantum fields are assumed to be in generalised thermal equilibrium states as introduced in the previous section, where the state and field parameters $\beta$ (possibly different values for the two quantum fields), $m$ and $a_F$, as well as the sampling functions $f$ determining the generalised vacuum states of the two fields, are considered to be undetermined for the time being. Let us stress once more that there is no principal obstruction for formulating this model with more realistic quantum fields of higher spin, we just consider scalar quantum fields for simplicity and ease of presentation.

An exact computation of the energy density of the two quantum fields in the generalised thermal states would require to solve the coupled system -- the so-called {\em  backreaction problem} -- consisting of the quantum fields propagating on a FLRW spacetime, which in turn is a solution of the {\em  semiclassical Friedmann equation}\index{semiclassical Friedmann equation}
\begin{equation}
H^2 = \frac{8 \pi G}{3} \left(\varrho^0 + \varrho^m\right)\,,
\end{equation}
where $\varrho^m\doteq\omega^m( \wick{T^m_{00}})$, $\varrho^0\doteq\omega^0( \wick{T^0_{00}})$ are the energy densities of the two quantum fields in the respective generalised thermal states and the $00$-component of the stress--energy tensor is here taken with respect to cosmological time $t$. An exact solution of the backreaction problem is quite involved, as it requires solving simultaneously the mode equation \eqref{eq_TimeODEScalar} for all $k$ and the semiclassical Friedmann equation. Notwithstanding, there have been quantitative numerical treatments of the backreaction problem, see e.g. \cite{chap3_Anderson1, chap3_Anderson:1984jf, chap3_Anderson:1985cw, chap3_Anderson:1985ds, chap3_Baacke:2010bm}, as well as numerous qualitative treatments including \cite{chap3_Eltzner:2010nx}, where the backreaction problem in FLRW spacetimes is set up in full generality from the point of view of the algebraic approach to QFT on curved spacetimes, \cite{chap3_DFP}, where the same point of view is considered and the coupled system is solved exactly for conformally coupled massless scalar quantum fields and approximately for massive ones, and \cite{chap3_Pinamonti:2008cx,chap3_Pinamonti:2013wya}, where exact solutions of the backreaction problem are shown to exist.

However, in this work we follow a simplified strategy in order to avoid solving the full backreaction problem, which is justified in view of our aim. We assume that the two quantum fields in our model are propagating on a FLRW spacetime which is an exact solution of the Friedmann equation in the $\Lambda$CDM--model, i.e.
\begin{equation}\label{eq_friedmannlcdm}
\frac{H^2}{H^2_0}=\frac{\varrho_{\text{$\Lambda$CDM}}}{\varrho_0}=\Omega_\Lambda + \frac{\Omega_m}{a^3} + \frac{\Omega_r}{a^4}\,,
\end{equation}
where $H_0\simeq 10^{-33}$eV denotes the Hubble rate of today, the so--called {\em  Hubble constant}\index{Hubble constant}, $\varrho_0\simeq 10^{-11}$eV$^4$ is the energy density of today and $\Omega_\Lambda$, $\Omega_\rad$ and $\Omega_\mat$ denote respectively the present--day fraction of the total energy density contributed by the cosmological constant, matter and radiation, cf. Section \ref{sec:cosmobrief}. For definiteness we consider the sample values $\Omega_\mat=0.30$, $\Omega_\rad = 10^{-4}$, $\Omega_\Lambda=1-\Omega_\mat-\Omega_\rad$, rather than currently measured values from e.g. the Planck collaboration \cite{chap3_Ade:2013zuv}, because the exact values are not essential for our results. Given this background spacetime, we aim to prove that the field and state parameters of our model, as well as the SLE sampling functions, can be adjusted in such a way that the energy density of the quantum fields in our model matches the one in the $\Lambda$CDM--model up to negligible corrections for all redshifts $z\in[0,10^9]$, i.e.
$$\frac{\varrho^0 + \varrho^m}{\varrho_0}\simeq\Omega_\Lambda + \frac{\Omega_m}{a^3} + \frac{\Omega_r}{a^4}=\frac{\varrho_\text{$\Lambda$CDM}}{\varrho_0}\,.$$
Once we succeed to obtain this result, we have clearly solved the full coupled system in an approximative sense to a good accuracy in particular. We recall that the restriction to post--Big Bang Nucleosynthesis (BBN) redshifts $z\in[0,10^9]$ is motivated by the fact that interactions between matter fields are assumed to have negligible large--scale effects in this period of the cosmological history, cf. Section \ref{sec:cosmo_outline_lcdm}.

In order to compute the quantum energy density $\varrho^0 + \varrho^m$, we start from \eqref{eq_deftmunu} which parametrises the freedom in defining the energy density as an observable and gives a possible `model definition'. The renormalisation freedom for the energy density is readily computed via $g_{00}=-1$, $G_{00}=3H^2$ and
\begin{equation}\label{eq_J00}J_{00}=\frac13 I_{00}=6 \dot{H}^2 - 12 \ddot{H}H- 36 \dot{H} H^2\,.
\end{equation}
In order to compute the energy density for each quantum field following from \eqref{eq_deftmunu}, one has to first subtract the geometric Hadamard parametrix from the two--point function of the given state and then to apply a suitable bidifferential operator followed by taking the coinciding point limit. As the states we consider here are given as integrals over spatial momenta, it seems advisable to try to re--write the geometric Hadamard parametrix also in this form, in order to perform a mode--by--mode subtraction and a momentum space integral afterwards. This is indeed possible, as elaborated in \cite{chap3_DFP, chap3_Pinamonti:2010is, chap3_Eltzner:2010nx, chap3_Schlemmer, chap3_Degner, chap3_Zschoche:2013ola}. The details are quite involved, thus we omit them and present directly the result. To this avail we follow \cite{chap3_Degner}, where results of \cite{chap3_Schlemmer} are used. In \cite{chap3_Degner} only the minimally coupled case $\xi=0$ is discussed, but it is not difficult to generalise the results there to arbitrary $\xi$.

Doing this, we find the following result for the total energy density of the massless and massive conformally coupled scalar fields in the generalised thermal states.
\begin{equation}\label{eq_totaled}
\frac{\varrho^0 + \varrho^m}{\varrho_0}=\frac{\varrho^m_{\text{gvac}}+\varrho^0_{\text{gvac}}+\varrho^m_{\text{gth}}+\varrho^0_{\text{gth}}}{\varrho_0}+\gamma \frac{H^4}{H^4_0}+\Omega_\Lambda+\delta\frac{H^2}{H^2_0}+\epsilon\frac{J_{00}}{H^4_0}
\end{equation}
$$\gamma\doteq\frac{8\pi G H^2_0}{360 \pi^2}\qquad \Omega_\Lambda=\frac{8\pi G \alpha_1}{3 H^2_0}\qquad\delta\doteq\frac{8\pi G \alpha_2}{3 H^2_0}\qquad \epsilon\doteq\frac{8\pi GH^2_0}{3 }(3\alpha_3 + \alpha_4)\,.$$
Here $\Omega_\Lambda$, $\delta$ and $\epsilon$ parametrise the freedom in the definition of the energy density as per \eqref{eq_deftmunu}. The number of free parameters in this equation has been reduced to three, because $I_{\mu\nu}$ and $J_{\mu\nu}$ are proportional in Robertson-Walker spacetimes, cf. \eqref{eq_J00}. As already discussed in Section \ref{sec_freedomcosmo}, we omit the freedom parametrised by $\delta$ in the following, as it renormalises the Newton constant and we consider this to be already given as an external input. For now we will also neglect the contribution parametrised by $\epsilon$, as it turns out to be negligible for $0\le\epsilon\ll1$; we will analyse the influence of this new term, which does not appear in the $\Lambda$CDM--model, separately in the next section. Thus, for the remainder of this section,  $\Omega_\Lambda$ parametrises the residual freedom in the definition of the quantum energy density. 

The term proportional to $\gamma$, which is also not present in  the $\Lambda$CDM--model, is the contribution of the trace anomaly, cf. Section \ref{sec:stresstensor}. This term is fixed by the field content, i.e. by the number and spins of the fields in the model and always proportional to $H^4$, barring contributions proportional to $J_{00}$ which we prefer to subsume in the parameter $\epsilon$. We have given here the value of $\gamma$ for two scalar fields, see Table 1
on page 179 of \cite{chap3_Birrell:1982ix} for the values in case of higher spin. As $\gamma\simeq 10^{-122}$ and $H< H_0 z^2$ in the $\Lambda$CDM--model for large redshifts, this term can be safely neglected for $z<10^9$; we will comment on the relevance of this term for $z\gg 10^9$ in Section \ref{sec:dark}. 

Finally, the remaining terms in \eqref{eq_totaled} denote the genuinely quantum state--dependent contributions to the energy densities of the two quantum fields. We have split these contributions into parts which are already present for infinite inverse temperature parameter $\beta$ in the generalised thermal states, and thus could be considered as contributions due to the generalised vacuum states ($\varrho^m_\text{gvac}$, $\varrho^0_\text{gvac}$), and into the remaining terms, which could be interpreted as purely thermal contributions ($\varrho^m_\text{gth}$, $\varrho^0_\text{gth}$). Note that $\varrho^m_\text{gvac}$, $\varrho^0_\text{gvac}$ are not uniquely defined in this way, but only up to the general renormalisation freedom of the quantum energy density, i.e. one could ``shuffle parts of'' $\Omega_\Lambda$, $\delta$ and $\epsilon$ into e.g. $\varrho^m_\text{gth}$ and vice versa, without changing any physical interpretation of the total energy density. We have have fixed the renormalisation freedom of $\varrho^m_\text{gvac}$ (and thus $\varrho^0_\text{gvac}$ for $m=0$) in such a way that it has a particularly simple form in the conformally coupled case, cf. \eqref{eq_rhofinalconformal}.  With this in mind, the state--dependent contributions read as follows, where the massless case is simply obtained by inserting $m=0$, and we give here the result for arbitrary coupling $\xi$ for completeness.

\begin{align}
\varrho^m_{\text{gvac}} &= \frac{1}{2\pi^2}\int\limits_0^\infty dk k^2\left\{\frac{1}{2a^4}\left(|\chi_k^\prime|^2+\left(6\xi -1\right)aH\left(|\chi_k|^2\right)^\prime\phantom{\left(|\chi_k|^2\right)^\prime}\right.\phantom{\frac{\frac{ 72 H^2 m^2}{ 6\xi -1} }{16 k^3}}\right.\label{eq_rhofinal}\\
&\qquad \left.+\left(k^2+m^2a^2-\left(6\xi -1\right)H^2a^2\right)|\chi_k|^2\right)\notag\\
&\qquad -\frac{k}{2 a^4}-\frac{m^2-H^2(6\xi -1)}{4 a^2 k}\notag\\
&\qquad\left.-\Theta(k-ma)\frac{-m^4 + \left(\xi-\frac16\right)^2\left(-\frac{ 72 H^2 m^2}{ 6\xi -1} - 216 H^2 \dot{H} + 
 36 \dot{H}^2  - 72 H \ddot{H}\right) }{16 k^3}\right\} \notag\\
&\quad-\left(\xi-\frac16\right)\frac{72 H^4  + 72 H^2 \dot{H}  + 
 18 \dot{H}^2   - 216 H^2 \dot{H} (\xi -\frac16)- 
 108 \dot{H}^2 (\xi -\frac16)}{96 \pi^2}\notag\\
&\quad-\frac{1-4 \log 2}{128 \pi^2}m^4
 -\frac{H^2 m^2}{96 \pi^2}\notag\\
\label{eq_rhofinalth}\varrho^m_\text{gth} & = \frac{1}{2\pi^2}\int\limits_0^\infty dk k^2\frac{1}{a^4}\frac{1}{e^{\beta k_0}-1}\left(|\chi_k^\prime|^2+\left(6\xi -1\right)aH\left(|\chi_k|^2\right)^\prime\right.\\
&\left.\phantom{\left(|\chi_k|^2\right)^\prime}+\left(k^2+m^2a^2-\left(6\xi -1\right)H^2a^2\right)|\chi_k|^2\right)\notag
\end{align}
In the conformally coupled case $\xi=\frac16$ $\varrho^m_{\text{gvac}}$ simplifies considerably to
\begin{align}\label{eq_rhofinalconformal}
\varrho^m_\text{gvac}& = \frac{1}{2\pi^2}\int\limits_0^\infty dk k^2\left\{\frac{1}{2a^4}\left(|\chi_k^\prime|^2+\left(k^2+m^2a^2\right)|\chi_k|^2\right)\right.\\
&\qquad\qquad\left.-\left(\frac{k}{2 a^4}+\frac{m^2}{4 a^2 k}-\Theta(k-ma)\frac{m^4}{16 k^3}\right)\right\} -\frac{1-4 \log 2}{128 \pi^2}m^4
 -\frac{H^2 m^2}{96 \pi^2}\notag\\
 &=\frac{1}{2\pi^2}\int\limits_0^\infty dk k^2\frac{1}{2a^4}\left\{\left(|\chi_k^\prime|^2+\left(k^2+m^2a^2\right)|\chi_k|^2\right)\right.\notag\\
&\qquad\qquad -\left. \left(|\chi_{0,k}^\prime|^2+\left(k^2+m^2a^2\right)|\chi_{0,k}|^2\right)\right\}\,,\notag
\end{align}

\noindent where $\chi_{0,k}$ are the adiabatic modes of order 0, cf. the previous section. This implies that the so--called {\em  Hadamard point--splitting regularisation} of the energy density coincides with the so--called {\em  adiabatic regularisation of order zero}\index{adiabatic regularisation} up to the trace anomaly term and terms which can be subsumed under the regularisation freedom. 

\subsection{Computation of the Energy Density}

In the following we shall analyse the individual state--dependent terms in the energy density.

\subsubsection{Computation of $\varrho^m_\text{gvac}$}

Following our general strategy, we first aim to show that in states of low energy  defined by a sampling function of sufficiently large support in time on the FLRW spacetime specified by \eqref{eq_friedmannlcdm}, $\varrho^m_\text{gvac}$ is for all $z\in[0,10^9]$ negligible in comparison to the total energy density in the $\Lambda$CDM--model. Results in this direction have been reported in \cite{chap3_Degner} for the simplified situation of a de Sitter spacetime background  (corresponding to $\Omega_\mat=\Omega_\rad=0$). Here we generalise these results to $\Lambda$CDM--backgrounds. One can easily see that $\varrho^m_\text{gvac}=0$ in the case of $m=0$. For masses in the range of the Hubble constant $m\simeq H_0$ and states of low energy we have performed numerical computations and found $\varrho^m_\text{gvac}/\varrho_\text{$\Lambda$CDM} < 10^{-116}$, see Figures \ref{fig_rhonormsmallz}, \ref{fig_rhonormlargez}. To achieve this result, we have rewritten all expressions in terms of the redshift $z$ as a time variable and solved the equation \eqref{eq_NormalisationScalarMode} with initial conditions at $z=0$ given by the value and derivative of the adiabatic modes of order zero $\chi_{k,0}$ there. Note that a state of low energy does not depend on the choice of a mode basis, but the choice we made seemed to be numerically favoured. To fix the state of low energy, we chose a sampling function which was a symmetric bump function in $z$ supported in the interval $z\in (10^{-2}, 10^{-2}+10^{-4})$ for definiteness. In order to make the numerical computations feasible, we chose a logarithmic sampling of $k$ with $10^3$ sampling points, where the boundaries of the sampling region have been chosen such that the integrand of $\varrho^m_\text{gvac}$, cf. \eqref{eq_rhofinalconformal}, was vanishing in $k$--space to a large numerical accuracy outside of the sampling region for all $z\in[0,10^9]$. We have computed the mode coefficients $c_i(k)$ in the mode basis chosen at each sampling point by a numerical integration in $z$ and finally the energy density by means of a sum over the sampling points in $k$--space. Thus we have approximated the integral in \eqref{eq_rhofinalconformal} by a Riemann sum with logarithmic sampling. As our main aim here is to demonstrate that $\varrho^m_\text{gvac}/\varrho_\text{$\Lambda$CDM}\ll 1$ in general, we have not performed an extensive analysis of the dependence of $\varrho^m_\text{gvac}$ on the width of the sampling function, but we have observed that the maximum amplitude of $\varrho^m_\text{gvac}/\varrho_\text{$\Lambda$CDM}$ seems to be monotonically growing with shrinking width of the sampling function, in accordance with the computations of \cite{chap3_Degner} in de Sitter spacetime.

\begin{figure}[htb]
\includegraphics[width=1\columnwidth]{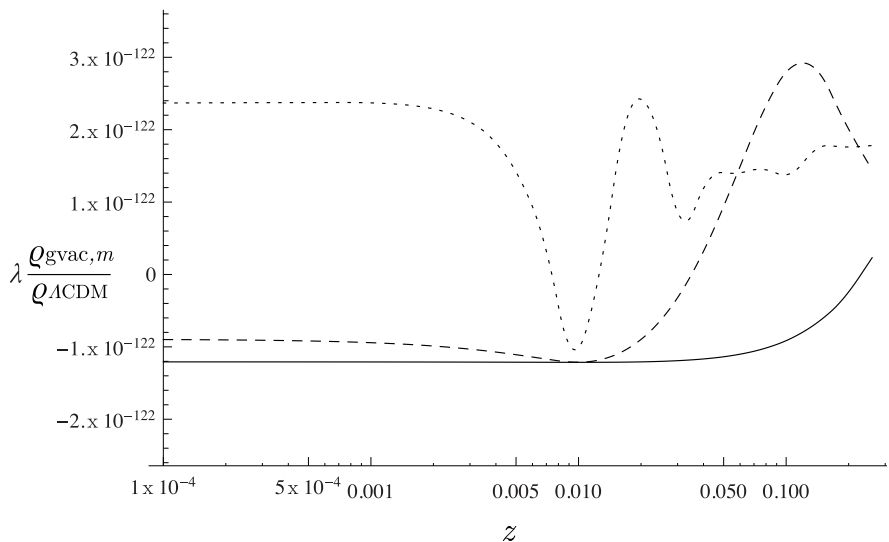}
\caption{\label{fig_rhonormsmallz}$\lambda \varrho^m_\text{gvac}/\varrho_\text{$\Lambda$CDM}$ for $z<1$ for various values of $m$ (rescaled for ease of presentation). The dotted line corresponds to $m=100H_0$ and $\lambda=10^{-2}$, the dashed line to $m=10H_0$ and $\lambda=1$ and the solid line to $m=H_0$ and $\lambda=10^{2}$. One sees nicely how the energy density is minimal in the support of the sampling function at around $z=10^{-2}$.}
\end{figure}

\begin{figure}[htb]
\includegraphics[width=1\columnwidth]{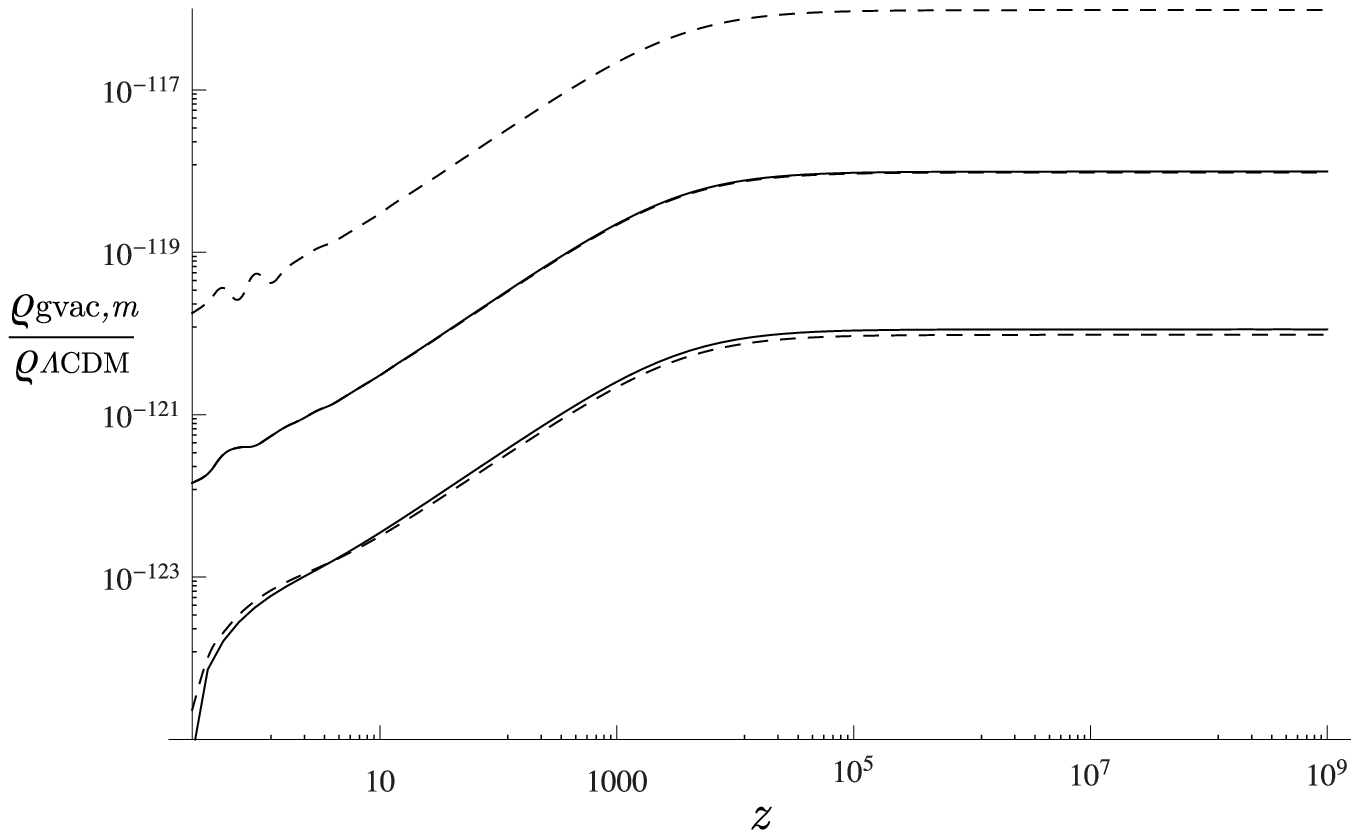}
\caption{\label{fig_rhonormlargez}$\varrho^m_\text{gvac}/\varrho_\text{$\Lambda$CDM}$ for $z>1$ for various values of $m$. The upper line corresponds to $m=100H_0$, the middle lines to $m=10H_0$ and the lower lines to $m=H_0$; solid lines (dashed lines) indicate results obtained with exact modes (zeroth order adiabatic modes). $\varrho^m_\text{gvac}/\varrho_\text{$\Lambda$CDM}$ becomes constant for large $z$ because there both energy densities scale like $a^{-4}$, c.f. \eqref{eq_scaling} and the related discussion.}
\end{figure}

Unfortunately, we have not been been able to compute $\varrho^m_\text{gvac}/\varrho_\text{$\Lambda$CDM}$ for $m>10^2 H_0$ in the way outlined above because for large masses the modes oscillate heavily, and thus it costs a lot of computer power to solve the mode equation for such a large $z$--interval we are interested in and to the numerical accuracy which is necessary to obtain reliable results for the coefficients of the state of low energy and $\varrho^m_\text{gvac}/\varrho_\text{$\Lambda$CDM}$, see also \cite[Section 8.4]{chap3_Daniel} for related considerations. However, realistic field masses in the GeV regime are rather of the order of $10^{42} H_0$. In the numerical computations outlined above we have observed that $\varrho^m_\text{gvac}/\varrho_\text{$\Lambda$CDM}$ seemed to grow quadratically with $m$, see Figure \ref{fig_rhonormlargez}, but looking at the results of \cite{chap3_Degner} in de Sitter spacetime one could maybe expect that $\varrho^m_\text{gvac}/\varrho_\text{$\Lambda$CDM}$ decreases for large masses. Moreover, even if a potential quadratic growth of $\varrho^m_\text{gvac}/\varrho_\text{$\Lambda$CDM}$ with $m$ would still imply $\varrho^m_\text{gvac}/\varrho_\text{$\Lambda$CDM}\ll 1$ for realistic masses and given $\varrho^m_\text{gvac}/\varrho_\text{$\Lambda$CDM} \sim 10^{-120}$ for $m= H_0$, it would be better to have a more firm understanding of the large mass regime.

In view of the numerical problems for large masses we had to resort to an approximation in order to be able to compute $\varrho^m_\text{gvac}/\varrho_\text{$\Lambda$CDM}$. In fact, we have taken the adiabatic modes of order zero as basis modes for computing the state of low energy. Of course these modes are not exact solutions of the mode equations, but the failure of these modes to satisfy the exact mode equation is decreasing with increasing mass and thus one can expect that the error in all quantities derived from these modes rather than exact modes is also decreasing with increasing mass. We have checked numerically that the energy density computed with adiabatic modes rather than exact modes matched the `exact' result quite well already for masses in the regime $m\simeq H_0$, see Figure \ref{fig_rhonormlargez}. For more details regarding error estimates for adiabatic modes we refer the reader to \cite{chap3_Olver}.

Inserting the adiabatic modes $\chi_{0,k}$ we obtain the following expressions for the coefficients $c_i(k)$ of the states of low energy.

$$c_1(k) = \int dz \,f(z) \left\{\frac{m^4 H}{16 \Omega_0(k)^5(1+z)} +\frac{\Omega_0(k)(1+z)^3}{2 H}\right\}$$
$$c_2(k) = \int dz \,f(z) \left\{\frac{m^4 H}{16 \Omega_0(k)^5(1+z)}-i\frac{m^2}{4\Omega_0(k)^2}\right\}\exp\left({-2i\int^z_{z_0}\frac{\Omega_0(k)}{H}dz^\prime}\right)$$

We now perform a further approximation. We take as a sampling function a Gaussian with mean $z_0$ and variance $\sigma\ll 1$
$$f(z)=\frac{1}{\sqrt{2\pi \sigma^2}}\exp\left(-\frac{(z-z_0)^2}{2\sigma^2}\right)$$
and take the zeroth term of the Taylor expansion of both the expressions in the curly brackets in the integrands of $c_i(k)$ and of the integrand appearing in the exponent of the exponential in $c_2(k)$. Without performing a detailed error analysis we note that this is justified for $\sigma\ll 1$ because the higher coefficients of the associated Taylor series differ from the lowest coefficient roughly by factors of either $\partial_zH/H|_{z=z_0}$ or $H(z_0)/m$, both of which are either smaller than or of order one under the assumption of large masses and a $\Lambda$CDM--background. We can now compute the $z$--integrals, which corresponds to considering the Fourier--transform of $f$ in the case of $c_2$. Using $H_0/m\ll 1$ (and thus $H(z_0)/m\ll1$), we can estimate the resulting coefficients as follows
$$c_1(k) > 1$$
$$|c_2(k)|< \exp\left(-\frac{k^2 \sigma^2}{H(z_0)^2}\right)\exp\left(-\frac{m^2\sigma^2}{H(z_0)^2(1+z_0)^2}\right)\,.$$
For $H(z_0)(1+z_0)/(m\sigma)\ll1$, $|c_2(k)|\ll 1$ and we can approximate the Bogoliubov coefficients $\lambda(k)$ and $\mu(k)$ as $\mu(k) \simeq \frac{|c_2(k)|}{2 c_1(k)}$, $\lambda(k)\simeq 1$ and thus estimate $\varrho^m_\text{gvac}$ as
\begin{align*}|\varrho^m_\text{gvac}|<&\frac{1}{4 a^4}\int\limits^\infty_0 dk k^2(\mu^2+\mu |\lambda|) \left(|\chi_{0,k}^\prime|^2+\Omega^2_0|\chi_{0,k}|^2\right)<\frac{1}{ a^4}\int\limits^\infty_0 dk k^2\mu |\lambda| \Omega_0\\<& \frac{1}{ a^4}\frac{H(z_0)^3 m}{\sigma^3}\exp\left(-\frac{m^2\sigma^2}{H(z_0)^2(1+z_0)^2}\right)
\end{align*}
such that, barring our approximations, we indeed get a result which shows that the energy density decreases -- exponentially -- for large masses. Note that for not too small $\sigma$ the bound we found is in general small compared to $\varrho_\text{$\Lambda$CDM}$ even if we forget about the exponential because $H_0 m$ is much smaller than the the square of the Planck mass, i.e. $1/G$. We also see that the bound grows with growing $z_0$, i.e. if we `prepare' the state of low energy further in the past, and that it diverges if the width of the sampling function goes to zero; this is in accord with the results of \cite{chap3_Degner} in de Sitter spacetime. Note that we could have chosen any rapidly decreasing or even compactly supported sampling function in order to obtain a bound which is rapidly decreasing in $m/H_0$, thus one could say that the result does not depend on the shape of the sampling function as long as its width is not too small. Finally, one could of course directly take the point of view that for large masses the adiabatic modes $\chi_{0,k}$ define `good states' themselves and conclude that in these states $\varrho^m_\text{gvac}\simeq 0$ on account of \eqref{eq_rhofinalconformal}.

\subsubsection{Computation of $\varrho^m_\text{gtherm}$}

We now proceed to analyse the thermal parts of the state--dependent contributions to the total energy density. Inserting $\xi=\frac16$ in \eqref{eq_rhofinalth}, we find
\begin{equation}\label{eq_thermaldensity}\varrho^m_\text{gth} = \frac{1}{2\pi^2}\frac{1}{a^4}\int\limits_0^\infty dk k^2\frac{1}{e^{\beta k_0}-1}\left(|\chi_k^\prime|^2+\left(k^2+m^2a^2\right)|\chi_k|^2\right)\end{equation}
with $k_0=\sqrt{k^2+a_F^2m^2}$. Before performing actual computations, we would like to mention a general result about the scaling behaviour of the energy density with respect to $a$ \cite{chap3_Pinamonti}. To wit, using the equation of motion \eqref{eq_NormalisationScalarMode} and the assumption that $H>0$, one can compute the derivative of $Q_k\doteq|\chi_k^\prime|^2+\left(k^2+m^2a^2\right)|\chi_k|^2$ with respect to $a$ and obtain the following inequalities
\begin{equation}\label{eq_scaling}\frac{k^2+a^2m^2}{k^2+m^2}\frac{Q_k(a=1)}{a^4}\leq\frac{Q_k(a)}{a^4}\leq \frac{Q_k(a=1)}{a^4}\,.\end{equation}
From these one can already deduce that $\varrho^m_\text{gth}$ has a scaling behaviour with respect to $a$ which lies between $a^{-2}$ and $a^{-4}$ and approaches $a^{-4}$ in the limit of vanishing $a$, in fact this still holds if we replace the Bose--Einstein factors in the generalised thermal states by arbitrary functions of $k$. Moreover \eqref{eq_scaling} also implies that $\varrho^m_\text{gvac}$ can not scale with a power of $a$ lower than $-4$ for small $a$ on $\Lambda$CDM backgrounds, cf. \eqref{eq_rhofinalconformal}.

Proceeding with actual computations we find that in the massless case $\varrho^m_{\text{gth}}$ can be computed exactly and analytically and the result is
\begin{equation}\label{eq_thermaldensitymassless}
\varrho^0_{\text{gth}}=\frac{\pi^2}{30}\frac{1}{\beta^4 a^4}\,.
\end{equation}
As in the massless case the state of low energy is the conformal vacuum and the associated generalised thermal state is the conformal temperature state with temperature parameter $\beta=1/T$, this result in fact holds for fields of all spin, i.e. the generalised thermal energy density in this case is always the one in Minkowski spacetime rescaled by $a^{-4}$. Thus a computation with e.g. photons or massless neutrinos yields the same result \eqref{eq_thermaldensitymassless} up to numerical factors due to the number of degrees of freedom and the difference between Bosons and Fermions.

In the massive case it is not possible to compute $\varrho^m_\text{gth}$ analytically and exactly, but we have to resort to approximations once more. We recall that the massive scalar field in our model should represent baryonic matter and Dark Matter in a simplified way. Thus we take typical values of $\beta$, $a_F$ and $m$ from Chapter 5.2 in \cite{chap3_Kolb:1990vq} computed by means of effective Boltzmann equations. A popular candidate for Dark Matter is a weakly interacting massive particle (WIMP), e.g. a heavy neutrino, for which  \cite{chap3_Kolb:1990vq} computes
\begin{equation}\label{eq_kolbvals}x_F \doteq \beta a_F m \simeq 15 + 3\log(m/\text{GeV})\,,\end{equation}
$$a_F \simeq 10^{-12}(m/\text{GeV})^{-1}\,.$$
We shall take these numbers as sample values although working with a scalar field, because for large masses $m\gg H_0$, the `thermal energy densities' $\varrho^m_\text{gth}$ in generalised thermal states for free fields of spin $0$ and $\frac12$ can be shown to approximately coincide up to constant numerical factors on the basis of the results of \cite{chap3_Dappiaggi:2010gt} and \cite[Section IV.5]{chap3_Hack:2010iw}.

Considering $m>1$GeV, we can compute $\varrho^m_\text{gth}$ approximatively as follows. We recall from the computation of $\varrho^m_\text{gvac}$ that for large masses $m\gg H_0$ one can consider the adiabatic modes of order zero $\chi_{0,k}$ as approximative basis modes for the computation of the state of low energy and that with respect to this basis one finds for the coefficients of the state of low energy $\lambda\simeq 1$, $\mu\simeq 0$, thus we can insert those modes in \eqref{eq_thermaldensitymassless} instead of the modes of the state of low energy. Using $m\gg H_0$ once more, we have $|\chi_{0,k}^\prime|^2+\left(k^2+m^2a^2\right)|\chi_{0,k}|^2\simeq \sqrt{k^2+m^2a^2}$ and using $x_F>15$ we can approximate the Bose--Einstein factor in \eqref{eq_thermaldensitymassless} as $1/(e^{\beta k_0}-1)\simeq e^{-\beta k_0}$. Finally, we can rewrite the integral in \eqref{eq_thermaldensitymassless} in terms of the variable $y=k/(a_F m)$ and compute, using $a/a_F\gg1$ for the redshift interval $z\in[0,10^9]$ we are interested in,
$$\varrho^m_\text{gth}\simeq \frac{1}{2\pi^2}\frac{a_F^3 m^4}{a^3}\int\limits^\infty_0 dy \;y^2 e^{-x_F\sqrt{y^2+1}}\,.$$
This already gives the desired result $\varrho^m_\text{gth}\propto a^{-3}$. The remaining integral can be computed numerically, however, for $x_F \gg 1$ only $y\ll 1$ contribute to the integral and one can approximate $\sqrt{y^2+1}\simeq 1+y^2/2$ and compute 
$$\varrho^m_\text{gth}\simeq \frac{1}{(2\pi)^{3/2}}\frac{m}{\beta^3 a^3}x^{\frac32}_F e^{-x_F}\,,$$
which for $a=a_F=1$ (unsurprisingly) coincides with the thermal energy density for massive scalar fields in Minkowski spacetime.

\subsubsection{The Total Energy Density}

Collecting the results of this section, we find for the total energy density of our model

$$\frac{\varrho^0 + \varrho^m}{\varrho_0}\simeq\Omega_\Lambda +  \frac{1}{(2\pi)^{3/2}}\frac{m}{\beta_1^3 a^3\varrho_0}x^{\frac32}_F e^{-x_F}\frac{1}{a^3} + \frac{\pi^2}{30\beta_2^4\varrho_0}\frac{1}{a^4}\,,$$
where we wrote $\beta_1$, $\beta_2$ in order to emphasise that the generalised thermal states for the massive and massless conformally coupled scalar fields can have different temperature parameters $\beta$. We recall that the thermal contribution of the massless scalar field has been computed exactly, while the one of the massive scalar field is an approximative result. The above result shows that we indeed succeeded in modelling radiation by a massless scalar field and matter by a massive scalar field in suitable generalised thermal states. Obviously, we can choose the free parameters $m$, $\beta_i$, $x_F$ in such a way that the prefactors of the matter and radiation terms have their correct $\Lambda$CDM--values $\Omega_\mat$ and $\Omega_\rad$, e.g. for the former we could choose the sample values \eqref{eq_kolbvals} with $m\simeq 100$GeV, and for the latter $1/\beta\simeq 1$K, i.e. the temperature of the CMB. Finally, we model Dark Energy simply by a cosmological constant, which in our context appears as a parametrisation of the freedom in defining energy density as an observable.

Since our description of the the standard cosmological model within quantum field theory on curved spacetime reproduces the energy density of the original $\Lambda$CDM--model up to negligible corrections, it can obviously be matched to the observational data as good as this model.

\subsection{Deviations from the Standard Model and their Phenomenological Consequences: $\epsilon J_{00}$ and Dark Radiation}
\label{sec:dark}

Our analysis in the previous sections implies that there exist quantum states in which the total energy density in quantum field theory on curved spacetimes differs from the one in the $\Lambda$CDM--model only by the higher derivative term $\epsilon J_{00}$ and terms which are generally negligible or become important only at redshifts $z\gg 10^9$. The prefactor $\epsilon$ of $J_{00}$ is not determined by the theory but a free parameter so far and it seems advisable to study its impact on the cosmological expansion. Indeed, as the second main result of Section \ref{sec:lcdmqft} we demonstrate in  the following that such a term can provide a natural explanation of Dark Radiation. The fact that the contribution $\epsilon J_{00}$ to the energy density can look like radiation for large $z$ has already been observed in e.g. \cite{chap3_Kofman:1985aw}, but to our knowledge the relevance of this for the phenomenon of Dark Radiation has not been discussed so far.

At this point we would like to mention that in some other works, see e.g. \cite{chap3_Shapiro:2003gm} and references therein, the parameter $\epsilon$ is not considered to be a free parameter, but it is rather taken as determined by the field content of the QFT model, just like the parameter $\gamma$ in \eqref{eq_totaled}. This is motivated by the fact that most common computational schemes for regularising the quantum stress--energy tensor yield the same result for $\epsilon$, which is thus taken to be the correct value. In this work we follow the point of view of \cite{chap3_Walda,chap3_HW04} and thus start from the premise that there is no physical principle within QFT on curved spacetimes which determines the value of $\epsilon$ {\em  a priori}; thus in particular we do not attribute any physical `meaning' to a specific computational scheme, even if most common schemes give the same result.

To start our analysis, we briefly review the notion of Dark Radiation and the related observations.
The fraction $\Omega_\rad$ of the radiation energy density in the $\Lambda$CDM--model is computed as
\begin{equation}\label{eq_radiation}\Omega_\rad=\Omega_\gamma\left(1+\frac78 \left(\frac{4}{11}\right)^{4/3} N_\text{eff}\right)\end{equation}
where $\Omega_\gamma\simeq 5\times 10^{-5}$ is the fraction due to electromagnetic radiation, which can be computed by inserting into
\eqref{eq_thermaldensitymassless} the CMB temperature $T_\text{CMB}\simeq 2.725$K, dividing by today's energy density $\varrho_0=3H^2_0/(8\pi G)\simeq 1.33 \times 10^{-11}$eV (and multiplying by two for the two degrees of freedom of the photon). Moreover, $N_\text{eff}$ is the number of neutrino families and the factor $7/8 (4/11)^{4/3}=0.2271$ takes into account that neutrinos are Fermions and `colder' than the CMB photons, because they have decoupled from the hot early bath in the universe earlier than electrons and positrons and have thus not been `heated up' by the decoupling of the latter in contract to the photons. The standard value for $N_\text{eff}$ is not 3 as one would expect, but rather $N_\text{eff}=3.046$ because the value $7/8 (4/11)^{4/3}$ in \eqref{eq_radiation} is computed assuming e.g. instantaneous decoupling of the neutrinos and corrections have to be taken into account in a more detailed analysis \cite{chap3_Mangano:2005cc}; it is customary to take these corrections into account by considering $N=3.046$ as the standard value of the `neutrino family number' rather than changing the factor $7/8 (4/11)^{4/3}$ in this formula, hence the nomenclature $N_\text{eff}$. Consequently, it is convenient to parametrise any contribution to $\Omega_\rad$ which is not due to electromagnetic radiation and the three neutrino families in the standard model of particle physics by $\Delta N_\text{eff}\doteq N_\text{eff}-3.046$.

One of the two main observational inputs to determine $\Omega_\rad$ and thus $N_\text{eff}$ is the primordial fraction of light elements in the early universe as resulting from the so-called Big Bang Nucleosynthesis (BBN), which has occurred at around $z\simeq 10^9$ and thus in the radiation--dominated era. This is the case because the nucleosynthesis processes which happened at that time depend sensitively on the expansion rate $H\simeq H_0 \sqrt{\Omega_\rad}/a^2$, see e.g. \cite{chap3_Kolb:1990vq,chap3_Kneller:2004jz,chap3_Ellis}. The other main observational source for the determination of $N_\text{eff}$ is the cosmic microwave background radiation (CMB). This radiation was emitted at about $z\simeq 1100$, but the CMB power spectrum is sensitive to the expansion before this point, e.g. to the redshift $z_\text{eq}$ at which the energy densities of matter and radiation were equal, see \cite[Section 6.3]{chap3_Ade:2013zuv} and the references therein for details; for standard values, $z_\text{eq}\simeq 3000$. 

The observations to date do not give a conclusive value for $N_\text{eff}$, and the value inferred from observational data depends on the data sets chosen. The Planck collaboration \cite{chap3_Ade:2013zuv} reports e.g. values of $N_\text{eff}=3.36^{+0.68}_{-0.64}$ at 95\% confidence level from combined CMB power spectrum data sets, $N_\text{eff}=3.52^{+0.48}_{-0.45}$ at 95\% confidence level from combining these data sets with direct measurements of the Hubble constant $H_0$ and of the power spectrum of the three-dimensional distribution of galaxies (so--called baryon acoustic oscillation, BAO), and $N_\text{eff}=3.41\pm 0.30$ at 68\% confidence level from combining the CMB power spectrum data sets with BBN data. Yet, one can infer from these values that there is a mild, but not very significant, preference for $\Delta N_\text{eff}>0$. Thus there has been an increasing interest in models which can explain a potential excess in radiation and thus $\Delta N_\text{eff}>0$. Most of these models assume additional particles/fields, e.g. a fourth, sterile, neutrino, whereas other consider geometric effects from e.g. modifications of General Relativity. Moreover, in most models $\Delta N_\text{eff}$ is constant and thus affects BBN and CMB physics alike, while in others $\Delta N_\text{eff}$ is generated only after BBN and thus affects only CMB physics. 

In the following we shall propose a new and alternative explanation for Dark Radiation which follows naturally from our analysis of the $\Lambda$CDM--model in quantum field theory on curved spacetimes and has the interesting characteristic that it generates a value of $\Delta N_\text{eff}$ which {\em  increases} with $z$ and thus affects BBN physics more than CMB physics.

Following the motivation outlined at the beginning of this section, we solve the equation
\begin{equation}
\label{eq_odebox}
\frac{H^2}{H^2_0}=\Omega_\Lambda + \frac{\Omega_\mat}{a^3} + \frac{\Omega_\rad}{a^4} + \epsilon \frac{J_{00}}{H^4_0}\,, 
\end{equation}
which can be rewritten as a second order ordinary differential equation for $H$ in $z$, numerically 
with $\Lambda$CDM--initial conditions $H(z=0)=H_0$, $\partial_zH(z=0)=H_0(3\Omega_\mat+4\Omega_\rad)/2$. As before, we consider for definiteness $\Omega_\mat=0.30$, $\Omega_\Lambda=1-\Omega_\mat-\Omega_\rad$, because the exact values of these parameters are not essential for our analysis. Looking at the characteristics of the solution to this ordinary differential equation, it turns out that a non--zero $\epsilon$ generates a time--varying $\Delta N_\text{eff}>0$. In more detail, we define for the solution $H$ of \eqref{eq_odebox}
$$\Delta N_\text{eff}(z)\doteq \frac{\frac{H^2}{H^2_0}-\Omega_\Lambda - \Omega_\mat(1+z)^3-\Omega_\rad(1+z)^4}{0.2271(1+z)^4}\,,$$
and sample this observable at the redshift $z=10^9$ associated to BBN physics and at the redshift $z=3000$ associated to CMB physics. We collect our results in Figure \ref{fig_deltaneffofalpha}.

\begin{figure}
\includegraphics[width=1\columnwidth]{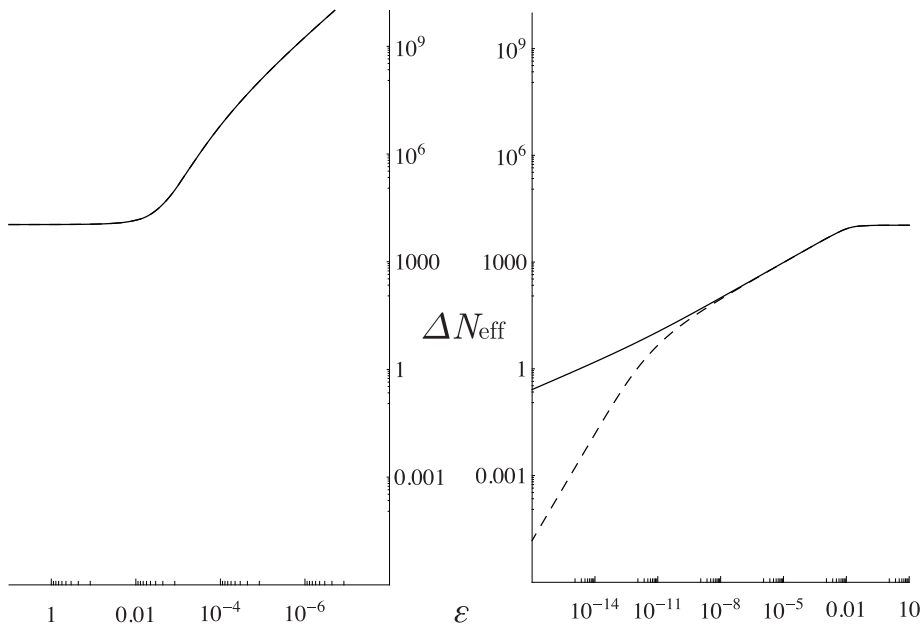}
\caption{\label{fig_deltaneffofalpha}$\Delta N_\text{eff}(z)$ depending on $\epsilon$ for $z=10^9$ (BBN, solid line) and $z=3\times10^3$ (CMB, dashed line). For $\epsilon<0$ and $\epsilon$ positive and large enough, the values at the two redshifts coincide because the maximum value of $\Delta N_\text{eff}(z)$ is reached already for $z<3\times10^3$ in these cases.}
\end{figure}

\begin{figure}
\includegraphics[width=1\columnwidth]{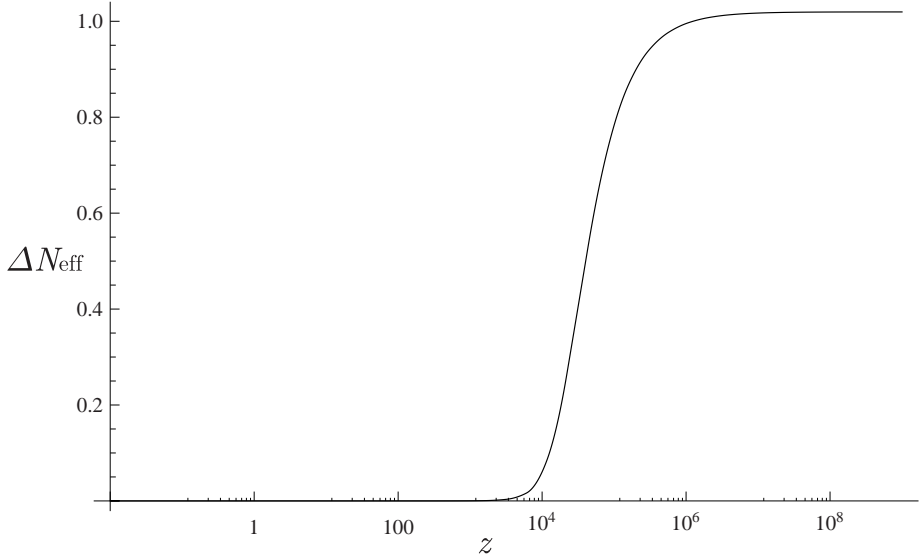}
\caption{\label{fig_deltaneffofz}$\Delta N_\text{eff}(z)$ for $\epsilon=2\times 10^{-15}$.}
\end{figure}

As can be inferred from this figure, $\Delta N_\text{eff}(z)$ is monotonically increasing in $\epsilon$, where positive and negative values of $\epsilon$ result in very different behaviours. For positive values of $\epsilon$ one finds that $\Delta N_\text{eff}(z)$ vanishes in the limit of vanishing $\epsilon$, as one would expect. On the other hand, it turns out that for negative values of $\epsilon$, $\Delta N_\text{eff}(z)$ diverges as $\epsilon$ approaches zero. While this seems to be puzzling at first sight, it fits well with previous qualitative analyses of the effect of the higher derivative term $J_{00}$. In fact, it is known that the inclusion of this higher derivative term can lead to unstable solutions of the semiclassical Einstein equations, where for $\epsilon<0$ ($\epsilon>0$) the class of solutions we consider here, effectively fixed by the $\Lambda$CDM initial conditions, turns out to be unstable (stable), see e.g. \cite{chap3_Anderson1, chap3_FW96, chap3_Haensel, chap3_Koksma:2009sa, chap3_ParkerSimon, chap3_Star} and the discussion at the end of this section. Thus, the divergence of $\Delta N_\text{eff}(z)$ as $\epsilon$ approaches zero from below can be just interpreted as a sign of this instability. 

In \cite{chap3_Wald4,chap3_DFP}, $\epsilon=0$ has been chosen on conceptual grounds in order to discard unstable solutions altogether. However, as we see here, a non--zero $\epsilon$ can have interesting phenomenological implications. After all, taking quantum field theory on curved spacetimes seriously, $\epsilon$ is a free parameter of the  theory, which we can only fix in a more fundamental theory or by observations. Indeed, we see  in Figure \ref{fig_deltaneffofalpha} that $\epsilon<0$, corresponding to an unstable solution of the semiclassical Einstein equation, is already ruled out by observations because it generally leads to $\Delta N_\text{eff}(z)\gg 1$ which is  certainly not compatible with value of $\Delta N_\text{eff}\simeq 0.5-1.0$ inferred from observations as mentioned above. On the other hand we see that in order to not exceed $\Delta N_\text{eff}=1$ at both BBN and CMB we have to choose $0\le \epsilon < 2\times 10^{-15}$, thus, without performing a detailed fit of BBN and CMB data, we can say that the values for $\Delta N_\text{eff}$ reported e.g. by the Planck collaboration in \cite{chap3_Ade:2013zuv} give an upper bound of about $2\times 10^{-15}$ for $\epsilon$. We plot $\Delta N_\text{eff}(z)$ for for this value of $\epsilon$ and redshifts $0<z<10^9$ in Figure \ref{fig_deltaneffofz}. As already anticipated in Figure \ref{fig_deltaneffofalpha}, one can nicely see how $\Delta N_\text{eff}(z)$ is monotonically growing in $z$, with $\Delta N_\text{eff}(z=0)=0$ as fixed by our initial conditions. Moreover one can see clearly that if one wants to meet the bounds on $\Delta N_\text{eff}$ at the BBN redshift, the excess in the effective number of neutrinos at the CMB is negligible, which is the characteristic signature of this potential explanation for Dark Radiation. We have not considered the influence of the initial conditions for \eqref{eq_odebox} on $\Delta N_\text{eff}$, but we expect that for the initial conditions compatible with low--$z$ observational data such as Supernova type Ia data and Baryon Acoustic Oscillation data, $\Delta N_\text{eff}$ will not differ considerably from the form we found as these data will not allow for large deviations in the initial conditions from the $\Lambda$CDM ones we chose.

As a further, rather pedagogical remark, we would like to comment on the fact that, for large absolute values of $\epsilon$, $\Delta N_\text{eff}$ does not depend on the sign of $\epsilon$, as can be seen from Figure. \ref{fig_deltaneffofalpha}. This phenomenon can be understood as follows. Naturally, for large absolute values of $\epsilon$, the other terms in \eqref{eq_odebox} become negligible and one effectively solves for $J_{00}=0$. The solution of this ordinary differential equation with initial conditions $H(z=0)=c$, $\partial_zH(z=0)=d$ is
$$H(z)=\frac{c^{1/3}d^{2/3}\left(\frac{2c}{d}-1+(1+z)^3\right)^{2/3}}{2^{2/3}}$$
and thus, inserting the $\Lambda$CDM initial conditions $H(z=0)=H_0$, $\partial_zH(z=0)=H_0(3\Omega_\mat+4\Omega_\rad)/2$ we, find for large $z$, $N_\text{eff}(z)\simeq 10^4$ as in Figure \ref{fig_deltaneffofalpha}.

The above analysis may be interpreted as providing an upper bound for the parameter $\epsilon$ in $\epsilon J_{00}$, which quantifies higher order derivative corrections to classical General Relativity, based on BBN and CMB data. Two further ways to obtain bounds on $\epsilon$ are known to us. To this avail, let us recall that the tensors $I_{\mu\nu}$ and $J_{\mu\nu}$ in \eqref{eq_deftmunu} can be obtained as variational derivatives with respect to the metric of the Lagrangean densities $\sqrt{-g}R^2$ and $\sqrt{-g}R_{\mu\nu}R^{\mu\nu}$. Consequently, an upper bound on $\epsilon$ can be obtained by studying modifications of the Newtonian potential arising from such higher--derivative modifications of the gravitational action. To wit, the Lagrangean
$$L=\sqrt{-g}\left(\frac{R}{16\pi G}+c_1 R^2 + c_2 R_{\mu\nu}R^{\mu\nu}\right)$$
leads to the Newtonian potential of a point mass $m$ \cite{chap3_Stelle:1977ry}
\begin{equation}\label{eq_newton}
\phi = \frac{-m G}{r}\left(1+\frac{1}{3}e^{-m_1 r}-\frac{4}{3}e^{-m_2 r}\right)\end{equation}
$$m_1=\frac{1}{\sqrt{32 \pi G(-3c_1-c_2)}}\qquad m_2=\frac{1}{\sqrt{16\pi G c_2}}\,.$$
Using recent data \cite{chap3_Kapner:2006si} from torsion--balance experiments to test the gravitational inverse--square law at $\sim 10^{-4}$m and assuming that the two Yukawa corrections don't cancel each other at this length scale, one obtains $-c_1,c_2<10^{61}$ \cite{chap3_Calmet:2008tn}. To compare this with our results, we recall that in our treatment these higher curvature terms appear on the right hand side of the semiclassical Einstein equation and that we have computed in units of $H_0$, thus we have 
$$\epsilon=\frac{(-3c_1-c_2)8\pi G H^2_0}{3}\simeq (-3c_1-c_2)\times 10^{-121}$$
which would imply $\epsilon < 10^{-60}$ and thus a stronger bound than the one we inferred from cosmological observations. Of course such a low value of $\epsilon$ leads to $\Delta N_\text{eff}\ll 1$ at both BBN and CMB and thus the amount Dark Radiation generated by this term would be negligible.

Notwithstanding, there are still several aspects of our analysis which are of interest. First of all, our bound on $\epsilon$ is completely independent from the one inferred from laboratory experiments and can thus be considered as an additional confirmation of those results. Moreover, it is still possible that the Yukawa corrections in \eqref{eq_newton} cancel each other on the length scales relevant for the experiments described in \cite{chap3_Kapner:2006si}, such that $\epsilon$ could be as large as our upper bound, which in this case would give a real bound on one and hence both Yukawa corrections. Finally, the bounds inferred from \cite{chap3_Kapner:2006si} and from our analysis stem from phenomena on completely different length scales. As a rough estimate we note that the diameter of our observable universe, which today is about $6/H_0\simeq 10^{27} $m, was at e.g. $z=10^9$ still $10^{18}$m and thus much larger than the submillimeter scales relevant for the torsion--balance experiments. Thus it could be that effects we have not considered so far, e.g. state--dependent effects which are due to the small--scale structure of the quantum states we have fixed only on cosmological scales so far, affect the comparison between the two different sources of input for the determination of $\epsilon$.

Finally, a further bound on $\epsilon$ can be obtained by analysing the effects of higher derivative contributions to the gravitational Lagrangean in the context of Inflation. In fact, the first inflationary model proposed by Starobinsky in \cite{chap3_Star} was based on the observation that the semiclassical Friedmann equation for massless (and in the scalar case conformally coupled) quantum fields in the conformal vacuum state, namely,
\begin{equation}\label{eq_staro}\frac{H^2}{H^2_0}=\gamma \frac{H^4}{H^4_0}+\epsilon\frac{J_{00}}{H^4_0}+\Omega_\Lambda\,,
\end{equation}
has solutions which correspond to a de Sitter phase $a(t)\propto \exp(H_+t)$ with $H_+ \simeq H_0/\sqrt{\gamma}$ which for $\epsilon>0$ `decays' to a de Sitter phase $a(t)\propto \exp(H_-t)$ with $H_- \simeq H_0 \sqrt{\Omega_\Lambda}$. For various reasons, this model was discarded as a viable inflationary model, see e.g. the discussion in Section 9.4 of \cite{chap3_Linde:2005ht}, and Starobinsky instead proposed an explanation for Inflation based only on \eqref{eq_staro} with $\gamma=0$ in \cite{chap3_Starobinsky:1983zz}; thus Inflation in this model occurs due to an $R^2$-modification of the gravitational action. This model is now widely known as the {\em  Starobinsky model}\index{Starobinsky model} and still in very good agreement with CMB data \cite{chap3_Ade:2013uln};  interestingly, it is mostly attributed to Ref. \cite{chap3_Star} rather than Ref. \cite{chap3_Starobinsky:1983zz}, although the two models are quite different. 

Neglecting the trace anomaly term $\gamma H^4/H^4_0$ in \eqref{eq_staro} can be justified if $\epsilon\gg \gamma$ and if $\partial_z H/H$ is not too small, but note that $\gamma>0$ is in contrast to $\epsilon>0$ an {\em  a priori} {\em  prediction} of QFT on curved spacetimes and that the trace anomaly term will always dominate for $z$ large enough in almost all solutions of \eqref{eq_staro}. However, this occurs generally where $H$ of the order of the Planck mass $M_P=1/\sqrt{G}\simeq 10^{19}$GeV (recall that $\gamma\simeq 10^{-122}$ in the case of a scalar field and thus $H_+ \simeq H_0/\sqrt{\gamma}\simeq M_P$ in the original Starobinsky model), and one can argue that the semiclassical Einstein equation should not be trusted in this regime, see e.g. \cite{chap3_Wald4}. Notwithstanding, taking \eqref{eq_staro} seriously and solving it with $\Lambda$CDM-initial conditions, i.e. requiring that today we are in the `lower de Sitter branch' ruled by $H_-$, already leads to the lower bound $\epsilon>0$ because with $\epsilon=0$ and these initial conditions, $H$ becomes imaginary for large $z$, cf. \cite{chap3_Wald4,chap3_DFP}.

A more concrete lower bound on $\epsilon$ can be obtained by confronting the Starobinsky model with CMB data. In this model, the parameter here called $\epsilon$ (not to be confused with the usual slow--roll parameter $\epsilon$ in Inflation) is quantified in terms of the `scalaron mass' \cite{chap3_Star,chap3_Starobinsky:1983zz}
$$M=\frac{H_0}{\sqrt{6\epsilon}}\,,$$
whose magnitude determines the amplitude of the primordial fluctuations, which turns out to be linear in $M$ \cite{chap3_Linde:2005ht}. Comparison with CMB data yields \cite{chap3_Kaneda:2010ut}
\begin{equation}\label{eq_starovalues}
M\simeq 10^{-5} M_P\qquad \Rightarrow \qquad \epsilon\simeq 10^{-113}\,.
\end{equation}
Thus, if Inflation occurred due to the $\epsilon J_{00}$ contribution to the energy density, then $\epsilon$ is too small for generating a considerable amount of Dark Radiation. However, if Inflation has a different origin, then we can interpret the value in \eqref{eq_starovalues} as an an upper bound on $M$ and thus obtain the lower bound $\epsilon > 10^{-113}$.

After discussing potential phenomenological implications of the energy density contribution $\epsilon J_{00}$ which is not present in the $\Lambda$CDM--model, we note that potential further corrections to the $\Lambda$CDM--model can come from specifics of the quantum state we have neglected in our analysis. We have chosen the quantum states in our discussion such that their characteristic energy density was entirely of thermal nature, but we have seen that also pure, non--thermal states can have contributions to the energy density which scale like $a^{-4}$, cf. Figure \ref{fig_rhonormlargez}. It could be that there exist states which are compatible with observations and have sizable energy--density contributions of this kind; these states would then provide a further alternative explanation for Dark Radiation which does not call for the introduction of new fields and the associated particles.

\section{A Birds--Eye View of Perturbations in Inflation}
\label{sec:inflation}

A prominent application of quantum field theory in curved spacetime is the analysis of perturbations in Inflation. In the simplest models of Inflation, the characteristic exponential expansion is driven by the energy density of a classical scalar field $\varphi$. The perturbation $\phi$ of this field, combined in a gauge--invariant way with the perturbations $\gamma$ of the classical metric $g$, are considered as quantum fields propagating on the classical background $M_\varphi=(M,g,\varphi)$, where $M=(M,g)$ is a spatially flat FLRW spacetime and $\varphi$ depends only on time in the homogeneous and isotropic FLRW coordinates. Rephrased in more abstract terms, one may say that the analysis of perturbations in inflation consists in quantizing the field theory of the tuple $\Gg = (g,\varphi)^t$ perturbatively around a background which satisfies the classical Einstein equation and is of FLRW--type. This perturbative quantum field theory is usually truncated at linear order.

The standard textbook treatment of the quantum theory of perturbations in Inflation, see e.g. \cite{chap3_Straumann:2005mz}, consists in using the FLRW--symmetry of the background $M_\varphi$ in order to split the metric perturbation $\gamma$ into components which transform as scalars, vectors and tensors under the isometry group of FLRW--backgrounds, the Euclidean group ${\mathbb E}^3$. Subsequently, gauge--invariant linear combinations of these components and the scalar field perturbation $\phi$ are identified and quantized in a canonical fashion. As the above--mentioned splitting is non--local and depends heavily on the FLRW--symmetry, it is a priori not clear whether it captures all local observables of the theory. A systematic analysis of this issue from the point of view of algebraic quantum field theory has been performed in \cite{chap3_ThomasInflation}. In this section we shall review the main steps and results of this analysis. For simplicity we consider the special and commonly assumed case where the scalar field is minimally coupled to the metric, the case of general coupling is treated in \cite{chap3_ThomasInflation}.

\subsection{The Linearised Einstein--Klein--Gordon System on General Curved Spacetimes}
\label{sec:lekgs}

We first consider the quantization of the linearised Einstein--Klein--Gordon system on arbitrary backgrounds $M_\varphi=(M,g,\varphi)$, such that $M$ is four--dimensional, $M=(M,g)$ is globally hyperbolic, and $\Gg = (g,\varphi)^t$ satisfies the coupled Einstein--Klein--Gordon equations. Later we point our attention towards such backgrounds $M_\varphi$ which are in addition of FLRW--type and compare the general quantization procedure with the usual approach to the quantization of perturbations in Inflation.

The linearised Einstein--Klein--Gordon system is a Bosonic gauge theory with linearised diffeomorphisms as local gauge transformations. Thus, the strategy to quantize this system is to show that it fits in the axiomatic framework of general linear gauge theories on curved spacetimes reviewed in Section \ref{sec:linear_gauge}. In particular, we have to show that this model satisfies the axioms of Definiton \ref{def_lineargauge}. 

To this avail, we introduce the vector bundles over $M$ $\cV\doteq \bigvee^2 T^*M \oplus \left(M\times \RR\right)$, where $\bigvee$ denotes the symmetric tensor product, and $\cW \doteq TM$. The space of smooth sections of a vector bundle such as $\cV$ will be denoted by $\Se(\cV)$ as before. We recall that important subspaces of $\Se(\cV)$ are $\Sec(\cV)$ and $\Sesc(\cV)$ the space of smooth sections of compact and space--like compact support, respectively. Both the background fields $\Gg = (g,\varphi)^t$ and their perturbations $\Gamma = (\gamma,\phi)^t$ are elements of $\Se(\cV)$, whereas gauge--transformations (linearised diffeomorphisms) will be parametrised by $\varsigma\in\Se(\cW)$. We introduce on such sections symmetric and non--degenerate bilinear forms by
\beq\label{eq_formV}
\langle \Gamma_1, \Gamma_2\rangle_\cV \doteq \int\limits_M \vol \left(g^{\alpha\beta} g^{cd}\gamma_{1,ac}\gamma_{2,bd}+\phi_1\phi_2\right)\,,
\eeq
\beq\label{eq_formW}
\langle\varsigma_1, \varsigma_2\rangle_\cW  \doteq \int\limits_M  \vol g^{\alpha\beta}\varsigma_{1,a}\varsigma_{2,b}\,.
\eeq
These bilinear forms are well--defined for pairs of sections with compact overlapping support.

The starting point of the analysis is the Einstein--Hilbert--Klein--Gordon action for $\Gg = (g,\varphi)^t\in\Se(\cV)$
$$
S(\Gg)=\int\limits_M \vol\left(\frac{R}{2}-\frac{(\nabla_\mu\varphi)\nabla^\mu\varphi}{2}-V(\varphi)\right),
$$
where $V(\varphi)$ is an arbitrary smooth potential. In this section, we shall use units in which $8\pi G = 1$, $G$ being Newton’s gravitational constant; this renders the Klein--Gordon field, and thus $\Gg = (g,\varphi)^t$, dimensionless. The Euler--Lagrange equations of $S(\Gg)$ are the Einstein--Klein--Gordon--equations
\beq\label{eq_fullcoupled}
\Se(\cV)\ni \EL (\Gg)=\left(\begin{array}{c}
E^2_{\alpha\beta}\\E^0
\end{array}\right)\doteq\left(\begin{array}{c}
\frac{1}{2}(G_{\alpha\beta}- T_{\alpha\beta})\\
-\Box\varphi + \partial_\varphi V
\end{array}\right)=0\,.
\eeq

In order to obtain the linearised theory, we split $\Gg$ into a background $\Gg$ (which, slightly abusing notation, we denote by the same symbol) and a perturbation $\Gamma$ and formally expand (omitting boundary terms)
$$
S(\Gg+\Gamma)=S(\Gg)-\left\langle \EL (\Gg),\Gamma\right\rangle_\cV -\frac12\left\langle\Gamma,P\Gamma\right\rangle_\cV+\cO(\Gamma^3).
$$
Here $P$ is the second order partial differential operator 
\begin{gather}P:\Se(\cV)\to\Se(\cV)\qquad P=\begin{pmatrix}P_0 & P_2\\P_3 & P_1\end{pmatrix}\label{eq_originaleom}\\
(P_0 \gamma)_{\alpha\beta}=\frac{1}{4}\left(-\nabla^\mu\nabla_\mu \gamma_{\alpha\beta} + 2 \nabla^\mu\nabla_{(\alpha}\gamma_{\beta)\nu}-g_{\alpha\beta}\nabla^\mu\nabla^\nu\gamma_{\mu\nu}-\nabla_\alpha\nabla_\beta{\gamma}_{\mu}^{\phantom{c}\mu}\right.\notag\\
\left. +g_{\alpha\beta}\nabla_\mu\nabla^\mu{\gamma}_{\nu}^{\phantom{c}\nu}+g_{\alpha\beta}R^{\mu\nu}\gamma_{\mu\nu}-R\gamma_{\alpha\beta}\right)+
\left(\frac{1}{4}(\nabla_\mu\varphi)(\nabla^\mu\varphi)+\frac12 V\right) \gamma_{\alpha\beta}+\notag\\
-\frac{1}{4}g_{\alpha\beta}\gamma_{\mu\nu}(\nabla^\mu\varphi)\nabla^\nu\varphi \notag- 2 {E^2}^\mu_{\phantom{\alpha}(\alpha} \gamma_{\beta)\mu}+\frac12 E^2_{\alpha\beta} {\gamma}_{\mu}^{\phantom{c}\mu}\\
(P_2\phi)_{\alpha\beta}=\left\{g_{\alpha\beta}\left(\frac{1}{2}(\nabla^\mu\varphi)\nabla_\mu+\frac12 \partial_\varphi V\right)-\left(\nabla_{(\alpha}\varphi\right)\nabla_{\beta)}\phantom{\frac12}\!\!\!\!\right\}\phi\notag\\
P_3 \gamma = \left((\nabla^\mu\varphi)\nabla^\nu-\frac12 (\nabla^\alpha\varphi)\nabla_\alpha g^{\mu\nu}+(\nabla^\mu\nabla^\nu\varphi)+\frac12 E_0 g^{\mu\nu}\right) \gamma_{\mu\nu}\notag\\
P_1\phi=\left(-\nabla_\mu\nabla^\mu+\partial^2_\varphi V\right)\phi\notag\,.
\end{gather}
$P$ is formally self--adjoint with respect to $\langle \cdot,\cdot\rangle_\cV$ even if the background is off--shell, i.e. even if $\EL (\gG)\neq 0$, cf. the discussion in \cite[Section 2.2.1]{chap3_ThomasInflation}.

$S(\Gg+\Gamma)$ is invariant under diffeomorphisms of $M$, in particular under those generated by an arbitrary but fixed compactly supported vector field $\varsigma \in\Sec(\cW)$. Given such a diffeomorphism, $\Gg+\Gamma$ transforms as
$$
\Gg+\Gamma\mapsto \Gg+\Gamma+\cL_\varsigma \Gg+\cL_\varsigma \Gamma  +\cO(\varsigma^2)\,,
$$
where $\cL_\varsigma$ denotes the Lie derivative w.r.t. $\varsigma$. To first order in $\varsigma$ and $\Gamma$, the diffeomorphism-invariance of $S(\Gg+\Gamma)$ reads
$$
P\cL_\varsigma \Gg=\cL_\varsigma \EL (\Gg)\,,
$$
where the term on the right hand side arises from the $\cL_\varsigma \Gamma$ contribution of the transformed $\Gg+\Gamma$. These observations imply the following: we may consistently truncate the diffeomorphism--invariant field theory for $\Gg+\Gamma$ at joint linear order in $\Gamma$ and $\varsigma$ if and only if we assume that the background $\Gg$ is on--shell, i.e. $\EL (\Gg)=0$\footnote{Strictly speaking $\cL_\varsigma \EL (\Gg)=0$ is satisfied even if $E_0=-\Box\varphi +  \partial_\varphi V=c$ with $c$ constant but non--zero. However, one may absorb $c$ by redefining $V(\varphi)$.}; for this reason we shall only consider on--shell backgrounds in the following. In this case, one may think of the linearised Einstein--Klein--Gordon theory as originating from the quadratic action
$$
S^{(2)}(\Gamma)\doteq-\frac12\left\langle\Gamma,P\Gamma\right\rangle_\cV\,,
$$
which is for all $\varsigma\in\Sec(\cW)$ invariant under the affine transformation
$$
\Gamma\mapsto \Gamma+\cL_\varsigma \Gg\,.
$$
Defining $K:\Se(\cW)\to\Se(\cV)$ by $K\varsigma \doteq\cL_\varsigma \Gg$, we may express this gauge--invariance as
$$
P\circ K=0\,.
$$
We recall that this automatically implies that the equation of motion $P\Gamma=0$ for the perturbations does not have a well--posed Cauchy problem as non--trivial solutions with compact support exist.

In order to cast the linearised Einstein--Klein--Gordon theory into a form which satisfies Definition \ref{def_lineargauge}, such that the results of Section \ref{sec:linear_gauge} may be readily applied, we introduce a field redefinition which generalises the trace--reversal well--known from linearised gravity (the appearing numerical factors are introduced in order to homogenise the normalisation of the $\gamma$--$\gamma$ and $\phi$--$\phi$ components of the principal symbol of $P$).
$$
\widebar {\,\cdot\,}:\Se(\cV)\mapsto \Se(\cV)\qquad \Gamma=\left(\begin{array}{c}\gamma_{\alpha\beta}\\\phi\end{array}\right)\mapsto\widebar {\Gamma}=\left(\begin{array}{c}\displaystyle\frac{1}{4}\left(\gamma_{\alpha\beta}-\frac12 g_{\alpha\beta}{\gamma}_{\mu}^{\phantom{c}\mu}\right)\\\displaystyle\phi\end{array}\right)
$$
Using this field redefinition and observing that it is invertible, we may now define 
$$
{\widebar {P}}\doteq P\circ \widebar {\,\cdot\,}^{-1}\qquad \widebar {K}\doteq\widebar {\,\cdot\,}\circ K\qquad \langle \cdot ,\cdot\rangle_{\widebar {\cV}}  \doteq \langle \widebar {\,\cdot\,}^{-1} ,\cdot\rangle_\cV\qquad \Theta\doteq\begin{pmatrix}\theta_{\alpha\beta}\\\phi\end{pmatrix}\doteq\widebar {\Gamma}\,.
$$
These definitions are tailored in such a way that the second order action for $\Gamma$ may now be re--written as
$$
S^{(2)}(\Gamma) = -\frac12 \langle \Gamma, P \Gamma\rangle_\cV=-\frac12\langle \Theta,{\widebar {P}}\Theta\rangle_{\widebar {\cV}} \doteq \widebar {S}^{(2)}(\Theta)\,.
$$
Moreover, $\widebar {S}^{(2)}(\Theta)$ is invariant under the affine transformation $\Theta\mapsto \Theta +\widebar K\varsigma$ for all $\varsigma\in\Sec(\cW)$, $\langle \cdot ,\cdot\rangle_{\widebar {\cV}}$ is symmetric and non--degenerate and ${\widebar {P}}$ is formally self--adjoint with respect to the redefined bilinear form. 

We now provide the expressions for ${\widebar {P}}$, $\widebar {K}$ and $\langle \cdot, \cdot\rangle_{\widebar {\cV}} $. As we always assume that the background metric and scalar field satisfy the full Einstein--Klein--Gordon equations, we have used these to simplify ${\widebar {P}}$.

\begin{equation}\label{eq_redefinedP}
\widebar {P}:\Se(\cV)\to\Se(\cV)\qquad \widebar {P}=\begin{pmatrix}\widebar {P}_0 & \widebar {P}_2\\\widebar {P}_3 & \widebar {P}_1\end{pmatrix}\end{equation}
$$(\widebar {P}_0 \theta)_{\alpha\beta}=-\nabla_\mu\nabla^\mu \theta_{\alpha\beta}+2\nabla^\mu \nabla_{(\alpha } \theta_{\beta )\mu} - g_{\alpha\beta} \nabla^\mu\nabla^\nu \theta_{\mu\nu}-2 V\theta_{\alpha\beta}$$
$$\left(\widebar {P}_2 \phi\right)_{\alpha\beta} = \left(-(\nabla_{(\alpha }\varphi)\nabla_{\beta )}+\frac{1}{2}g_{\alpha\beta}(\nabla^\mu \varphi)\nabla_\mu+\frac{\partial_\varphi V}{2}g_{\alpha\beta}\right)\phi$$ 
$$\widebar {P}_3 \theta = 4\left(-\frac{1}{2}\partial_\varphi V g^{\alpha\beta}+\left(\nabla^\alpha\varphi\right)\nabla^\beta\right)\theta_{\alpha\beta}$$
$$\widebar {P}_1 \phi = \left(-\nabla_\alpha\nabla^\alpha+\partial^2_\varphi V\right)\phi\,.$$

\begin{equation}\label{eq_redefinedK}\widebar {K}:\Se(\cW)\to\Se(\cV)\qquad \widebar {K}\varsigma=\begin{pmatrix}\displaystyle\frac{1}{2}\left( \nabla_{(\alpha }\varsigma_{\beta )}-\frac12 g_{\alpha\beta} \nabla_\mu\varsigma^c\right)\\\displaystyle\varsigma^c\nabla_\mu\varphi\end{pmatrix}\end{equation}

\begin{equation}\label{eq_redefinedForm}\langle \Theta_1,\Theta_2\rangle_{\widebar {\cV}} =
\int\limits_M \vol \left(4{\theta_1}_{\alpha\beta}\theta_2^{\alpha\beta}-2{\theta_1}_{\mu}^{\phantom{\mu}\mu}{\theta_2}_{\nu}^{\phantom{\mu}\nu}+\phi_1\phi_2\right)\end{equation}

In order to satisfy the remaining conditions of Definition \ref{def_lineargauge}, we need to check whether $R\doteq \widebar {K}^\dagger\circ \widebar {K}$ is Cauchy--hyperbolic and whether there exists a `gauge--fixing operator' $T:\Se(\cW)\to\Se(\cV)$ such that the `gauge--fixing equation of motion operator' $\widetilde P \doteq \widebar{P}+T\circ \widebar{K}^\dagger$ is Green--hyperbolic and $Q\doteq \widebar {K}^\dagger\circ T$ is Cauchy--hyperbolic as well. It is not difficult to check that these conditions are satisfied by choosing $T\doteq 2\widebar{K}$, cf. \cite{chap3_ThomasInflation}. In particular, with this choice of $T$, $\widetilde P$ is of the form
\begin{equation}\label{eq_PTilde}
\widetilde{P}:\Se(\cV)\to\Se(\cV)\qquad \widetilde{P}=\begin{pmatrix}\widetilde{P}_0 & \widetilde{P}_2\\\widetilde{P}_3 & \widetilde{P}_1\end{pmatrix}
\end{equation}
$$(\widetilde{P}_0 \theta)_{\alpha\beta}=-\nabla_\mu\nabla^\mu \theta_{\alpha\beta}-2R_{\alpha\phantom{b}\beta}^{\phantom{a}\mu\phantom{a}\nu}\theta_{\mu\nu}+2(\nabla^\mu \varphi)(\nabla_{(\alpha} \varphi)\theta_{\beta)\mu}$$
$$(\widetilde{P}_2 \phi)_{\alpha\beta} =(\nabla_\alpha\nabla_\beta\varphi)\phi$$
$$\widetilde{P}_3 \theta = 4\left((\nabla^\alpha\nabla^\beta\varphi)-\frac{1}{2}\partial_\varphi V g^{\alpha\beta}\right)\theta_{\alpha\beta}$$
$$\widetilde{P}_1 \phi = \left(-\nabla_\alpha\nabla^\alpha+\partial^2_\phi V+\left(\nabla\varphi\right)^2\right)\phi\,.$$ 
Consequently, we have the following result.

\begin{proposition}
\label{prop_propLEKGS}The linearised Einstein--Klein--Gordon system defined by\\ $(M_\varphi,\cV,\cW,{\widebar {P}},\widebar {K})$, where 
\begin{itemize}
\item $M_\varphi\doteq(M,g,\varphi)$ with $(M,g)$ a four-dimensional globally hyperbolic spacetime and $(g,\varphi)$ a solution of the Einstein--Klein--Gordon equations \eqref{eq_fullcoupled}
\item $\cV\doteq \bigvee^2 T^*M \oplus \left(M\times \bR\right)$ and $\cW\doteq TM$ real vector bundles over $M$
\item the spaces of smooth sections $\Se(\cV)$ and $\Se(\cW)$ of $\cV$ and $\cW$ over $M$ are endowed with the bilinear forms $\langle\cdot,\cdot\rangle_{\widebar {\cV}} $ \eqref{eq_redefinedForm} and $\langle\cdot,\cdot\rangle_\cW$ \eqref{eq_formW}
\item ${\widebar {P}}$ is the differential operator ${\widebar {P}}:\Se(\cV)\to\Se(\cV)$ defined in \eqref{eq_redefinedP} and $\widebar {K}$ is the differential operator $\widebar {K}:\Se(\cW)\to\Se(\cV)$ defined in \eqref{eq_redefinedK}
\end{itemize}
satisfies the axioms of linear gauge theories in Definition \ref{def_lineargauge}:
\begin{enumerate}
\item $\langle\cdot,\cdot\rangle_{\widebar {\cV}} $ and $\langle\cdot,\cdot\rangle_\cW$ are symmetric and non--degenerate.
\item ${\widebar {P}}$ is formally selfadjoint with respect to $\langle\cdot,\cdot\rangle_{\widebar {\cV}} $ and satisfies ${\widebar {P}}\circ \widebar {K}=0$.
\item The differential operator $R:\Se(\cW)\to\Se(\cW)$, $R\doteq\widebar {K}^\dagger \circ \widebar {K}$, with $\widebar {K}^\dagger:\Se(\cV)\to\Se(\cW)$, $\langle \widebar {K}^\dagger\cdot,\cdot \rangle_\cW\doteq\langle \cdot,\widebar {K}\cdot \rangle_{\widebar {\cV}} $ is a multiple of a normally hyperbolic operator and thus has a well--posed Cauchy problem.
\item There exists a differential operator $T:\Se(\cW)\to\Se(\cV)$, e.g. $T=2\widebar {K}$, such that $\widetilde{P}:\Se(\cV)\to\Se(\cV)$, $\widetilde{P}\doteq {\widebar {P}} + T\circ \widebar {K}^\dagger$ \eqref{eq_PTilde} and $Q:\Se(\cW)\to\Se(\cW)$, $Q\doteq \widebar {K}^\dagger \circ T$ are normally hyperbolic and have a well--posed Cauchy problem.
\end{enumerate}
\end{proposition}

We may now follow the construction reviewed in Section \ref{sec:linear_gauge} in order to construct a pre--symplectic space of gauge--invariant observables which may then be canonically quantized as outlined in Section \ref{sec:linear_quant}. To this end, we consider the spaces
$$
\wSol\doteq\{\Theta\in\Se(\cV)\,|\,\widebar P \Theta = 0\}\,,\qquad 
\wSolsc  \doteq\wSol\cap \Sesc(\cV)\,,
$$
$$
\wcG\doteq\widebar K\left[\Se(\cW)\right]\,,\qquad \wcGsc\doteq \wcG\cap \Sesc(\cV)\,,\qquad \wcGscc\doteq  \widebar K\left[\Sesc(\cW)\right]\,.
$$
The space of gauge--equivalence classes of solutions $\widebar\Sol/\widebar \cG$ may be interpreted as the space of pure states in the classical linearised Einstein--Klein--Gordon field theory. The space of (regular) linear functionals on $\widebar\Sol/\widebar \cG$ is parametrised by the labelling space
$$
\widebar\cE\doteq\kerc\left(\widebar K^\dagger \right)\left/\widebar P\left[\Sec(\cV)\right]\right.\,,\qquad \kerc\left(\widebar K^\dagger \right)\doteq\{h\in\Sec(\cV)\,|\,\widebar K^\dagger h=0\}.
$$
Indeed, the dual pairing 
$$
\widebar{\Sol}/\widebar{\cG}\times \widebar{\cE}\ni \left([\Theta],[f]\right)\mapsto \left\langle[\Theta],[f]\right\rangle\doteq\left\langle\Theta,f\right\rangle_{\widebar{\cV}}
$$
is by construction independent of the representatives and thus well--defined. $\widebar\cE$ thus conveniently parametrises linear and gauge--invariant local on--shell observables in the classical field theory. In order to endow $\widebar\cE$ with a pre--symplectic structure, we consider the causal propagator $E^{\widetilde P}$ of the normally hyperbolic gauge--fixed equation of motion operator $\widetilde P=\widebar P + T\circ \widebar K^\dagger$ and define a bilinear form $\widebar\tau$ on $\widebar\cE$ by
$$
\widebar\tau([f_1],[f_2])\doteq\left\langle f_1,E^{\widetilde P}f_2\right\rangle_{\widebar \cV}\,.
$$
Proposition \ref{prop_prop_G} and Proposition \ref{prop_presymplectic_gauge} then imply:

\begin{enumerate}
\item $\widebar\tau$ is independent of the representatives and thus well--defined.
\item $\widebar\tau$ is antisymmetric.
\item $\widebar\tau$ is independent of the gauge--fixing operator $T$ appearing in the definition of $\widetilde P$.
\item $(\widebar \cE,\widebar \tau)$ is a well--defined pre--symplectic space.
\end{enumerate}

An `equal--time version' of the covariant pre--symplectic space $(\widebar \cE,\widebar \tau)$ can be constructed by setting 
$$
\ip{\Theta_1,\Theta_2}_{\wSol}\doteq \ip{\widebar P\Theta_1^+,\Theta_2}_{\widebar \cV}\,,
$$
where $\Theta_1^+$ is the future part of $\Theta_1$. Theorem \ref{thm_proptau2} then implies that $(\wSolsc/\wcGscc,\ip{\cdot,\cdot}_\wSol)$ is a well--defined pre--symplectic space and that the causal propagator of $\widetilde P$  descends to an isomorphism 
$$E^{\widetilde P}:\left(\widebar \cE,\widebar \tau\right)\to\left(\wSolsc/\wcGscc,\ip{\cdot,\cdot}_\wSol\right)\,.$$
Moreover, given the current $\widebar {j}$
\begin{equation}
\label{eq_current}\widebar {j}:\Se(\cV)\times \Se(\cV)\to  T^*M
\end{equation}
\begin{align*}(\Theta_1, \Theta_2) \mapsto \widebar {j}_\alpha(\Theta_1,\Theta_2)\doteq &-4{\theta_1}^{\mu\nu}\nabla_\alpha{\theta_2}_{\mu\nu}+8{\theta_1}_{\alpha}^{\phantom{c}\nu}\nabla^\mu{\theta_2}_{\mu\nu}+2{\theta_1}_{\mu}^{\phantom{c}\mu}\nabla_\alpha{\theta_2}_{\nu}^{\phantom{c}\nu}-\\
&-\phi_1\nabla_\alpha\phi_2+4\phi_1 \left(\nabla^\mu\varphi\right){\theta_2}_{\alpha\mu}-\text{`}1\leftrightarrow 2\text{'}\end{align*}
whose covariant divergence satisfies
\begin{equation}\label{eq_currentconservation}
 \nabla^\alpha\widebar {j}_\alpha(\Theta_1,\Theta_2)=\left\langle\left\langle \Theta_1,{\widebar {P}}\Theta_2\right\rangle\right\rangle_{\widebar {\cV}} -\left\langle\left\langle {\widebar {P}}\Theta_1,\Theta_2\right\rangle\right\rangle_{\widebar {\cV}}\,, 
\end{equation}
where $\ipp{\cdot,\cdot}_{\widebar {\cV}}$ denotes the integrand of $\ip{\cdot,\cdot}_{\widebar {\cV}}$, Theorem \ref{thm_proptau2} further implies that
$$
\ip{\Theta_1,\Theta_2}_{\wSol}=\int\limits_\Sigma d\Sigma\; N^\alpha \widebar j_\alpha(\Theta_1,\Theta_2)
$$
for an arbitrary but fixed Cauchy surface $\Sigma$ with future pointing unit normal vector $N$.

The field redefinition $\widebar\cdot$ was helpful for uncovering important structural properties of the linearised Einstein--Klein--Gordon system $(M_\varphi, \cV, \cW, \widebar P, \widebar K)$, and to show that it fits into the axioms listed in Definition \ref{def_lineargauge}. The original system $(M_\varphi, \cV, \cW, P,  K)$ does not directly satisfy these axioms because e.g. the operator $K^\dagger \circ K$ does not have a well--posed Cauchy problem. Yet, the field redefinition is merely a computational trick and not of physical significance. To see this explicitly, we consider the spaces related to the original system $(M_\varphi, \cV, \cW, P,  K)$:
$$
\Sol\doteq\{\Gamma\in\Se(\cV)\,|\, P \Gamma = 0\}\,,\qquad 
\Solsc  \doteq\Sol\cap \Sesc(\cV)\,,
$$
$$
\cG\doteq K\left[\Se(\cW)\right]\,,\qquad \cGsc\doteq \cG\cap \Sesc(\cV)\,,\qquad \cGscc\doteq   K\left[\Sesc(\cW)\right]\,.
$$
$$
\cE\doteq\kerc\left(K^\dagger \right)\left/P\left[\Sec(\cV)\right]\right.\,,\qquad \kerc\left( K^\dagger \right)\doteq\{h\in\Sec(\cV)\,|\, K^\dagger h=0\}.
$$
It is not difficult to see that $\widebar\cE=\cE$ and as before, we may now observe that $\cE$ parametrises  linear functionals on $\Sol/\cG$. Due to
$$
\langle \Gamma,h\rangle_\cV=\langle \widebar\Gamma,h\rangle_{\widebar\cV}
$$
the physical interpretation of these linear functionals in terms of local observables is manifestly independent of the field redefinition. Finally, it follows from the previous discussion that 
$$
\widetilde P\circ \widebar{\,\cdot\,}=P+T\circ K^\dagger\circ \widebar{\,\cdot\,}
$$
is not normally hyperbolic, but still Green--hyperbolic. The causal propagator of $\widetilde P\circ \widebar{\,\cdot\,}$ is $\widebar{\,\cdot\,}^{\,-1}\circ E^{\widetilde P}$ and the induced bilinear form on $\cE$
$$
\tau([h_1],[h_2])\doteq\left\langle h_1,\left[\widebar{\,\cdot\,}^{\,-1}\circ E_{\widetilde P}\right]h_2\right\rangle_\cV
$$
manifestly equals $\widebar\tau$. Similarly, $(\Solsc/\cGscc,\ip{\cdot,\cdot}_\Sol)$ with 
$$
\ip{\Theta_1,\Theta_2}_{\Sol}\doteq \ip{P\Theta_1^+,\Theta_2}_{\cV}\,,
$$
is a pre--symplectic space which is isomorphic to $(\wSolsc/\wcGscc,\ip{\cdot,\cdot}_\wSol)$ via the field redefinition map. Consequently, we have the isomorphisms
$$
\left(\wSolsc/\wcGscc,\ip{\cdot,\cdot}_\wSol\right)\simeq \left(\widebar \cE,\widebar \tau\right)=\left(\cE,\tau\right)\simeq \left(\Solsc/\cGscc,\ip{\cdot,\cdot}_\Sol\right)\,.
$$
These considerations in classical field theory clearly carry over to the quantum theory constructed as in Section \ref{sec:linear_quant} and the algebras of quantum observables induced by the pre--symplectic spaces $\left(\cE,\tau\right)$ and $\left(\widebar\cE,\widebar\tau\right)$ coincide.

For the analysis of perturbations in Inflation it is useful to observe that, for an arbitrary but fixed Cauchy surface $\Sigma$ with future pointing unit normal vector $N$,
$$
\ip{\Gamma_1,\Gamma_2}_{\Sol}=\int\limits_\Sigma d\Sigma\; N^\alpha j_\alpha(\Gamma_1, \Gamma_2)
$$
with the current $j$
\begin{equation}
\label{eq_currentorig}j:\Se(\cV)\times \Se(\cV)\to  T^*M
\end{equation}
\begin{align*}j_\alpha(\Gamma_1, \Gamma_2)\doteq &-\frac{1}{4}{\gamma_1}^{\mu\nu}\nabla_\alpha {\gamma_2}_{\mu\nu}+\frac{1}{2}{\gamma_1}_{\alpha}^{\phantom{c}\nu}\nabla^\mu {\gamma_2}_{\mu\nu}+\frac{1}{4}{\gamma_1}_{\nu}^{\phantom{c}\nu}\nabla_\alpha {\gamma_2}_{\mu}^{\phantom{c}\mu}-\frac{1}{4}{\gamma_1}_{\nu}^{\phantom{c}\nu}\nabla_\mu{\gamma_2}_{\alpha}^{\phantom{c}\mu}\notag\\
&-\frac{1}{4}{\gamma_1}_{\alpha\mu}\nabla^\mu {\gamma_2}_{\nu}^{\phantom{c}\nu}-\phi_1\nabla_\alpha \phi_2+\frac{1}{2}{\gamma_1}_{\nu}^{\phantom{c}\nu}\left(\nabla_\alpha \varphi\right)\phi_2+\phi_1 \left(\nabla^\mu \varphi\right){\gamma_2}_{\alpha\mu}\\
&-\text{`}1\leftrightarrow 2\text{'}\end{align*}
whose covariant divergence satisfies
\begin{equation}\label{eq_currentconservationorig}
 \nabla^\alpha{j}_\alpha(\Gamma_1,\Gamma_2)=\left\langle\left\langle \Gamma_1,{P}\Gamma_2\right\rangle\right\rangle_{ {\cV}} -\left\langle\left\langle {P}\Gamma_1,\Gamma_2\right\rangle\right\rangle_{ {\cV}}\,, 
\end{equation}
where $\ipp{\cdot,\cdot}_{ {\cV}}$ denotes the integrand of $\ip{\cdot,\cdot}_{ {\cV}}$. This follows from the identity $j(\Gamma_1,\Gamma_2)=\widebar j(\widebar\Gamma_1,\widebar\Gamma_2)$.

\subsection{The Standard Approach to Perturbations in Inflation and Comparison of Approaches}
\label{sec:inflation_standard}

After discussing the gauge--invariant quantization of the linearised Einstein--Klein--Gordon system on general on--shell backgrounds $M_\varphi=(M,g,\varphi)$, we consider the special case of FLRW--type on--shell backgrounds with flat spatial sections, which is the field--theoretic model of perturbations in Inflation. We recall that these backgrounds are characterised by
$$
g=a(\eta)^2(-d\eta^2+d{\vec{x}}^2)\,,\qquad \varphi(\eta,{\vec{x}})=\varphi(\eta)\,.
$$
Two reoccurring important quantities are 
$$
\Hc\doteq\frac{a^\prime}{a}=aH\,, \qquad z\doteq\frac{a\varphi^\prime}{\Hc}
$$
and we shall indicate spatial indices in the homogeneous and isotropic (i.e. comoving) FLRW--coordinates by Latin letters $i, j, k, \ldots$. We shall use the convention that these indices will be raised and lowered by means of the Euclidean metric $\delta_{ij}$ rather than by means of the induced metric $a(\eta)^2\delta_{ij}$. For the discussion of perturbations in Inflation it is convenient to work with the original formulation $(M_\varphi, \cV, \cW, P,  K)$ of the linearised Einstein--Klein--Gordon system rather than its redefined version $(M_\varphi, \cV, \cW, \widebar P,  \widebar K)$.

In the following discussion, a special role is played by sections which vanish at spatial infinity (with all derivatives).
$$
\Seinf(\cV)\doteq\{\Gamma\in\Se(\cV)\,|\,\partial_{i_1}\cdots\partial_{i_n}\Gamma(\eta,{\vec{x}}) \text{ vanishes for $|{\vec{x}}|\to\infty$ for all $n\in\NN_0$}\}
$$
$$
\Seinf(\cW)\doteq\{\varsigma\in\Se(\cW)\,|\,\partial_{i_1}\cdots\partial_{i_n}\varsigma(\eta,{\vec{x}}) \text{ vanishes for $|{\vec{x}}|\to\infty$ for all $n\in\NN_0$}\}
$$
Namely, one can uniquely decompose $\Gamma=(\gamma_{\alpha\beta}, \phi)^t\in \Seinf(\cV)$ as
$$\gamma_{\alpha\beta}=a(\eta)^2\left(\begin{array}{cc}
2 A & \left(\partial_i B-V_i\right)^t\\
\;\partial_i B-V_i\; & -2\left(\partial_i\partial_j C
+\delta_{ij}D+\partial_{(i} W_{j)}+T_{ij}\right)
\end{array}\right)
$$
$$
A,B,C,D\in C_\infty^\infty(M,\RR)\,,\quad V, W\in  C_\infty^\infty(M,\RR^3)\,,\quad \partial^i V_i=\partial^i W_i=0
$$
$$
T\in C_\infty^\infty(M,{\textstyle\bigvee^2}\RR^3)\,,\quad T_i^i=0\,,\quad \partial^i T_{ij}=0\,.
$$
The components $B, C, D, W_i$ are solutions of certain Poisson equations, e.g. 
$$
\Delta B = \frac{\partial^i\gamma_{0i}}{a^2}\,,\qquad \Delta = \partial_i \partial^i
$$
\index{Delta@$\Delta$, Laplace operator on $\bR^3$}and the uniqueness of the above decomposition results from the unique solvability of such equations under the assumption that the solutions vanish at infinity. For the same reason, such a decomposition for general $\Gamma\in\Se(\cV)$ can only be unique up to harmonic functions. In fact, this non--uniqueness is a non--trivial obstacle for proving that a decomposition which is smooth in $\eta$ exists in general, though we presume that this is the case. Notwithstanding, we shall only need the existence and uniqueness of the decomposition for $\Seinf(\cV)$ in the following. 

Owing to their transformation properties under the Euclidean group ${\mathbb E}^3$, the components $A,B,C,D,\phi$ of a section in $\Gamma\in\Seinf(\cV)$ are called `scalar', $V_i, W_j$ are called `vector' and $T_{ij}$ are called `tensor' and similarly for the components of $\varsigma\in\Seinf(\cW)$. Following this nomenclature, we say that $\Gamma = (\gamma_{\alpha\beta},\phi)^t\in \Se(\cV)$ ...
\begin{itemize}
\item[...] is of scalar type if $\gamma_{\alpha\beta}$ can be decomposed as above with $V_i=W_i=T_{ij}=0$.
\item[...] is of vector type if $\phi=0$ and $\gamma_{\alpha\beta}$ can be decomposed as above with $A=B=C=D=T_{ij}=0$.
\item[...] is of tensor type if $\phi=0$ and $\gamma_{\alpha\beta}$ can be decomposed as above with $A=B=C=D=V_i=W_i=0$.
\end{itemize}
Based on this we define the following section spaces.
$$
\Se^{S/V/T}(\cV)\doteq\{\Gamma\in\Se(\cV)\,|\,\Gamma \text{ is of scalar/vector/tensor type}\}
$$
$$
\Se^{S/V/T}_{\infty/0}(\cV)\doteq \Se_{\infty/0}(\cV)\cap\Se^{S/V/T}(\cV)
$$
The existence and uniqueness of the decomposition of sections vanishing at spatial infinity may be cast in the following form
$$
\Seinf(\cV)=\Seinf^S(\cV)\oplus\Seinf^V(\cV)\oplus\Seinf^T(\cV)
$$
and one may check that the decomposition is orthogonal w.r.t. $\langle\cdot,\cdot\rangle_\cV$ and $\langle\cdot,\cdot\rangle_\cW$. However, the decomposition is non--local in space, and thus one has
$$
\Sec^S(\cV)\oplus\Sec^V(\cV)\oplus\Sec^T(\cV)\subsetneq\Sec(\cV)\,!
$$
Similar decomposition results hold for $\varsigma\in\Seinf(\cW)$ and we extend the above nomenclature to this case in the obvious way. 

By the existence and uniqueness of the decomposition, the individual components induce well--defined functionals on $\Seinf(\cV)$ and $\Seinf(\cW)$, e.g. $A:\Seinf(\cV)\to C^\infty_\infty(M,\RR)$. This existence and uniqueness further implies that there exist projectors 
$$
\cP^{S/V/T}_\cV:\Seinf(\cV)\to\Seinf^{S/V/T}(\cV)\qquad\text{and}\qquad \cP^{S/V}_\cW:\Seinf(\cW)\to\Seinf^{S/V}(\cW)
$$ 
which are formally selfadjoint w.r.t. $\langle\cdot,\cdot\rangle_\cV$ and $\langle\cdot,\cdot\rangle_\cW$. It is not difficult to check that the gauge transformation operator $K$ and the equation of motion operator $P$ commute with these projectors, i.e.
$$
\cP_\cV^{S/V}\circ K|_{\Seinf(\cW)}=K\circ \cP_\cW^{S/V}\,,\qquad \cP_\cV^{T}\circ K|_{\Seinf(\cW)}=0\,,
$$
$$
\cP^{S/V/T}_\cV\circ P|_{\Seinf(\cV)}=P\circ \cP^{S/V/T}_\cV\,.
$$
Thus the equations of motions and gauge transformations decouple for sections which vanish at spatial infinity and we may consider subspaces 
$$
\Sol^{S/V/T}_\infty\subset \Sol_\infty \qquad\text{and}\qquad\cG^{S/V}_\infty\subset \cG_\infty
$$
which are defined in the obvious way.

The decomposed equations of motion may be expressed in terms of gauge--invariant linear combinations of the decomposition components, i.e. in terms of functionals on $\Sol_\infty/\cG_\infty$.
$$
\Psi\doteq A-(\partial_\eta + \Hc)(B+C^\prime)\qquad\Phi\doteq D-\Hc(B+C^\prime)
$$
$$
\chi\doteq \phi - \varphi^\prime(B+E^\prime)\qquad \mu\doteq-\frac{z}{a}\Phi+\chi
$$
$$
X_i\doteq W_i^\prime-V_i\qquad T_{ij}
$$
$\Psi$ and $\Phi$ are the so-called {\em Bardeen potentials}\index{Bardeen potentials} whereas $\mu$ is the {\em Mukhanov--Sasaki variable}\index{Mukhanov--Sasaki variable} and is of particular physical significance because it is related to the perturbation of the scalar curvature of the spatial slices. In terms of these gauge--invariant quantities, the equations of motion $P\Gamma=0$ for $\Gamma\in\Seinf(\cV)$ read (see \cite{chap3_ThomasInflation} for details):

\begin{itemize}
\item[] scalar:
$$
P^\mu\mu\doteq\left(\nabla_\mu\nabla^\mu + \frac{R}{6}-\frac{z^{\prime\prime}}{z a^2}\right)\mu=0 
$$
$$
\Psi=\frac{\cH}{2a^2 \alpha}\left(\int\limits^\eta_{\eta_0}d\tilde\eta\; a z\mu+\lambda_0\right),\qquad\Phi=-\Psi\,,\qquad \chi=\frac{2}{ \varphi^\prime}(\partial_\eta+\cH)\Psi\,,
$$
where $\eta_0$ is arbitrary and $\lambda_0$ is the unique solution of
$$
\left.\Delta\lambda_0=az\left(\mu^\prime+\left(\frac{\cH^\prime}{\cH}-\frac{\varphi^{\prime\prime}}{\varphi^\prime}\right)\mu\right)\right\vert_{\eta=\eta_0}.
$$
\item[] vector:
$$
\Delta X_i=0\,,\qquad (\partial_\eta+2\Hc)X_i=0
$$
\item[] tensor:
$$
P^TT_{ij}\doteq\frac{1}{a^2}\left((\partial_\eta+2\Hc)\partial_\eta- \Delta \right)T_{ij}=0
$$
\end{itemize}

We see that $\mu$ is a conformally coupled scalar field with a particular time--dependent mass and that $\Psi$, $\Phi$ and $\chi$ are non--local functionals of $\mu$, whereas no non--trivial vector solutions vanishing at spatial infinity exist. Moreover, the tensor components $T_{ij}$ satisfy a normally hyperbolic equation. The last statement may have been deduced directly from the previous analysis by observing that the field redefinition $\widebar{\,\cdot\,}$ acts trivially on $\Seinf^T(\cV)$ and that $\Seinf^T(\cV)$ lies in the kernel of $K^\dagger$. Thus $P$ restricted to $\Seinf^T(\cV)$ coincides with the normally hyperbolic $\widetilde P=P\circ \widebar{\,\cdot\,}^{\,-1}+T\circ K^\dagger$.

The standard treatment of the quantum theory of perturbations in Inflation may be rephrased in the present context as follows. One considers the symplectic spaces $(\cE^\mu,\tau^\mu)$ 
$$
\cE^\mu\doteq C^\infty_0(M,\RR)/P^\mu\left[C^\infty_0(M,\RR)\right]
$$
$$
\cE^\mu\times\cE^\mu\ni ([f_1],[f_2])\mapsto \tau^\mu([f_1],[f_2])\doteq\langle f_1, E^\mu f_2\rangle
$$
$$
E^\mu \text{ causal propagator of $P^\mu$}
$$
and $(\cE^{TT},\tau^T)$ 
$$
\cE^{TT} \doteq C^\infty_0(M,T)/P^T\left[C^\infty_0(M,T)\right]
$$
$$
C^\infty(M,T)\doteq \{T\in C^\infty(M,{\textstyle\bigvee^2}\RR^3)\,|\, T_i^i=0,  \partial^i T_{ij}=0\}
$$
$$
\cE^{TT}\times\cE^{TT}\ni ([f_1],[f_2])\mapsto \tau^T([f_1],[f_2])\doteq\langle f_1, E^T f_2\rangle_T
$$
$$
E^T \text{ causal propagator of $P^T$}
$$
$$
 C^\infty(M,T)^2\ni(f_1,f_2)\mapsto \langle f_1,f_2\rangle_T\doteq\int\limits_M \vol \delta^{ik}\delta^{jl}f_{1,ij} f_{2,kl}\,.
$$
$P^\mu$ and $P^T$ are both formally self--adjoint and one may show that $\langle \cdot,\cdot\rangle_T$ is non--degenerate on $C_\infty^\infty(M,T)$ \cite{chap3_ThomasInflation}. Thus $\tau^\mu$ and $\tau^T$ are antisymmetric and non--degenerate by standard results reviewed in Section \ref{sec_generalnongauge}.

In the standard treatment of perturbations in Inflation one effectively assumes that all gauge--invariant (polynomial) local observables in the quantum theory are spanned by the local observables obtained from the canonical quantization of $(\cE^\mu,\tau^\mu)$ and $(\cE^{TT},\tau^T)$. On the other hand, one may take the point of view that the construction outlined in the previous discussion of the linearised Einstein--Klein--Gordon system, i.e. the canonical quantization of the pre--symplectic space $(\cE,\tau)$ should yield all (polynomial) local gauge--invariant observables. In order to compare these two approaches, we define
$$
\kerc^{S/V/T}\left(K^\dagger\right)\doteq\kerc\left(K^\dagger\right)\cap \Se^{S/V/T}(\cV)
$$
$$
\cE^{S/V/T}\doteq\kerc^{S/V/T}\left(K^\dagger\right)\left/P\left[\Sec^{S/V/T}(\cV)\right]\right.\,.
$$ 
One may then prove the following results \cite{chap3_ThomasInflation}.
\begin{theorem}
The following relations hold.
\begin{enumerate}
\item $\cE^{S/V/T}\subset \cE$ .
\item $\cE^V=\{0\}$.
\item $\left(\cE^S,\tau\right)$ and $\left(\cE^\mu,\tau^\mu\right)$ are isomorphic.
\item $\left(\cE^T,\tau\right)$ and $\left(\cE^{TT},\tau^T\right)$ are isomorphic. 
\item $\cE^S\oplus\cE^T\subsetneq \cE$.
\item $\cE^S\oplus\cE^T$ is separating on $\Sol_\infty/\cG_\infty=\Sol^S_\infty/\cG^S_\infty\oplus \Sol^T_\infty/\cG^T_\infty$.
\item $\tau$ is non--degenerate on $\cE$.
\end{enumerate}
\end{theorem}

The proofs of most of these statements are unfortunately quite cumbersome. In particular, one needs to analyse how the solution theory of hyperbolic operators and the results of Theorem \ref{thm_proptau2} intertwine with the decomposition of perturbations into scalar, vector and tensor contributions. Once this is understood, the third and fourth statement e.g. follow essentially from the fact that the equal--time versions of the compared symplectic spaces are isomorphic because the corresponding symplectic forms are given by Cauchy surface--integrals of (normal components of) `conserved currents' which can be seen to match by explicit computations.

 Even if the proof of the above theorem is rather lengthy, the statements have a straightforward physical interpretation, which is valid both for the classical and the quantum theory of perturbations in Inflation as should be clear from the discussion of the quantization of pre--symplectic spaces in Section \ref{sec:linear_quant}. The first statement implies that there are local observables which  may be meaningfully classified as `scalar' and `tensor'. This may seem surprising in view of the fact that the decomposition of configurations in scalar/vector/tensor components is a priori non--local. However the second statement entails that there are indeed no non--trivial local `vector' observables. From the third and fourth statement one can infer that the standard treatment of perturbations in Inflation captures the same local scalar and tensor observables that one obtains from the general gauge--invariant quantization of the linearised Einstein--Klein--Gordon system. An interesting result found in \cite{chap3_Eltzner:2013soa} implies that the scalar field $\mu$ is in fact the unique field with a normally hyperbolic equation of motion whose associated symplectic space is equivalent to $\left(\cE^S,\tau\right)$.

However, statement 5. implies that not all local observables of the linearised Einstein--Klein--Gordon system are spanned by local observables of scalar and tensor type. In this sense, the standard approach to the quantization of perturbations in inflation `misses' some local observables. However, the sixth statement entails that the observables captured in the standard approach are still sufficient to measure configurations of the perturbations which vanish at spatial infinity. These configurations are considered to be `small' in a certain sense. Presumably this statement can be generalised by proving that local observables of scalar and tensor type separate quantum states whose correlation functions vanish at spatial infinity in each argument.

Finally, the last statement is somewhat independent of the others and may be interpreted such as to say that the quantum theory of the linearised Einstein--Klein--Gordon system on FLRW backgrounds does not contain non--trivial classical observables.






\backmatter
\printindex

\end{document}